\documentclass[11pt,letterpaper]{article}
\usepackage[margin=1in]{geometry}

\usepackage[square,sort,comma,numbers]{natbib}
\usepackage{graphicx}
\usepackage{cancel}
\usepackage{amsmath,amsfonts,amssymb,bbm,bm}
\usepackage{xspace}

\usepackage{amsthm}
\usepackage{booktabs} 
\usepackage{tikz-cd}
\usepackage{xcolor}
\usepackage{nameref}
\definecolor{ForestGreen}{rgb}{0.1333,0.5451,0.1333}
\definecolor{DarkRed}{rgb}{0.65,0,0}
\definecolor{Red}{rgb}{1,0,0}
 \usepackage[linktocpage=true,
 pagebackref=false,colorlinks,
linkcolor=DarkRed,citecolor=ForestGreen,bookmarks,bookmarksopen,bookmarksnumbered]
 {hyperref}
\RequirePackage[T1]{fontenc} 
\RequirePackage[type1,tt=false]{libertine} 
\RequirePackage{textcomp} 
\RequirePackage[varqu,varl]{inconsolata}
\RequirePackage[libertine,vvarbb]{newtxmath} 
\RequirePackage[scr=rsfso]{mathalpha} 
\RequirePackage{bm}

\usepackage[linesnumbered,ruled,vlined]{algorithm2e}
\SetKwInput{KwAssume}{Assumptions}
\SetKwInput{KwInput}{Input}                
\SetKwInput{KwOutput}{Output} 

\newtheorem{theorem}{Theorem}[section]
\newtheorem{lemma}[theorem]{Lemma}
\newtheorem{proposition}[theorem]{Proposition}

\theoremstyle{definition}
\newtheorem{definition}[theorem]{Definition}

\newcommand{\exan}{\textit{ex-ante}\xspace}
\newcommand{\expo}{\textit{ex-post}\xspace}
\newcommand{\calA}{\mathcal{A}}
\newcommand{\calM}{\mathcal{M}}
\newcommand{\calF}{\mathcal{F}}
\newcommand{\calD}{\mathcal{D}}
\newcommand{\bX}{\mathbf{X}}
\newcommand{\truc}{{\tt truc}}

\DeclareMathOperator*{\argmin}{argmin}
\DeclareMathOperator*{\argmax}{argmax}

\title{Truthful and Almost Envy-Free Mechanism of Allocating Indivisible Goods: the Power of Randomness}

\author{Xiaolin Bu\\
lin\_bu@sjtu.edu.cn\\
Shanghai Jiao Tong University
\and
Biaoshuai Tao\\
bstao@sjtu.edu.cn\\
Shanghai Jiao Tong University}

\date{}

\begin{document}

\maketitle

\begin{abstract}
We study the problem of \emph{fairly} and \emph{truthfully} allocating $m$ indivisible items to $n$ agents with additive preferences.
Specifically, we consider truthful mechanisms outputting allocations that satisfy EF$^{+u}_{-v}$, where, in an EF$^{+u}_{-v}$ allocation, for any pair of agents $i$ and $j$, agent $i$ will not envy agent $j$ if $u$ items were added to $i$'s bundle and $v$ items were removed from $j$'s bundle.
Previous work easily indicates that, when restricted to deterministic mechanisms, truthfulness will lead to a poor guarantee of fairness: even with two agents, for any $u$ and $v$, EF$^{+u}_{-v}$ cannot be guaranteed by truthful mechanisms when the number of items is large enough.
In this work, we focus on randomized mechanisms, where we consider $\exan$ truthfulness and $\expo$ fairness.
For two agents, we present a truthful mechanism that achieves EF$^{+0}_{-1}$ (i.e., the well-studied fairness notion EF$1$). 
For three agents, we present a truthful mechanism that achieves EF$^{+1}_{-1}$.
For $n$ agents in general, we show that there exists a truthful mechanism that achieves EF$^{+0}_{-O(\sqrt{n})}$.
On the negative side, when considering the stronger notion EF$_{-v}^{+u}$X, we show that it cannot be achieved by any randomized truthful mechanism for any $u, v$, and any fixed number of agents.

We further consider fair and truthful mechanisms that also satisfy the standard efficiency guarantee: Pareto-optimality.
We provide a mechanism that simultaneously achieves truthfulness, EF$1$, and Pareto-optimality for bi-valued utilities (where agents' valuation on each item is either $p$ or $q$ for some $p>q\geq0$).
For tri-valued utilities (where agents' valuations on each item belong to $\{p,q,r\}$  for some $p>q>r\geq0$) and any $u,v$, we show that truthfulness is incompatible with EF$^{+u}_{-v}$ and Pareto-optimality even for two agents.
\end{abstract}

\section{Introduction}
\emph{Fair division} studies how to allocate a set of resources to a set of agents with heterogeneous preferences.
Starting from~\citet{Steinhaus48,Steinhaus49}, the fair division problem has been extensively studied by economists, mathematicians, and computer scientists.
Multiple textbooks and survey papers have been published on this topic~\cite{amanatidis2022fair,amanatidis2023fair,BramsTa96,BrandtCoEn16,procaccia2013cake,liu2023mixed,robertson1998cake}.
In this paper, we study the fair division problem when resources are \emph{indivisible} items.
Specifically, we aim to fairly allocate $m$ items to $n$ agents, where each agent has her own valuations on those $m$ items.
Unless otherwise stated, agents' valuations on the items are additive, i.e., each agent's valuation on a set of items is the sum of her values for all the individual items.

Among various fairness criteria, \emph{envy-freeness}~\cite{Foley67,varian1973equity} is arguably the most studied notion, which says that, for any pair of agents $i$ and $j$, agent $i$ should value her own allocated share weakly more than agent $j$'s, i.e., agent $i$ does not envy agent $j$.
However, when indivisible items are concerned, envy-free allocation may not exist (e.g., all the agents value the items equally, but $m$ is not a multiple of $n$).
It is then natural to define relaxations of envy-freeness that are tractable.
The most popular line of research considers envy-freeness up to the addition or/and removal of a small number of items.
In particular, an allocation is ``almost envy-free'' if, for each pair of agents $i$ and $j$, agent $i$ will no longer envy agent $j$ if a small number of items is (hypothetically) added to agent $i$'s allocated bundle and a small number of items is (hypothetically) removed from agent $j$'s allocated bundle.
Among this type of relaxation, \emph{envy-freeness up to one item} (EF1) receives the most significant attention.
It is well-known that an EF1 allocation always exists even for general monotone valuations that need not be additive, and it can be computed efficiently~\cite{budish2011combinatorial,lipton2004approximately}.

When deploying a fair division algorithm in practice, agents may not honestly report their valuation preferences to the algorithm if they can benefit from strategic behaviors.
This motivates the study of the fair division problem from the mechanism design point of view.
Other than guaranteeing fairness, we would also like an algorithm, or a mechanism, to be \emph{truthful}, where truth-telling is each agent's dominant strategy.
Unfortunately, it is known that truthfulness is incompatible with most of the meaningful fairness notions for deterministic mechanisms~\cite{lipton2004approximately,amanatidis2017truthful,amanatidis2016truthful,caragiannis2009low,dobzinski2023fairness}, including those above-mentioned variants of envy-freeness~\cite{amanatidis2017truthful}.
In particular, \citet{amanatidis2017truthful} give a characterization of truthful mechanisms with two agents.
Their observation implies that no truthful mechanism can achieve envy-freeness even up to adding/removing an arbitrary number of items (see Theorem~\ref{thm:impossibility_deterministic}).
Truthfulness and (almost) envy-freeness are compatible only for very restrictive valuation functions~\cite{halpern2020fair,babaioff2021fair,barman2022truthful,christodoulou2024fair}.
Further, it is shown that under mild additional assumptions, the only deterministic truthful mechanism is serial/sequential-quota dictatorship~\cite{papai2000strategyproof,papai2001strategyproof,bouveret2023thou,babaioff2024truthful}, where each agent is asked to take a predefined number of items one by one\footnote{In serial dictatorship, the order of the agents is predefined and independent of the agents' preferences.
In sequential dictatorship, the order of the agents depends on the agents' preferences.}.
Such mechanisms clearly lack fairness guarantees.

As the incompatibility of truthfulness and fairness holds only under deterministic mechanisms, a natural approach to resolve the problem is to apply \emph{randomness} in mechanisms.
We aim to design \emph{randomized mechanisms} that are \emph{truthful in expectation}---truth-telling maximizes each agent's \emph{expected utility}, meanwhile guaranteeing that every allocation possibly output by the mechanism is almost envy-free.

\paragraph{The Power of Randomness.}
The use of randomness to achieve truthfulness has been proven successful in a variety of problems~\cite{mossel2010truthful,dobzinski2013power,AzizFrSh23,babaioff2022best,sun2024randomized}.
For fair division of \emph{divisible} resources, also known as \emph{the cake cutting problem}, truthfulness is incompatible with the \emph{share-based} fairness notion, proportionality, under deterministic mechanisms even for two agents and piecewise constant valuations~\cite{tao2022existence}, where proportionality requires each agent to receive no less than her proportional value ($1/n$ of the total value).
However, when randomness is introduced, \citet{mossel2010truthful} shows the existence of a truthful mechanism such that each agent receives exactly her proportional value in each output allocation, satisfying both proportionality and envy-freeness.
For indivisible goods, randomness allows us to obtain a mechanism that is truthful in expectation, with each output allocation satisfying proportionality up to one item (PROP$1$), where each agent achieves her proportional value after receiving one more item~\cite{AzizFrSh23,babaioff2022best}.

Despite the above progress, our understanding of the power of randomness for fair division of indivisible goods is still limited, especially for \emph{envy-based} fairness notions.
In this work, we offer new insights into the power of randomness by showing how it can bypass the above incompatibility results of deterministic mechanisms for envy-based fairness.

\paragraph{Technical Challenges.}
When designing randomized mechanisms, it is often convenient to view the indivisible items as being \emph{divisible}, and a fractional allocation is then viewed as a lottery of ``pure'' allocations of indivisible items.
In particular, the fraction of an item allocated to an agent corresponds to the probability that the item is allocated to the agent.
In this view, a randomized mechanism consists of \emph{two} parts: a \emph{fractional allocation rule} $\calF$ that decides a fractional allocation where items are viewed as divisible, and a \emph{decomposition rule} $\calD$ that decides the lottery of integral allocations.
Further, for $n$ agents and $m$ items, if we represent a fractional allocation by an $mn$-dimensional vector $\{x_{ig}\}_{i=1,\ldots,n;g=1,\ldots,m}$ where $x_{ig}$ is the fraction of item $g$ allocated to agent $i$ (where we have $\sum_{i=1}^nx_{ig}=1$ as each item is fully allocated), those fair integral allocations are then represented by some integral points in $\mathbb{R}^{mn}$.

To successfully design a truthful and fair randomized mechanism, we need $\calF$ to be truthful, and we also need the fractional allocation output by $\calF$ to be within the convex hull of fair integral allocations.
It is technically challenging to ensure both.

When considering $\calF$, many fractional allocation rules are known to provide fractional allocations that can be decomposed into integral allocations that satisfy fairness notions such as almost envy-freeness, yet fail to ensure truthfulness.
For example, the \emph{probabilistic serial rule}~\cite{bogomolnaia2001new}, which lets agents simultaneously eat the items with a constant speed from items with larger values to items with smaller values, outputs allocations that can be decomposed to EF1 allocations~\cite{AzizFrSh23}.
The \emph{maximum Nash welfare rule}~\cite{nash1950bargaining}, which finds an allocation that maximizes the product of agents' utilities, outputs allocations that can be decomposed to allocations satisfying envy-freeness up to adding and removing one item~\cite{AzizFrSh23}.
However, neither rule is truthful.
The maximum Nash welfare rule fails to be truthful even if agents' valuations on the items are restricted to one of the two values (an example is given in Sect.~\ref{sect:PObi}).
It is also easy to see that the probabilistic serial rule is not truthful: if an agent's most preferred item is not valued by anyone else while a slightly less preferred item is also favored by other agents, it is beneficial for this agent to claim that the latter item is more valuable so that she can compete with other agents for this item at an earlier time.
\citet{babaioff2022best} design algorithms that output allocations which are lotteries of integral allocations that satisfy share-based fairness notions such as approximate \emph{maximin share} (MMS).
\citet{feldman2023breaking} randomize the classical envy-cycle procedure for subadditive valuations.
However, the main focus of these papers is \emph{the best-of-both-worlds}: simultaneously achieving $\exan$ fairness (the fairness of the fractional allocation, which can be viewed as fairness in expectation) and $\expo$ fairness (where each realized integral allocation satisfies certain fairness notion).
In fact, none of those above-mentioned fractional allocation rules is truthful. 

A few fractional allocation rules are known to be truthful.
A natural example is the \emph{equal division rule}, where each item is evenly allocated to $n$ agents.
It is clearly truthful, as it ignores agents' valuations.
It is also $\exan$ envy-free: each agent does not envy any other agents \emph{in expectation}.
One may expect that the equal division fractional allocation, being at the ``center'' of the allocation space, has a reasonable chance to be within the convex hull of integral fair allocations.
While this is the case for PROP1 allocations~\cite{AzizFrSh23,babaioff2022best} (see Sect.~\ref{sec:sharebased} for details), we show in Theorem~\ref{thm:notEF1realizable} that the equal division fractional allocation may be outside of the convex hull of EF1 allocations.
Another class of truthful rules is provided by~\citet{freeman2023equivalence}, which are \emph{responsive to agents' valuations}.
However, as we will discuss later, analyzing whether such fractional allocations can be decomposed into fair integral allocations is technically challenging.
In fact, the allocation rules following their framework produce fractional allocations that are close to equal division. 
Given that the equal division allocation may be outside the convex hull of EF1 allocations, allocations close to it may also face the same limitation. 

Not only is designing truthful fractional allocations rule $\calF$ challenging, it is also challenging to analyze if the output fractional allocations are within the convex hull of fair integral allocations, especially for envy-based fairness criteria.
The generic math tools such as matrix decomposition and total unimodularity of matrices, widely used in the existing literature~\cite{budish2013designing,AzizFrSh23,babaioff2022best}, are useful mostly for share-based fairness notions.
Taking proportionality as an example, given a fractional allocation that is $\exan$ proportional, it is possible to find several integral PROP1 allocations that ``surround'' this fractional allocation.
Specifically, these integral PROP1 allocations can be characterized by the vertices of a polytope given by a set of linear constraints, and the coefficient matrix of this set of linear constraints is totally unimodular, which ensures all vertices of this polytope are integral (we refer the readers to \citet{AzizFrSh23} for details).
Similar techniques work for other share-based fairness criteria such as maximin share (MMS)~\cite{babaioff2022best}.
However, this becomes difficult for envy-based fairness criteria: the coefficient matrix for the linear constraints describing a share-based fairness criterion is relatively simple as we only need $n$ constraints (one for each agent); on the other hand, we need $n(n-1)$ constraints for a typical envy-based criterion, and this extra complexity usually makes the coefficient matrix fail to be totally unimodular.
The only known exception is the fractional allocation output by the probabilistic serial rule, where the ``simultaneously eating'' feature enables us to describe the EF1 allocations with only $n$ linear constraints (again, see \citet{AzizFrSh23} for details).
However, as we have mentioned, the probabilistic serial rule is not truthful.
It seems to us, besides those generic math tools, further progress on envy-based fairness notions requires more delicate analyses that are specific for envy-based fairness.

In conclusion, \emph{there is still a large gap between fractional allocation rules that are decomposable to fair integral allocations and those that are truthful.}

\subsection{Our Results}
In this paper, we mainly focus on envy-based fairness notions, and we show that randomized mechanisms provide significantly better fairness guarantees than their deterministic counterpart.

For $n=2$ agents, we provide a simple truthful randomized mechanism based on the equal division rule that outputs EF1 allocations (Sect.~\ref{sect:twoagents}).
We show that the equal division rule fails to guarantee the EF1 fairness property for $n=3$.
For $n=3$ agents, we provide a truthful randomized mechanism that outputs EF$^{+1}_{-1}$ allocations (envy-freeness up to adding and removing one item, also called envy-freeness up to one good more-and-less proposed by~\cite{barman2019proximity}) (Sect.~\ref{sect:threeagents}).
This is achieved by some carefully designed fractional allocation rule, and the decomposition to EF$^{+1}_{-1}$ allocations applies a series of techniques including proper coloring of regular bipartite graphs and rounding of vertex solutions of linear programs.
For general numbers of agents, we show two mechanisms based on the equal division rule: a truthful randomized mechanism that outputs EF$^{+0}_{-O(\sqrt{n})}$ allocations and a truthful randomized mechanism that outputs allocations simultaneously satisfying two share-based fairness notions---PROP1 and $\frac1n$-MMS (Sect.~\ref{sect:nagents}).
On the negative side, when considering the generalization of EFX, EF$_{-v}^{+u}$X, we show that for any $u$ and $v$, EF$_{-v}^{+u}$X cannot be achieved by any randomized truthful mechanism for any fixed number of agents (Sect.~\ref{sect:incom_efkx}).

Finally, we study efficient randomized truthful mechanisms that satisfy Pareto-optimality in addition to fairness (Sect.~\ref{sect:PO}).
We show that the truthful EF1 Pareto-optimal mechanism for binary valuations~\cite{halpern2020fair,babaioff2021fair,barman2022truthful} generalizes to the bi-valued valuations (where an agent's value to an item can only take two values $p$ or $q$) if randomization is allowed.
This is complemented by the impossibility result that, for any $u$ and $v$, EF$^{+u}_{-v}$ is incompatible with Pareto-optimality for randomized truthful mechanisms for two agents with tri-valued valuations.

\subsection{Related Work}
\paragraph{Truthful mechanisms for allocating indivisible resources.}
Deterministic truthful mechanisms are known to be incompatible with most of the meaningful fairness notions~\cite{lipton2004approximately,amanatidis2017truthful,amanatidis2016truthful,caragiannis2009low,dobzinski2023fairness}.
Positive results are shown only for restrictive valuation functions such as binary valuations and matroid-rank valuations~\cite{halpern2020fair,babaioff2021fair,barman2022truthful}, leveled valuations~\cite{christodoulou2024fair}, or have a poor fairness guarantee such as $\frac1{\lfloor m/2\rfloor}$-MMS even for two agents~\cite{amanatidis2016truthful,amanatidis2017truthful}.
For randomized mechanisms, the equal division rule, which is truthful in expectation, is shown to be PROP1-realizable~\cite{AzizFrSh23,babaioff2022best}, or decomposable to allocations such that with high probability, each agent's value difference to each pair of bundles can be bounded~\cite{caragiannis2009low}.

Another line of work provides characterizations that the only deterministic truthful mechanism is the serial/sequential-quota dictatorship under mild additional assumptions~\cite{papai2001strategyproof,papai2000strategyproof,babaioff2024truthful,hosseini2019multiple,bouveret2023thou}, for example, non-bossiness and Pareto-optimality for strict preferences~\cite{papai2001strategyproof}, or non-bossiness, Pareto-optimality, and neutrality for quantity-monotonic preferences~\cite{papai2000strategyproof}.
Babaioff and Manaker Morag~\cite{babaioff2024truthful} obtain the same characterization result for lexicographic valuations while dropping the requirement on Pareto-optimality.

For indivisible chores (items with negative utilities), \citet{sun2024randomized} provide a randomized mechanism that achieves truthfulness, best-of-both-world fairness, and efficiency under restricted additive valuations (where each item $g$ has an instinct value $v_g$, and each agent's value to $g$ is either $0$ or $v_g$).

\paragraph{Truthful mechanisms for allocating divisible resources.}
As mentioned earlier, the design of truthful mechanisms for allocating divisible items has been studied in the previous work~\cite{cole2013mechanism,freeman2023equivalence,shende2023strategy,satterthwaite1981strategy,garg2022efficient,momi2017efficient,li2023truthful}.
Other than divisible items, another different model for divisible resources is the \emph{cake-cutting} model, where a single heterogeneous ``item''---a piece of cake modeled as the interval $[0,1]$---is allocated to agents who may value different parts of the interval differently.
Truthful mechanism design problem has been extensively studied for the cake-cutting problem~\cite{mossel2010truthful,branzei2015dictatorship,chen2013truth,bei2017cake,menon2017deterministic,bei2020truthful,aziz2014cake,maya2012incentive,li2015truthful,asano2020cake,tao2022existence,bu2023existence}.
It is known that, for deterministic mechanisms, truthfulness is incompatible with fairness~\cite{branzei2015dictatorship,tao2022existence}, unless agents' valuations are binary~\cite{chen2013truth,bei2020truthful}.
However, when allowing randomness, we can simultaneously achieve truthfulness and envy-freeness~\cite{mossel2010truthful}.

\paragraph{Truthful mechanisms for house allocation problem.}
The \emph{house allocation} problem is similar to the fair division problem, except that each agent receives exactly one item.
Truthful mechanism design problem has also been addressed for this problem~\cite{shende2023strategy,svensson1999strategy,suzuki2023strategyproof}.
A typical truthful mechanism is the \emph{serial dictatorship} rule that asks agents to take their favorite items one by one in a fixed order.
It is known that, under some mild technical assumptions, the serial dictatorship rule or its randomized version (with random agent orders) is the unique truthful mechanism~\cite{svensson1999strategy}.
However, its counterpart in the fair division problem, the \emph{round-robin} algorithm or its randomized version, is known to be untruthful.

\paragraph{Other aspects on mechanism design and strategic behaviors.}
Some work focuses on designing mechanisms that satisfy weaker truthfulness guarantees.
One line of work studies mechanisms with good \emph{incentive ratios}~\cite{HUANG2024103491,DBLP:journals/corr/abs-2308-08903,tao2023fair,xiao2020algorithms,huang2025incentive}, where the incentive ratio of a mechanism is the maximum possible ratio of the utility an agent can obtain by strategic behaviors over the utility that she gets by truth-telling.
Incentive ratios measure the degree of untruthfulness; in particular, a mechanism with an incentive ratio of $1$ is truthful.
Other than incentive ratios, other more tractable truthful notions have been proposed and studied~\cite{DBLP:conf/nips/0001V22,troyan2020obvious,ortega2022obvious,bu2023existence,brams2006better,hartman2025s}, including non-obvious manipulability, risk-averse truthfulness, maximin strategyproofness, etc.

Instead of designing a truthful mechanism, some work considers untruthful mechanisms and analyzes the outcome allocations that are the results of agents' strategic behaviors~\cite{DBLP:conf/wine/AmanatidisBFLLR21,DBLP:conf/sigecom/AmanatidisBL0R23}.
Notably, \citet{DBLP:conf/wine/AmanatidisBFLLR21,DBLP:conf/sigecom/AmanatidisBL0R23} show that an EF1 allocation is obtained by the round-robin algorithm (which is known to be untruthful) if agents' strategic plays form a Nash equilibrium.

\paragraph{Other fairness notions.}
A variant of envy-freeness that is stronger than EF1 is \emph{envy-freeness up to any item} (EFX)~\cite{CKMP+19}, which says that an agent $i$ will no longer envy an agent $j$ if \emph{any} item is removed from $j$'s bundle.
EFX has also been extensively studied, although mostly in an existential aspect~\cite{hv2025efx,chaudhury2020efx,chaudhury2021little,plaut2020almost,akrami2023efx,caragiannis2019envy,berger2022almost,mahara2023extension,feldman2023breaking}.
For additive valuations, we know that EFX allocations always exist for up to three agents~\cite{chaudhury2020efx}.
For at least four agents, the existence of EFX allocations is an open problem.

Other than those envy-based fairness notions, another line of fairness notions is \emph{share-based}, where a threshold that more or less represents the ``average share'' is defined for each agent, and the allocation is fair if every agent receives a bundle that has value weakly higher than this threshold.
Examples include \emph{proportionality up to one/any item}~\cite{conitzer2017fair,aziz2020polynomial} and \emph{maximin share}~\cite{amanatidis2017approximation,akrami2024breaking,procaccia2014fair,kurokawa2016can}.

\section{Preliminaries}
A set of $m$ items $M=\{1,\ldots,m\}$ is allocated to a set of $n$ agents $N=\{1,\ldots,n\}$.
Each agent $i$ has a \emph{valuation function} $v_i:\{0,1\}^M\to\mathbb{R}_{\geq0}$.
We assume the valuation functions are additive: $v_i(S)=\sum_{g\in S}v_i(\{g\})$ for each $i\in N$.
For simplicity of notation, we use $v_{ig}$ and $v_i(g)$ interchangeably to denote $v_i(\{g\})$.
An \emph{allocation} is a partition $\calA=(A_1,\ldots,A_n)$ of $M$ where $A_i$ is the set of items allocated to agent $i$.
An allocation is \emph{envy-free} if each agent believes the bundle she receives has a weakly higher value than any other agent according to her own valuation function: for every pair of agents $i$ and $j$, we have $v_i(A_i)\geq v_i(A_j)$.
It is clear that an envy-free allocation may not exist.
For example, when $m<n$, there are agents who receive no item at all.
We consider the following relaxation of envy-freeness.

\begin{definition}
    For nonnegative integers $u$ and $v$, an allocation $\calA=(A_1,\ldots,A_n)$ is \emph{envy-free up to adding $u$ items and removing $v$ items}, denoted by EF$_{-v}^{+u}$, if for every pair of agents $i$ and $j$, there exist item sets $S_i$ and $S_j$ satisfying $|S_i|\leq u$ and $|S_j|\leq v$ such that $v_i(A_i\cup S_i)\geq v_i(A_j\setminus S_j)$.
\end{definition}

In words, an allocation is EF$_{-v}^{+u}$ if, for every pair of agents $i$ and $j$, $i$ will not envy $j$ if at most $u$ items were added to agent $i$'s bundle and at most $v$ items were removed from agent $j$'s bundle.
Given an allocation $\calA=(A_1,\ldots,A_n)$, we say that it is \emph{EF$_{-v}^{+u}$ for agent $i$} if the condition in the above definition holds for this particular agent $i$ and for any other agent $j$.
Given two bundles $X$ and $Y$ and an agent $i$ with valuation function $v_i$, we say that \emph{the EF$_{-v}^{+u}$ condition/relation is satisfied from $X$ to $Y$} if the condition in the above definition is satisfied, i.e., there exist item sets $S$ and $T$ satisfying $|S|\leq u$ and $|T|\leq v$ such that $v_i(X\cup S)\geq v_i(Y\setminus T)$.
When $u=0$ and $v=1$, EF$_{-v}^{+u}$ becomes the well-studied fairness notion EF1.

\begin{definition}
    An allocation $\calA=(A_1,\ldots,A_n)$ is \emph{envy-free up to one item}, denoted by EF1, if it is EF$^{+0}_{-1}$.
\end{definition}

An EF1 allocation always exists and can be efficiently computed~\cite{budish2011combinatorial,lipton2004approximately}.

Lastly, the following proposition is straightforward.
In Sect.~\ref{sect:threeagents}, we will frequently use the facts that EF1 implies EF$_{-0}^{+1}$ and that EF$_{-2}^{+0}$ implies EF$^{+1}_{-1}$, which follow from the proposition below.
\begin{proposition}\label{prop:EFvw}
    If an allocation is EF$^{+0}_{-v}$, then for any $w$ with $0\leq w\leq v$, the allocation is EF$^{+w}_{-(v-w)}$.
\end{proposition}
\begin{proof}
    Suppose $\calA=(A_1,\ldots,A_n)$ is EF$^{+0}_{-v}$.
    Consider any two agents $i$ and $j$.
    By definition, there exist an item set $S_j$ satisfying $S_j\subseteq A_j$ and $|S_j|\leq v$ such that $v_i(A_i)\geq v_i(A_j\setminus S_j)$.
    Choose any $W\subseteq S_j$ with $|W|=w$.
    Since $W\subseteq S_j\subseteq A_j$ which implies $W\cap A_i=\emptyset$, we have $v_i(A_i\cup W)=v_i(A_i)+v_i(W)\geq v_i(A_j\setminus S_j)+v_i(W)=v_i(A_j\setminus (S_j\setminus W))$, which implies EF$^{+w}_{-(v-w)}$.
\end{proof}

\subsection{Mechanisms}
A \emph{deterministic mechanism} $\calM$ is a function that takes the valuation profile $(v_1,\ldots,v_n)$ as input and outputs an allocation $\calA$.
We say that $\calM$ is EF$^{+u}_{-v}$ if it always outputs EF$^{+u}_{-v}$ allocations with respect to the input valuation profile $(v_1,\ldots,v_n)$.
We say that $\calM$ is \emph{truthful} if, for each agent $i$, truthfully reporting her valuation function $v_i$ maximizes agent $i$'s utility.
This is formally defined below.

\begin{definition}
    A deterministic mechanism $\calM$ is \emph{truthful} if, for any valuation profile $(v_1,\ldots,v_n)$, any agent $i$, and any valuation function $v_i'$, we have
    $v_i(A_i)\geq v_i(A_i')$, where $\calA=(A_1,\ldots,A_n)=\calM(v_1,\ldots,v_n)$ and $\calA'=(A_1',\ldots,A_n')=\calM(v_1,\ldots,v_{i-1},v_i',v_{i+1},\ldots,v_n)$.
\end{definition}

For deterministic mechanisms, truthfulness has a very low compatibility with fairness.
From the previous study of~\citet{amanatidis2017truthful}, it is easy to obtain the following theorem, where the proof is deferred to Appendix~\ref{append:impossibility_deterministic}.

\begin{theorem}\label{thm:impossibility_deterministic}
    For any nonnegative integers $u$ and $v$, there does not exist a deterministic mechanism $\calM$ that is truthful and EF$^{+u}_{-v}$ even with two agents.
\end{theorem}

A \emph{randomized mechanism} takes the valuation profile $(v_1,\ldots,v_n)$ as input and outputs a \emph{probability distribution} of allocations $\{(p_k,\calA_k)\}_{k=1,\ldots,K}$ where the allocation $\calA_k$ is output by the mechanism with probability $p_k$ and $\sum_{k=1}^Kp_k=1$.
Given a probability distribution of allocations $\{(p_k,\calA_k)\}_{k=1,\ldots,K}$, we can compute the marginal probability $x_{ig}$ that an item $g$ is allocated to an agent $i$.
By the law of total probability, we have $\sum_{i=1}^nx_{ig}=1$ for each item $g$.
This enables us to view the matrix $\bX=\{x_{ig}\}_{i\in N,g\in M}$ as a \emph{fractional allocation} where each item $g$ becomes \emph{divisible} and $x_{ig}$ is the \emph{fraction} of item $g$ allocated to agent $i$.
Therefore, we can interpret a randomized mechanism by a tuple $(\calF,\calD)$ where
\begin{itemize}
    \item $\calF$ is a function that takes the valuation profile $(v_1,\ldots,v_n)$ as input and outputs a fractional allocation $\bX=\{x_{ig}\}_{i\in N,g\in M}$ (that satisfies $\sum_{i=1}^nx_{ig}=1$ for each $g\in M$), and
    \item $\calD$ is a function that takes a fractional allocation $\bX=\{x_{ig}\}_{i\in N,g\in M}$ and the valuation profile $(v_1,\ldots,v_n)$ as input and outputs a probability distribution of allocations $\{(p_k,\calA_k)\}_{k=1,\ldots,K}$ such that $x_{ig}$ is the marginal probability that item $g$ is allocated to agent $i$ for any $i\in N$ and $g\in M$.
\end{itemize}
Thereafter, we will use this interpretation and use $(\calF,\calD)$ to denote a randomized mechanism.

A randomized mechanism is EF$^{+u}_{-v}$ if every allocation $\calA_k$ possibly output by the mechanism satisfies EF$^{+u}_{-v}$.
Clearly, for a randomized mechanism $(\calF,\calD)$ to be EF$^{+u}_{-v}$, we must be able to decompose the fractional allocation $\bX=\{x_{ig}\}_{i\in N,g\in M}$ output by $\calF$  to a probability distribution of allocations $\{(p_k,\calA_k)\}_{k=1,\ldots,K}$ where each $\calA_k$ is EF$^{+u}_{-v}$.
This is not always possible---considering the example where $x_{1g}=1$ and $x_{ig}=0$ for each $g$ and each $i\neq 1$, in which case agent $1$ deterministically receives all the items.

\begin{definition}
    Given a valuation profile $(v_1,\ldots,v_n)$, a fractional allocation $\bX=\{x_{ig}\}_{i\in N,g\in M}$ is \emph{EF$^{+u}_{-v}$-realizable} if it can be written as a probability distribution of allocations $\{(p_k,\calA_k)\}_{k=1,\ldots,K}$ such that each $\calA_k$ is EF$^{+u}_{-v}$ and $x_{ig}$ is the marginal probability that item $g$ is allocated to agent $i$ for any $i\in N$ and $g\in M$.
\end{definition}

Therefore, an EF$^{+u}_{-v}$ randomized mechanism $(\calF,\calD)$ first applies $\calF$ to output a fractional allocation $\bX$ that is EF$^{+u}_{-v}$-realizable, and then applies $\calD$ to $\bX$ to obtain a probability distribution of EF$^{+u}_{-v}$ allocations.

Given a fractional allocation $\bX=\{x_{ig}\}_{i\in N,g\in M}$, an agent $i$'s \emph{expected utility}, denoted by $v_i(\bX)$, is naturally given by
$$v_i(\bX)=\sum_{g\in M}v_{ig}x_{ig}.$$
We say a randomized mechanism is \emph{truthful} if, for each agent $i$, truthfully reporting her valuation function $v_i$ maximizes her expected utility, formally defined below.

\begin{definition}
    A randomized mechanism $(\calF,\calD)$ is \emph{truthful} if, for any valuation profile $(v_1,\ldots,v_n)$, any agent $i$, and any valuation function $v_i'$, we have $v_i(\bX)\geq v_i(\bX')$, where $\bX=\calF(v_1,\ldots,v_n)$ and $\bX'=\calF(v_1,\ldots,v_{i-1},v_i',v_{i+1},\ldots,v_n)$.
\end{definition}

Obviously, the property of truthfulness depends only on $\calF$ (not on $\calD$), which completely characterizes an agent's expected utility.
Based on this, we will say the fractional allocation rule $\calF$ is truthful if it satisfies the truthful property defined in the definition above.
To successfully design a mechanism $(\calF,\calD)$ that simultaneously guarantees the truthfulness and fairness property EF$^{+u}_{-v}$, we need to design $\calF$ that is truthful and meanwhile guarantee that the fractional allocation output by $\calF$ is always EF$^{+u}_{-v}$-realizable.

One natural truthful fractional allocation rule is the \emph{equal division rule}.

\begin{definition}
    The \emph{equal division rule}, denoted by $\calF_=$, is the function that outputs the fractional allocation $\bX=\{x_{ig}\}_{i\in N,g\in M}$ with $x_{ig}=\frac1n$ for any $i\in N$ and $g\in M$.
\end{definition}

The equal division rule $\calF_=$ is clearly truthful, as it ignores agents' reported valuation functions.
When designing randomized mechanisms based on $\calF_=$, the challenging part falls into EF$^{+u}_{-v}$-realizability.
As we will see later, the equal division fractional allocation $\bX=\{x_{ig}\}_{i\in N,g\in M}$ with $x_{ig}=\frac1n$ fails to be EF1-realizable for some valuation profiles.

Some other fractional division rules $\calF$ have better guarantees on EF$^{+u}_{-v}$-realizability.
For example, \citet{AzizFrSh23} show that the \emph{probabilistic serial rule}~\cite{bogomolnaia2001new} is EF1-realizable and the \emph{maximum Nash welfare rule}~\cite{nash1950bargaining} is EF$^{+1}_{-1}$-realizable.
However, it is well-known that neither of them is truthful.

\subsection{Pareto-Optimality}
Given a valuation profile $(v_1,\ldots,v_n)$, an allocation $\calA=(A_1,\ldots,A_n)$ is \emph{Pareto-optimal} if there does not exist another allocation $\calA'=(A_1',\ldots,A_n')$ such that $v_i(A_i')\geq v_i(A_i)$ for all $i\in N$ and $v_{i^\ast}(A_{i^\ast}')>v_{i^\ast}(A_{i^\ast})$ for some agent $i^\ast$.
If such an allocation $\calA'$ exists, we say that $\calA'$ \emph{Pareto-dominates} $\calA$ or $\calA'$ is a \emph{Pareto-improvement} of $\calA$.
Similarly, we can define Pareto-optimality for fractional allocations.
We say that a fractional allocation $\bX$ is Pareto-optimal if there does not exist another fractional allocation $\bX'$ such that $v_i(\bX')\geq v_i(\bX)$ for all $i\in N$ and $v_{i^\ast}(\bX')>v_{i^\ast}(\bX)$ for some agent $i^\ast$.
The notions of Pareto-domination and Pareto-improvement are defined similarly.

We say that a deterministic mechanism $\calM$ is Pareto-optimal if it always outputs Pareto-optimal allocations.
For randomized mechanisms, we can define Pareto-optimality in two different ways.

\begin{definition}
    A randomized mechanism $(\calF,\calD)$ is $\exan$ \emph{Pareto-optimal} if $\calF$ always outputs Pareto-optimal fractional allocations.
\end{definition}

\begin{definition}
    A randomized mechanism $(\calF,\calD)$ is $\expo$ \emph{Pareto-optimal} if each allocation $\calA_k$ possibly output by the mechanism is Pareto-optimal.
\end{definition}

It is easy to see that $\exan$ Pareto-optimality implies $\expo$ Pareto-optimality: if an allocation $\calA_k$ is output with a positive probability and is not Pareto-optimal, a Pareto-improvement $\calA_k'$ to $\calA_k$ results in an $\exan$ Pareto-improvement.
In particular, given the output distribution $\{(p_k,\calA_k)\}_{k=1,\ldots,K}$, if a particular $\calA_{k^\ast}$ is not Pareto-optimal, replacing $\calA_{k^\ast}$ by $\calA_{k^\ast}'$ that Pareto-dominates $\calA_{k^\ast}$ and leaving the remaining $K-1$ allocations unchanged gives a better distribution where every agent's expected utility is weakly increased and some agent's expected utility is strictly increased.

\subsection{A Technical Lemma}
In this section, we state and prove a technical lemma that is used multiple times in our paper.

\begin{lemma}\label{lem:matching}
    Given a $k$-regular bipartite (multi-)graph $G$, there exists a $k$-coloring of its edges such that the $k$ incident edges of each vertex have distinct colors.
    Moreover, such a $k$-coloring can be found in polynomial time.
\end{lemma}
\begin{proof}
    We use a well-known corollary of Hall's theorem: a $k$-regular bipartite graph with $k\geq1$ contains a perfect matching.
    If we find and remove a perfect matching $M$ in the $k$-regular bipartite graph, we obtain a $(k-1)$-regular bipartite graph.
    Therefore, by iteratively finding and removing a perfect matching for $k$ times, we can decompose $G$ into $k$ disjoint perfect matchings.
    We can then obtain a valid $k$-coloring.
    In addition, a perfect matching can be found in polynomial time with standard algorithms.
\end{proof}

\section{Truthful and EF1 Mechanism for Two Agents}
\label{sect:twoagents}
In this section, we present a simple randomized mechanism $(\calF,\calD)$ for two agents that simultaneously guarantees the truthfulness and EF1 fairness property.
In addition, we show that using the equal division rule $\calF=\calF_=$ suffices.

\begin{theorem}\label{thm:two}
    There exists a truthful and EF1 randomized mechanism $(\calF_=,\calD)$ for $n=2$. In addition, the mechanism outputs the distribution of EF1 allocations in polynomial time.
\end{theorem}

The truthfulness of the mechanism is guaranteed by the property of the equal division rule.
To prove Theorem~\ref{thm:two}, it remains to show that the equal division fractional allocation $\{x_{ig}\}_{i\in N,g\in M}$ with $x_{ig}=\frac12$ is EF1-realizable and there exists a polynomial time decomposition rule $\calD$ that achieves this.
This is implied by the following proposition first proved by~\citet{kyropoulou2020almost}.
We present our proof in Appendix~\ref{append:prop:perfecttwo} for completeness, which contains some ideas for our result in Sect.~\ref{sect:threeagents}.

\begin{proposition}\label{prop:perfecttwo}
    Given a valuation profile of two agents $(v_1,v_2)$, we can compute in polynomial time a partition $(X,Y)$ of $M$ such that both $(A_1,A_2)=(X,Y)$ and $(A_1,A_2)=(Y,X)$ are EF1 allocations.
\end{proposition}

\paragraph{Proof of Theorem~\ref{thm:two}.} With the proposition above, Theorem~\ref{thm:two} is now straightforward.
Given the equal division fractional allocation rule $\calF_=$, the decomposition rule $\calD$ produces the allocation $(X,Y)$ with probability $0.5$ and the allocation $(Y,X)$ with probability $0.5$.
Since $(X,Y)$ is a partition, each item is allocated to each agent with probability $0.5$, which matches the equal division fractional allocation.

\subsection{Barriers for Extending to Three Agents}
To design a truthful and EF1 (or EF$^{+u}_{-v}$) mechanism for more than two agents, one natural idea is to use the equal division rule $\calF_=$ that guarantees truthfulness, as we did for two agents.
However, EF$^{+u}_{-v}$-realizability becomes challenging.
It may be natural to believe that the equal division fractional allocation, being envy-free in the fractional sense, has a good EF$^{+u}_{-v}$-realizability.
However, the following theorem suggests that at least this is not true for EF1, which may be surprising to the readers.

\begin{theorem}\label{thm:notEF1realizable}
    For three agents, there exists a valuation profile $(v_1,v_2,v_3)$ such that the equal division fractional allocation $\bX=\{x_{ig}\}_{i=1,2,3;g\in M}$ is not EF1-realizable.
\end{theorem}
\begin{proof}
        Consider the example of three agents and four items with valuations defined in the table below.
\begin{center}
    \begin{tabular}{ccccc}
    \hline
         & $g_1$ & $g_2$ & $g_3$ & $g_4$ \\
         \hline
      $v_1$ & 2 & 1 & 4 & 4\\
      $v_2$ & 1 & 2 & 4 & 4\\
      $v_3$ & 4 & 4 & 2 & 1\\
      \hline
    \end{tabular}
\end{center}

    We will show that the equal division fractional allocation is not EF1-realizable for this instance. Suppose for the sake of contradiction it is.

    Firstly, in all realized allocations, agent $1$ must receive at least one item from $\{g_1,g_3,g_4\}$, for otherwise, agent $1$ will at most receive $g_2$ which has value only $1$, and one of agent $2$ or $3$ will receive at least two items from $\{g_1,g_3,g_4\}$, which violates the EF1 property.
    Moreover, in all realized allocations, agent $1$ must receive exactly one item from $\{g_1,g_3,g_4\}$: since the expected number of items agent $1$ receives from $\{g_1,g_3,g_4\}$ is $\frac13+\frac13+\frac13=1$, if agent $1$ receives at least two items from $\{g_1,g_3,g_4\}$ in an allocation with a positive probability, then she will receive no item from $\{g_1,g_3,g_4\}$ in another allocation with a positive probability, and we have seen that this violates the EF1 property.
    For a similar reason, agent $2$ must receive exactly one item from $\{g_2,g_3,g_4\}$ in all allocations, and agent $3$ must receive exactly one item from $\{g_1,g_2,g_3\}$ in all allocations.

    Consider one realized allocation where $g_3$ is allocated to agent $3$.
    Agent $3$ cannot further take $g_4$ in this allocation, for otherwise the EF1 property from agent $1$/$2$ to agent $3$ is violated.
    By symmetry of agents $1$ and $2$, assume without loss of generality that agent $1$ takes $g_4$.
    Since we have shown that agent $1$ cannot take more than one item in $\{g_1,g_3,g_4\}$ and agent $3$ cannot take more than one item in $\{g_1,g_2,g_3\}$, item $g_1$ must then be allocated to agent $2$.

    Finally, we show that the remaining item $g_2$ cannot be properly allocated to ensure the EF1 property.
    If $g_2$ is allocated to agent $1$, the EF1 property from agent $2$ to agent $1$ is violated.
    If $g_2$ is allocated to agent $2$, then the EF1 property from agent $3$ to agent $2$ is violated.
    If $g_2$ is allocated to agent $3$, then agent $3$ receives two items from $\{g_1,g_2,g_3\}$, and we have seen that this will eventually violate the EF1 property.
\end{proof}

We have seen that the equal division fractional allocation is not EF1-realizable.
On the other hand, we will show in Sect.~\ref{sect:nagents} that it is not too bad: the equal division fractional allocation is EF$^{+0}_{O(-\sqrt{n})}$-realizable.
However, even with three agents, we do not know the minimum values for $u$ and $v$ where the equal division fractional allocation is EF$^{+u}_{-v}$-realizable, and we believe this is an interesting open problem.
In the next section, we will use a more sophisticated fractional division rule $\calF$ to design a truthful and EF$^{+1}_{-1}$ mechanism.

\section{Truthful and EF$_{-1}^{+1}$ Mechanism for Three Agents}
\label{sect:threeagents}
In this section, we present a randomized mechanism $(\calF,\calD)$ for three agents that is truthful and EF$^{+1}_{-1}$.
We will prove the following theorem in this section.

\begin{theorem}\label{thm:three}
    There exists a truthful and EF$^{+1}_{-1}$ randomized mechanism $(\calF,\calD)$ for $n=3$. In addition, the mechanism outputs the distribution of EF$^{+1}_{-1}$ allocations in polynomial time.
\end{theorem}

The remaining part of this section is for the proof of Theorem~\ref{thm:three}, and it is organized as follows.
In Sect.~\ref{sect:three-fractional}, we define the fractional division rule $\calF$ and show that it is truthful.
From Sect.~\ref{sect:three-agent3} to Sect.~\ref{sect:three-comb}, we describe the decomposition rule $\calD$ and show that it is EF$^{+1}_{-1}$.

\paragraph{Technical Overview.}
Our fractional division rule $\calF$ in Sect.~\ref{sect:three-fractional} requires partitioning the items into $m/3$ groups based on the descending order of agent $3$'s value to each item, each of which consists of three items.
Depending on the valuations of agents $1$ and $2$, we classify these groups into Type I and Type II.
For a three-item group, it belongs to Type I if the favorite items for agents $1$ and $2$ are different, and it belongs to Type II otherwise.
Specifically, a group $G=\{a,b,c\}$ belongs to Type I if we can find two items $a$ and $b$ such that agent $1$ and agent $2$ respectively believe that $a$ and $b$ have the highest values, where ties are allowed but the tie-breaking rule should be carefully designed to ensure truthfulness.
For such a group, each of the agents $1, 2$, and $3$ will receive $2/3, 0$, and $1/3$ fraction of item $a$ respectively, $0, 2/3$, and $1/3$ fraction of $b$ respectively, and $c$ will be allocated to the three agents evenly.
A group that does not belong to Type I belongs to Type II, where each of the three items will be allocated evenly among the three agents.
We show in Lemma~\ref{lem:three_truthful} that $\calF$ is truthful.

The next four sections \ref{sect:three-agent3}-\ref{sect:three-comb} describe the decomposition rule $\calD$.
For an integral allocation, if each bundle contains exactly one item from each three-item group, we will show it is EF$1$ to agent $3$ in Sect.~\ref{sect:three-agent3}, and such an allocation is referred to as a \emph{regular} allocation.
In Sect.~\ref{sect:three-1} and Sect.~\ref{sect:three-2}, we will restrict ourselves to regular allocations.

In Sect.~\ref{sect:three-1}, we will handle items that fall into groups of Type I.
We show that the fractional allocation of Type I groups can be decomposed into a distribution over EF$1$ allocations.
Specifically, we first introduce three allocation rules, each of which generates a possible integral allocation for a three-item group, such that the marginal probabilities match the fractional allocation if each allocation rule is applied to each group with a specific probability.
We then find a three-partition of Type I groups through a $3$-coloring of a $3$-regular bipartite graph as in Lemma~\ref{lem:matching}.
This partition guarantees that, when groups in the same set in the partition follow the same allocation rule while groups in different sets follow different allocation rules, the resulting integral allocation of Type I items satisfies EF$1$.
Additionally, for further combination, we will show in Sect.~\ref{sect:three-comb} that the fairness guarantee between some specific pairs of bundles to an agent is stronger than EF$1$ through a more careful analysis.

In Sect.~\ref{sect:three-2}, we will handle items that fall into groups of Type II.
We will find a distribution of EF$^{+1}_{-1}$ allocations such that the marginal probabilities match the fractional allocation.
Specifically, we will construct a three-partition $(X,Y,Z)$ of the items in Type II groups such that any permutation of $(X,Y,Z)$ gives an EF$^{+1}_{-1}$ allocation.
To satisfy the envy-based fairness guarantee, we would like each bundle to be about the average value for agents $1$ and $2$, with the added requirement that any deviation of a bundle from this average can be adjusted by adding or removing one item.
The construction of bundle $X$ involves a careful rounding of the fractional solution of a linear program, which relies on the fact that both agents have the same favorite item in each Type II group.
Bundles $Y$ and $Z$ are then constructed using a combinatorial approach similar to that for Type I groups.
Similarly, for further combination, we remark that the fairness guarantee of our partition between some specific pairs of bundles to an agent is stronger than EF$_{-1}^{+1}$.

Combining the results above, we can obtain a truthful and EF$_{-2}^{+1}$ mechanism.
In Sect.~\ref{sect:three-comb}, we further improve the fairness guarantee to EF$_{-1}^{+1}$, which requires a more advanced choosing and combining method of the two types of allocations.
Note that in this section, the regularity of Type II allocations may be slightly violated, yet the marginal probability of each item allocated to each agent is still satisfied so that truthfulness is preserved.
As a result of this slight violation, the EF$1$ property for agent $3$ will be relaxed to EF$_{-1}^{+1}$, but it introduces more flexibility which makes it possible to improve the fairness guarantee from EF$_{-2}^{+1}$ to EF$_{-1}^{+1}$ for agents $1$ and $2$.
We then conclude that the fractional allocation output by $\calF$ is EF$^{+1}_{-1}$-realizable, which concludes Theorem~\ref{thm:three}.

\subsection{The Fractional Division Rule $\calF$ and Its Truthfulness}
\label{sect:three-fractional}
In this section, we define $\calF$ and prove that it is truthful.

First of all, we assume without loss of generality that $m$ is a multiple of $3$, for otherwise we can add one or two dummy items where agents have value $0$.
Let agent $3$ sort the items $g^{(3)}_1,\ldots,g^{(3)}_m$ by the descending values: $v_3(g^{(3)}_1)\geq v_3(g^{(3)}_2)\geq\cdots \geq v_3(g^{(3)}_m)$.
Ties are broken arbitrarily.
Based on agent $3$'s valuation, define the partition $(G^{(3)}_1,\ldots,G^{(3)}_{m/3})$ of $M$ where $G^{(3)}_j=\{g_{3j-2}^{(3)},g_{3j-1}^{(3)},g_{3j}^{(3)}\}$.

The rule $\calF$ decides how the three items in each group are (fractionally) allocated to the three agents.
For agent $3$, she always receives a fraction of $1/3$ for each item.
The allocation for agents $1$ and $2$ depends on their value rankings over the three items, which is specified as follows.

Let $a,b,c$ be the three items in a group $G^{(3)}_j$.
At a high level, the group falls into one of the following two cases, and we will call the first case \emph{Type I} and the second case \emph{Type II}.
\begin{itemize}
    \item \textbf{Type I}: If it is possible to find two items $a$ and $b$ such that agent $1$ believes $a$ has the highest value (ties are allowed) and agent $2$ believes $b$ has the highest value (ties are allowed), then $a$ is allocated such that agents $1$, $2$, and $3$ receive fractions of $2/3, 0$, and $1/3$ respectively, $b$ is allocated such that agent $1$, $2$, and $3$ receive fractions of $0, 2/3$, and $1/3$ respectively, and $c$ is divided equally among the three agents.
    \item \textbf{Type II}: If both agents $1$ and $2$ believe an item $a$ has values strictly higher than $b$ and $c$, then all three items are divided equally among the three agents.
\end{itemize}

For those Type I groups, there may be ties, i.e., the selections of $a$ and $b$ such that agent $1$ favors $a$ and agent $2$ favors $b$ may not be unique. This happens when an agent values two or three items equally.
In this case, ties need to be handled \emph{properly} to guarantee truthfulness.
In particular, breaking the ties by a consistent item index order, which is commonly used in other mechanisms, fails here, as illustrated in Appendix~\ref{append:subtlety}.
Intuitively, when an agent has a tie on the item with the highest value, the tie should be broken in a way that favors the other agent.

Let $F_1\subseteq G^{(3)}_j$ be the set of items that agent $1$ favors. Formally, $F_1$ is the set of items that agent $1$ equally prefers, and they have values strictly higher than items in $G^{(3)}_j\setminus F_1$.
We have $1\leq |F_1|\leq 3$.
Let $F_2$ be the set of items favored by agent $2$ which is defined similarly.
Algorithm~\ref{alg:tiebreaking} selects one ``favorite'' item for each of the two agents $1$ and $2$ given $F_1$ and $F_2$.
We have a total of ten cases.
The first case corresponds to Type II group.
The remaining nine cases correspond to Type I group.
For Case 2 and Case 3, the selection of favorite items is unique provided that the favorite items for both agents are different.
For Case 6, both agents believe two of the three items are equally valuable and more valuable than the third, in which case the selection of favorite items is also more or less unique.
For the remaining six cases, we always select the favorite items such that one agent's favorite item is least preferred (or equally least preferred) by the other agent.

\begin{algorithm}[h]
\caption{Selecting favorite items for agents $1$ and $2$ from each group $G^{(3)}_j$.}
\label{alg:tiebreaking}
\KwInput{$F_1,F_2\subseteq G^{(3)}_j$}
\KwOutput{A favorite item for agent $1$ and a favorite item for agent $2$}
If $|F_1\cup F_2|=1$, let the unique item in $F_1$ or $F_2$ be the favorite item for both agents\;
If $|F_1|=1$ and $|F_2\setminus F_1|=1$, let the item in $F_1$ be agent $1$'s favorite and the item in $F_2\setminus F_1$ be agent $2$'s favorite\;
If $|F_2|=1$ and $|F_1\setminus F_2|=1$, this is symmetric to the previous case and is handled similarly\;
If $|F_1|=1$ and $|F_2\setminus F_1|=2$, let the item in $F_1$ be agent $1$'s favorite and the item in $F_2\setminus F_1$ where agent $1$ has a less value be agent $2$'s favorite; if agent $1$ values both items in $F_2\setminus F_1$ equally, select an arbitrary item be agent $2$'s favorite\;
If $|F_2|=1$ and $|F_1\setminus F_2|=2$, this is symmetric to the previous case and is handled similarly\;
If $|F_1|=|F_2|=2$ and $F_1=F_2$, arbitrarily select an item from $F_1=F_2$ for agent $1$'s favorite, and the other item is agent $2$'s favorite\;
If $|F_1|=|F_2|=2$ and $F_1\neq F_2$, let the item in $F_1\setminus F_2$ be agent $1$'s favorite and the item in $F_2\setminus F_1$ be agent~$2$'s favorite\;
If $|F_1|=2$ and $|F_2|=3$, agent $1$ selects an arbitrary item from $F_1$ as her favorite, and agent $2$ selects the item in $F_2\setminus F_1$ as her favorite\;
If $|F_2|=2$ and $|F_1|=3$, this is symmetric to the previous case and is handled similarly\;
If $|F_1|=|F_2|=3$, select the favorite items for the two agents arbitrarily, provided that different items are selected for both agents\;
\end{algorithm}

Given the rule for selecting the favorite items determined, the allocation rule is given in Table~\ref{tab:ruleF}.

\begin{table}[h]
    \centering
    \begin{tabular}{ll}
    \hline
        Type I: & $a$ is agent $1$'s favorite item and $b$ is agent $2$'s favorite item  \\
        Allocation: & $(x_{1a},x_{2a},x_{3a})=(\frac23,0,\frac13)$, $(x_{1b},x_{2b},x_{3b})=(0,\frac23,\frac13)$, $(x_{1c},x_{2c},x_{3c})=(\frac13,\frac13,\frac13)$\\
    \hline
        Type II: & $a$ is the favorite item for both agents $1$ and $2$\\
        Allocation: & $(x_{1a},x_{2a},x_{3a})=(\frac13,\frac13,\frac13)$, $(x_{1b},x_{2b},x_{3b})=(\frac13,\frac13,\frac13)$, $(x_{1c},x_{2c},x_{3c})=(\frac13,\frac13,\frac13)$\\
    \hline
    \end{tabular}
    \caption{The allocation rule for each group $G^{(3)}_j=\{a,b,c\}$ with two different types.}
    \label{tab:ruleF}
\end{table}

We have then completely defined the allocation rule for each group $G^{(3)}_j$, which concludes the definition of the fractional division rule $\calF$.

\begin{lemma}\label{lem:three_truthful}
    The fractional division rule $\calF$ defined in this section is truthful.
\end{lemma}
\begin{proof}
    The truthfulness for agent $3$ is trivial, as she receives each item with the fraction $1/3$ regardless of the valuation function she reports.
    For agents $1$ and $2$, it suffices to analyze the truthfulness within each group $G^{(3)}_j=\{a,b,c\}$, as the partition of the groups depends solely on agent $3$'s valuation function.
    We will analyze agent $1$ without loss of generality.

    If the group is of Type II when agent $1$ reports her valuation truthfully, we have $v_1(a)>\max\{v_1(b),v_1(c)\}$ and $v_2(a)>\max\{v_2(b),v_2(c)\}$.
    If agent $1$ reports a valuation function such that the group is still of Type II, the allocation does not change, and agent $1$'s expected utility does not change.
    Suppose agent $1$ reports a valuation function that makes the group of Type I.
    Then agent $2$ will take $2/3$ fraction of item $a$, and agent $1$ will lose a fraction $1/3$ from item $a$ compared with truth-telling.
    Agent $1$ will gain a fraction $1/3$ from either item $b$ or item $c$.
    Since agent $1$ values item $a$ strictly higher than $b$ or $c$, the misreporting of the valuation function is harmful to agent $1$.

    Now, suppose the group is of Type I when agent $1$ reports truthfully.
    If she misreports her valuation function such that the group becomes a Type II group, then she loses a fraction $1/3$ from the item she prefers the most and gains a fraction $1/3$ from another item for which she may or may not prefer the most.
    It is clear that the misreporting is not beneficial.
    It remains to discuss the case where the group is still of Type I after agent $1$'s misreporting.

    Assume without loss of generality that $v_1(a)\geq v_1(b)\geq v_1(c)$.
    If $v_1(b)=v_1(c)$, it is easy to see that agent $1$ will always receive a value of $\frac23\cdot v_1(a)+\frac13\cdot v_1(b)$ by truth-telling, which already maximizes agent $1$'s utility subject to our division rule.
    The only case agent $1$ can possibly benefit is when $v_1(b)>v_1(c)$ and agent $1$ receives a value of $\frac23\cdot v_1(a)+\frac13\cdot v_1(c)$.
    In this case, $c$ is not the favorite item for agent $2$, and agent $1$ can only possibly benefit by misreporting such that $c$ becomes the favorite item for agent $2$.
    We will show that this is impossible.
    
    If $c\notin F_2$, item $c$ will never be the favorite item for agent $2$, and agent $1$ cannot change this fact by misreporting.
    If $c\in F_2$, by checking all the relevant cases in Algorithm~\ref{alg:tiebreaking}, we can see that $c$ is always selected as the favorite item for agent $2$.
    Indeed, cases 2, 3, 4, 5, 7, and 8 are the relevant cases, where it is easy to check that agent $2$ chooses item $c$ as her favorite.
    Case 1 is about Type I group and is thus irrelevant.
    We have $c\notin F_2$ for Case 6 (in particular, $c\notin F_1$ as $v_1(b)> v_1(c)$, so $c\notin F_2$ as $F_1=F_2$), so Case 6 is irrelevant.
    Case 9 and 10 are irrelevant as we have assumed $v_1(b)>v_1(c)$ (so $|F_1|=3$ is impossible).
\end{proof}

\subsection{Guaranteeing EF$_{-1}^{+1}$ Property for Agent 3}
\label{sect:three-agent3}
From this section on, we will describe the decomposition rule $\calD$ and show the fairness property EF$^{+1}_{-1}$.

Firstly, if each group $G^{(3)}_j$ is allocated such that each agent is allocated exactly one item, we show that the EF1 property for agent $3$ is satisfied.

\begin{proposition}\label{prop:regular}
    Let $(A_1,A_2,A_3)$ be an allocation where $|A_i\cap G^{(3)}_j|=1$ for each $i=1,2,3$ and each $j=1,\ldots,m/3$. Then the EF1 property is satisfied for agent $3$.
\end{proposition}
\begin{proof}
    For each $i=1,2,3$ and each $j=1,\ldots,m/3$, let $g_{ij}$ be the unique item in the set $A_i\cap G_j^{(3)}$.
    Then $A_1=\{g_{1j}\}_{j=1,\ldots,m/3}$, $A_2=\{g_{2j}\}_{j=1,\ldots,m/3}$, and $A_3=\{g_{3j}\}_{1,\ldots,m/3}$.
    By the way the groups are defined, we have $v_3(g_{i_1j})\geq v_3(g_{i_2(j+1)})$ for any $i_1,i_2\in\{1,2,3\}$ and any $j=1,\ldots,m/3-1$.
    Therefore, we have $v_3(A_3)\geq v_3(A_1\setminus\{g_{11}\})$ and $v_3(A_3)\geq v_3(A_2\setminus\{g_{21}\})$, which concludes the proposition.
\end{proof}

In Sect.~\ref{sect:three-1} and Sect.~\ref{sect:three-2}, we will make sure every allocation generated by $\calD$ satisfies the property in the proposition above.
We will call such allocations \emph{regular}.
By the proposition above, the property EF1 is guaranteed for agent $3$.
Moreover, given a regular allocation $(A_1,A_2,A_3)$, all its permutations (e.g., $(A_1,A_3,A_2)$, $(A_2,A_1,A_3)$, etc.) are also regular by definition.
Thereafter, we will also use the word ``regular'' to describe a three-partition of $M$.

In Sect.~\ref{sect:three-comb}, we will violate the regularity requirement by a little bit.
Specifically, it may have the form $(A_1\setminus\{g\}, A_2\cup\{g\}, A_3)$ for an item $g\in A_1$, or $(A_1\cup\{g'\}, A_2\setminus\{g'\}, A_3)$ for an item $g'\in A_2$, where $(A_1, A_2, A_3)$ is a regular allocation.
Therefore, EF$_{-1}^{+1}$ for agent $3$ is ensured, and we will only discuss agents $1$ and $2$ thereafter.

\subsection{EF1 Allocations for Type I}
\label{sect:three-1}
Let $M^I$ be the set of items that belong to groups of Type I.
In this section, we will find three regular (satisfying the property in Proposition~\ref{prop:regular}) EF1 allocations $\calA_1^I,\ldots,\calA_{3}^I$ of $M^I$, each of which is sampled with probability $1/3$, such that the marginal probability that each item is allocated to each agent follows the probability specified by Table~\ref{tab:ruleF}.

We first specify three possible allocations of the three items in a single group $\{a_1,a_2,b\}$, where $a_1$ is agent $1$'s favorite item and $a_2$ is agent $2$'s favorite item.
\begin{itemize}
    \item Allocation $\chi_s$: agent $1$ receives $a_1$, agent $2$ receives $a_2$, and agent $3$ receives $b$.
    \item Allocation $\chi_1$: agent $1$ receives $a_1$, agent $2$ receives $b$, and agent $3$ receives $a_2$.
    \item Allocation $\chi_2$: agent $1$ receives $b$, agent $2$ receives $a_2$, and agent $3$ receives $a_1$.
\end{itemize}
If each of $\chi_s$, $\chi_1$, and $\chi_2$ is sampled with probability $1/3$, it is straightforward to check that $a_1$ is allocated to agent $1$ with probability $2/3$, $a_2$ is allocated to agent $2$ with probability $2/3$, and each item is allocated to agent $3$ with probability $1/3$.
This agrees with the probabilities specified in Table~\ref{tab:ruleF}.
We will make sure that, for each group $G^{(3)}_j$, each of $\chi_s$, $\chi_1$, and $\chi_2$ appears exactly once in the three allocations $\calA_1^I, \calA_2^I, \calA_3^I$.

Given a group $G=\{a_1,a_2,b\}$, let $\Delta_1(G)=v_1(a_1)-v_1(b)$ and $\Delta_2(G)=v_2(a_2)-v_2(b)$.
Let both agents sort the groups by the descending values of $\Delta_1(\cdot)$ and $\Delta_2(\cdot)$ respectively.
Let $G^{(1)}_1,\ldots,G^{(1)}_{k}$ and $G^{(2)}_1,\ldots,G^{(2)}_{k}$ be the results of the sorting.
Notice that $(G^{(1)}_1,\ldots,G^{(1)}_{k})$ is a permutation of $(G^{(2)}_1,\ldots,G^{(2)}_{k})$ and $k$ is the number of Type I groups.
We will specify an allocation of $M^I$ by specifying one of $\chi_s,\chi_1,\chi_2$ for each of the $k$ groups.
We will assume $k$ is a multiple of $3$ without loss of generality: if not, we add dummy items that form dummy groups.
Next, define $H^{(1)}_{j}=\{G^{(1)}_{3j-2},G^{(1)}_{3j-2},G^{(1)}_{3j}\}$ for each $j=1,\ldots,k/3$.
Define $H^{(2)}_j$ similarly for each $j=1,\ldots,k/3$.
The proposition below gives a sufficient condition for an allocation to be EF1.

\begin{proposition}\label{prop:typeIEF1}
    Suppose an allocation $\calA=(A_1,A_2,A_3)$ satisfies that, for each $H^{(i)}_j$ with $i=1,2$ and $j=1,\ldots,k/3$, exactly one group in $H^{(i)}_j$ is allocated using rule $\chi_s$, exactly one group in $H^{(i)}_j$ is allocated using rule $\chi_1$, and exactly one group in $H^{(i)}_j$ is allocated using rule $\chi_2$.
    Then $\calA$ is EF1.
\end{proposition}
\begin{proof}
    The allocation $\calA$ is regular, so the EF1 property for agent $3$ is guaranteed by Proposition~\ref{prop:regular}.
    We will show that the EF1 property is also guaranteed for agent $1$, and the analysis for agent $2$ is similar.

    We first show that agent $1$ will not envy agent $2$ if one item were removed from agent $2$'s bundle.
    We find a lower bound to $v_1(A_1)-v_1(A_2)$ by considering the contribution of the allocation of each three-item group.
    Given a group $G=\{a_1,a_2,b\}$, if the allocation is $\chi_s$, its contribution to $v_1(A_1)-v_1(A_2)$ is $v_1(a_1)-v_1(a_2)\geq0$; if the allocation is $\chi_1$, its contribution is $v_1(a_1)-v_1(b)=\Delta_1(G)$; if the allocation is $\chi_2$, its contribution is $v_1(b)-v_1(a_2)\geq v_1(b)-v_1(a_1)=-\Delta_1(G)$.
    Consider two adjacent three-group sets $H^{(1)}_j=\{G^{(1)}_{3j-2},G^{(1)}_{3j-1},G^{(1)}_{3j}\}$ and $H^{(1)}_{j+1}=\{G^{(1)}_{3j+1},G^{(1)}_{3j+2},G^{(1)}_{3j+3}\}$.
    The value $\Delta_1$ for any of the three groups in $H^{(1)}_j$ is weakly higher than the value of $\Delta_1$ for any of the three groups in $H^{(1)}_{j+1}$.
    If each of $\chi_s,\chi_1,\chi_2$ appears exactly once in each of $H^{(1)}_j$ and $H^{(1)}_{j+1}$, the positive contribution to $v_1(A_1)-v_1(A_2)$ from the allocation $\chi_1$ in the group in $H^{(1)}_j$ is sufficient to offset the negative contribution to $v_1(A_1)-v_1(A_2)$ from the allocation $\chi_2$ in the group in $H^{(1)}_{j+1}$.
    In addition, $\chi_s$ and $\chi_1$ allocations never give negative contribution to $v_1(A_1)-v_1(A_2)$.
    As a result, $v_1(A_1)-v_1(A_2)$ is lower-bounded by the negative contribution of allocation $\chi_2$ in the first three-group set $H^{(1)}_1$.
    Let $\{a_1^\ast,a_2^\ast,b^\ast\}$ be the group in $H^{(1)}_1$ where allocation rule $\chi_2$ is applied.
    Then item $a_2^\ast$ is allocated to agent $2$.
    If we remove $a_2^\ast$ from agent $2$'s bundle, the negative contribution of allocation $\chi_2$ to $v_1(A_1)-v_1(A_2)$ is eliminated, and agent $1$ no longer envies agent $2$.

    The analysis for agent $1$ and agent $3$ is similar.
    For each group $G=\{a_1,a_2,b\}$, the allocation $\chi_s$ contributes $v_1(a_1)-v_1(b)=\Delta_1(G)$ to $v_1(A_1)-v_1(A_3)$, the allocation $\chi_1$ contributes $v_1(a_1)-v_1(a_2)\geq0$, and the allocation $\chi_2$ contributes $v_1(b)-v_1(a_1)=-\Delta_1(G)$.
    Again, $\chi_2$ is the only rule that can possibly give a negative contribution to $v_1(A_1)-v_1(A_3)$, and the effect for $\chi_2$ in $H^{(1)}_{j+1}$ is offset by the effect of $\chi_s$ in $H^{(1)}_j$.
    The remaining part of the analysis is similar.
\end{proof}

The next proposition is crucial for the existence of allocations that satisfy the condition in Proposition~\ref{prop:typeIEF1}.

\begin{proposition}\label{prop:typeI3coloring}
    It is possible to assign each of the $k$ Type I groups one of the three colors such that the three groups in each $H_j^{(i)}$ (where $i=1,2$ and $j=1,\ldots,k/3$) have distinct colors. In addition, such a $3$-coloring can be found in polynomial time.
\end{proposition}
\begin{proof}
    We construct a bipartite graph $G=(V_1,V_2,E)$ where $V_1$ is the set of $k/3$ vertices representing $H_1^{(1)},\ldots,H_{k/3}^{(1)}$, $V_2$ is the set of $k/3$ vertices representing $H^{(2)}_1,\ldots,H_{k/3}^{(2)}$, and each edge in $E$ represents a Type I group $G$ such that it is incident to the vertex $H^{(i)}_j$ if $G\in H^{(i)}_j$.
    Then $G$ is a $3$-regular bipartite graph.
    By Lemma~\ref{lem:matching}, we can find a $3$-coloring of the edges such that each vertex is incident to three edges with distinct colors.
    This corresponds to a valid $3$-coloring of the $k$ Type I groups.
    Moreover, such a $3$-coloring can be found in polynomial time.
\end{proof}

Finally, consider a $3$-coloring of the $k$ groups that satisfies the description in Proposition~\ref{prop:typeI3coloring}.
We may define a partition $(P,Q,R)$ of the groups from the $3$-coloring such that each set in the partition contains the groups with the same color.
There are six ways to match the three sets of groups $P,Q$, and $R$ to the three allocation rules $\chi_s,\chi_1$, and $\chi_2$.
For simplicity, we use $S_{\chi}$ to denote that the set $S$ is matched with the allocation rule $\chi$, where $S\in\{P, Q, R\}$ and $\chi\in\{\chi_s,\chi_1,\chi_2\}$.
Each matching specifies an allocation.
By Proposition~\ref{prop:typeIEF1}, all six matchings correspond to EF1 allocations.
We further choose three allocations out of the above six allocations, such that each bundle is matched to each allocation rule exactly once among the three allocations.
For example, the three allocations that satisfy this requirement can be $(P_{\chi_1}, Q_{\chi_s}, R_{\chi_2}), (P_{\chi_s}, Q_{\chi_2}, R_{\chi_1})$, and $(P_{\chi_2}, Q_{\chi_1}, R_{\chi_s})$; on the other hand, we cannot choose $(P_{\chi_s}, Q_{\chi_1}, R_{\chi_2}), (P_{\chi_s}, Q_{\chi_2}, R_{\chi_1})$, and $(P_{\chi_2}, Q_{\chi_1}, R_{\chi_s})$, as $P$ is matched with $\chi_s$ twice (also $Q$ is matched with $\chi_1$ twice).
For any three allocations satisfying the above requirement, if each of the three allocations is sampled with probability $1/3$, the marginal probabilities for the item assignments agree with the probabilities specified in Table~\ref{tab:ruleF}.

To achieve EF$^{+1}_{-2}$, we can arbitrarily choose three allocations (with the marginal probabilities for the item assignments agreeing with Table~\ref{tab:ruleF}), this, combined with the result in the next section, guarantees EF$^{+1}_{-2}$ fairness.
To further achieve EF$^{+1}_{-1}$, the three allocations must be carefully chosen; this is discussed in Sect.~\ref{sect:three-comb}.

\subsection{EF$^{+1}_{-1}$ Allocations for Type II}
\label{sect:three-2}
Let $M^{II}$ be the set of items that belong to groups of Type II.
In this section, we will find three regular EF$_{-1}^{+1}$ allocations $\calA_1^{II},\calA_2^{II},\calA_{3}^{II}$ of $M^{II}$, each of which is sampled with probability $1/3$, such that each item is allocated to each agent with probability $1/3$, which agrees with Table~\ref{tab:ruleF}.
To achieve this, we will find a partition $(X,Y,Z)$ of $M^{II}$ such that any permutation of $(X,Y,Z)$ gives a regular EF$^{+1}_{-1}$ allocation, i.e., the six allocations $(X,Y,Z),(X,Z,Y),(Y,X,Z),(Y,Z,X),(Z,X,Y)$, and $(Z,Y,X)$ are all regular and EF$^{+1}_{-1}$.
Among the six allocations, we further choose three allocations such that each bundle is allocated to each agent exactly once among the three allocations.
Since $(X,Y,Z)$ is a partition, by assigning probability $1/3$ to each of the three allocations we choose, it is ensured that each item is allocated to each agent with probability $1/3$.
In the remaining part of this section, we will describe how to find $X,Y$, and $Z$.

By Proposition~\ref{prop:regular}, the EF1 property for agent $3$ is always guaranteed if $(X,Y,Z)$ is regular.
To guarantee the EF$^{+1}_{-1}$ property for agents $1$ and $2$, the high-level idea is to find $X$, $Y$, and $Z$ such that both agents believe the three bundles have similar values.

We rescale the valuations of agents $1$ and $2$ such that $v_1(M^{II})=v_2(M^{II})=1$.
We will first find $X$ such that $X$ contains exactly one item from each Type II group and, for each agent $i\in\{1,2\}$, we have
\begin{itemize}
    \item if $v_i(X)<\frac13$, there exists an item $g_i^+\in M^{II}\setminus X$ such that $v_i(X)+\frac23\cdot v_i(g_i^+)\geq\frac13$;
    \item if $v_i(X)>\frac13$, there exists an item $g_i^-\in X$ such that $v_i(X)-\frac23\cdot v_i(g_i^-)\leq\frac13$.
\end{itemize}
In addition, fix an arbitrary agent $i\in\{1,2\}$, we will show that we can correspondingly construct $Y$ and $Z$ such that:
\begin{itemize}
    \item each of $Y$ and $Z$ contains exactly one item from each Type II group;
    \item the EF1 conditions from $Y$ to $Z$ and from $Z$ to $Y$ are both satisfied for agent $i$, i.e., if $v_i(Y)\le v_i(Z)$, there exists an item $h_i\in Z$ such that $v_i(Y)\ge v_i(Z\setminus \{h_i\})$, and vice versa;
    \item the EF$^{+0}_{-2}$ conditions from $Y$ to $Z$ and from $Z$ to $Y$ are both satisfied for the other agent $3-i$, i.e., if $v_{3-i}(Y)\le v_{3-i}(Z)$, there exists two items $h_{3-i}, h'_{3-i}\in Z$ such that $v_{3-i}(Y)\ge v_{3-i}(Z\setminus \{h_{3-i}, h'_{3-i}\})$, and vice versa.
\end{itemize}

Intuitively, we require the value of $X$ to be about the average, up to adding a $2/3$ fraction of an item or removing a $2/3$ fraction of an item. 
After $X$ is determined, our requirements guarantee the values of both $Y$ and $Z$ to be about the average of the value of the remaining items, up to adding or removing one item: $h_i$ for agent $i$ and the item with a higher value between $h_{3-i}$ and $h'_{3-i}$ for agent $3-i$.
These should be true from the perspective of both agents $1$ and $2$.

Before describing how to find $X$, $Y$, and $Z$, we first prove that any permutation of $(X,Y,Z)$ gives an EF$^{+1}_{-1}$ allocation.
Suppose we fix agent $i=2$ when constructing $X$ and $Y$.
Then the EF$_{-2}^{+0}$ condition from $Y$ to $Z$ and from $Z$ to $Y$ are both satisfied for agent $1$, and the EF1 condition from $Y$ to $Z$ and from $Z$ to $Y$ are both satisfied for agent $1$.
We will only analyze agent $1$, as the analysis of agent $2$ is the same since EF1 implies EF$_{-2}^{+0}$.

We let $T_1=\frac12\cdot v_1(M^{II}\setminus X)$ be the average values of $Y$ and $Z$ for agent $1$. Assume $v_1(Y)\le v_1(Z)$ without loss of generality.
To show that the EF$^{+1}_{-1}$ condition is satisfied for agent $1$ for any permutation of $(X,Y,Z)$, we discuss the following three cases with respect to the bundle agent $1$ receives.

Suppose agent $1$ receives $X$.
If $v_1(X)\geq\frac13$, then $T_1\leq\frac13$.
By our requirements for $Y$ and $Z$, there is one bundle with value at most $T_1$, and removing at most one item from the other bundle makes its value at most $T_1$.
After removal, agent $1$ will not envy the agent who receives $Y$ or $Z$.
The EF1 property is satisfied.
If $v_1(X)<\frac13$, by our requirement, there exists $g_1^+\in M^{II}\setminus X$ such that $v_1(X\cup\{g_1^+\})\geq\frac13+\frac13\cdot v_1(g_1^+)$.
Moreover, $T_1=\frac12(1-v_1(X))\leq\frac12(1-(\frac13-\frac23\cdot v_1(g_1^+)))=\frac13+\frac13\cdot v_1(g_1^+)$.
Thus, $v_1(X\cup\{g_1^+\})\geq T_1$.
On the one hand, adding $g_1^+$ to $X$ makes the bundle's value weakly higher than $T_1$.
On the other hand, by our requirements for $Y$ and $Z$, removing at most one item from the bundles makes the values of the bundles at most $T_1$.
The EF$^{+1}_{-1}$ property is met.
 
Suppose agent $1$ receives $Y$.
The EF$^{+0}_{-2}$ property, which is stronger than EF$_{-1}^{+1}$ (Proposition~\ref{prop:EFvw}), is satisfied trivially from agent $1$ to the agent receiving $Z$ by our requirements for $Y$ and $Z$.
After adding at most one item, the value of the bundle $Y$ is at least $T_1$.
If $v_1(X)\leq\frac13$, then $T_1\geq\frac13$.
The EF$^{+1}_{-0}$ property, which is stronger than EF$^{+1}_{-1}$, is satisfied from agent $1$ to the agent receiving $X$.
If $v_1(X)>\frac13$, by our requirement, there exists $g_1^-\in X$ such that $v_1(X\setminus\{g_1^-\})\leq\frac13-\frac13\cdot v_1(g_1^-)$.
On the other hand, $T_1=\frac12(1-v_1(X))\geq\frac12(1-(\frac13+\frac23\cdot v_1(g_1^-)))=\frac13-\frac13\cdot v_1(g_1^-)$, so $v_1(X\setminus\{g_1^-\})\leq T_1$.
Since adding at most one item to $Y$ makes $Y$'s value at least $T_1$, the EF$^{+1}_{-1}$ property from agent $1$ to the agent receiving $X$ is satisfied.

Suppose agent $1$ receives $Z$, where $v_1(Z)\ge T_1$.
The envy-freeness from agent $1$ to the agent receiving $Y$ holds trivially, and the EF$1$ property from agent $1$ to the agent receiving $X$ holds as removing at most one item from $X$ makes the value of the bundle at most $T_1$ as we have seen above.

Now, it remains to construct $X$, $Y$, and $Z$ that satisfy the requirements.
In the following, we let $k$ be the number of Type II groups.

\subsubsection{Constructing $X$}
\label{sect:three-2X}
We first consider a fractional subset $X'$ of $M^{II}$ such that $v_1(X')=v_2(X')=\frac13$.
Let $X'=\{x_{j\ell}\}_{j=1,\ldots,k;\ell=1,2,3}$ where $x_{j\ell}$ denotes the fraction of the $\ell$-th item in the $j$-th group.
We formulate the requirement $v_1(X')=v_2(X')=\frac13$ by the following linear constraints.
\begin{align*}
    \sum_{j=1}^k\sum_{\ell=1}^3v_1(x_{j\ell})&=\frac13\\
    \sum_{j=1}^k\sum_{\ell=1}^3v_2(x_{j\ell})&=\frac13\\
    x_{j1}+x_{j2}+x_{j3}&=1\tag{for each $j=1,\ldots,k$}\\
    x_{j\ell}&\geq0\tag{for each $j=1,\ldots,k$ and each $\ell=1,2,3$}
\end{align*}
The first two constraints enforce $v_1(X')=v_2(X')=\frac13$.
The third set of constraints enforces that exactly one unit of item is included to $X'$ in each group.
Notice that if an integral solution exists, we have found a set $X$ that meets our requirement.
However, this is not always possible.
We will find a fractional solution $X'$ and then round it to an integral solution such that the first two constraints are ``slightly'' violated within an error of at most a $\frac23$ fraction of an item.

Firstly, notice that the solution space is non-empty, as setting $x_{j\ell}=\frac13$ for each $j=1,\ldots,k$ and each $\ell=1,2,3$ gives a valid solution.

Secondly, notice that the set of linear constraints defines a polytope in $\mathbb{R}^{3k}$.
We find a vertex of this polytope, which corresponds to a solution $X'$ where $3k$ constraints are tight.
We already know that the first three types of constraints, which have a total of $k+2$ constraints, are tight.
Thus, at least $2k-2$ constraints of the type $x_{j\ell}\geq0$ are tight.
This implies at least $2k-2$ items are not included in $X'$ at all.
Furthermore, the third set of constraints ensures that at least one item is (partially) included.
This implies $X'$ contains one integral item for at least $k-2$ groups.

We will use the fractional allocation $X'$ that corresponds to a vertex of the polytope and perform rounding in each of the following cases:
\begin{itemize}
    \item Case 1: $X'$ contains one integral item for $k-1$ groups. In this case, only one group contains fractional items.
    \item Case 2: $X'$ contains one integral item for $k-2$ groups. Since we have seen that at least $2k-2$ items are not included and at least one item from each group is included, in each of the two groups where fractional allocation occurs, exactly two (out of three) items are (partially) allocated. Let $a,b$ be the two items in one group, and $c,d$ be the two items in the other group. Assume without loss of generality that $v_1(a)\geq v_1(b)$ and $v_1(c)\geq v_1(d)$. We further consider the following three subcases:
    \begin{itemize}
        \item Case 2(a): $v_2(a)\geq v_2(b)$ and $v_2(c)\geq v_2(d)$;
        \item Case 2(b): $v_2(a)\leq v_2(b)$ and $v_2(c)\leq v_2(d)$;
        \item Case 2(c): $v_2(a)> v_2(b)$ and $v_2(c)<v_2(d)$; moreover, $v_1(a)>v_1(b)$ and $v_1(c)>v_1(d)$.
    \end{itemize}
\end{itemize}
Notice that Case 2(a), 2(b), and 2(c) cover all the possible cases (although 2(a) and 2(b) overlaps).
In particular, for Case 2(c), if $v_1(a)=v_1(b)$, renaming items $a$ and $b$ gives us Case 2(b); if $v_1(c)=v_1(d)$, renaming items $c$ and $d$ gives us Case 2(a).

Now we discuss the rounding for each case.

\paragraph{Rounding for Case 1.}
Let $\{a,b,c\}$ be the group where items are fractionally included in $X'$, and let $x_a,x_b,x_c$ be the corresponding fractions.
Assume $x_a\geq x_b\geq x_c$ without loss of generality.
The rounding is given by keeping all the integral items in $X'$ and including only $a$ from the group $\{a,b,c\}$.
We obtain a bundle $X$ where $v_1(X)=\frac13+(1-x_a)v_1(a)-x_bv_1(b)-x_cv_1(c)$.
Moreover, since $x_a\geq x_b\geq x_c$ and $x_a+x_b+x_c=1$, we have $x_a\geq\frac13$ and $x_b+x_c\leq\frac23$.
Therefore, we have $v_1(X)\geq\frac13-x_bv_1(b)-x_cv_1(c)\geq\frac13-(x_b+x_c)\cdot\max\{v_1(b),v_1(c)\}\geq\frac13-\frac23\cdot\max\{v_1(b),v_1(c)\}$, and $v_1(X)\leq\frac13+(1-x_a)v_1(a)\leq\frac13+\frac23\cdot v_1(a)$.
The two conditions for $X$ are satisfied for agent $1$, with $g_1^-=a$ and $g_1^+$ being the item in $\{b,c\}$ that has a higher value to agent $1$.
Similarly, the two conditions for $X$ are also satisfied for agent $2$.

\paragraph{Rounding for Case 2(a).}
Let $x_a$ and $x_b$ be the fractions of items $a$ and $b$ (with $x_a+x_b=1$), and let $x_c$ and $x_d$ be the fractions of items $c$ and $d$ (with $x_c+x_d=1$).
We keep all the integral items in $X'$, and the rounding rule for the four fractional items is specified below:
\begin{enumerate}
    \item if $x_a\geq\frac23$, keep item $a$ in $X$; if $x_b\geq\frac23$, keep item $b$ in $X$;
    \item if $x_c\geq\frac23$, keep item $c$ in $X$; if $x_d\geq\frac23$, keep item $d$ in $X$;
    \item if $x_a,x_b\in(\frac13,\frac23)$ and $x_c,x_d\notin(\frac13,\frac23)$, we include $c$ or $d$ based on (2); if $c$ is included, we include $b$; if $d$ is included, we include $a$;
    \item if $x_a,x_b\notin(\frac13,\frac23)$ and $x_c,x_d\in(\frac13,\frac23)$, we include $a$ or $b$ based on (1); if $a$ is included, we include $d$; if $b$ is included, we include $c$;
    \item if $x_a,x_b,x_c,x_d\in(\frac13,\frac23)$, we include $a$ and $d$ (in fact, including $b$ and $c$ also works here).
\end{enumerate}

Let us first consider $a$ and $b$.
If $a$ is included and $b$ is discarded, by the Case 2(a) assumption on the values of $a$ and $b$, the value change for agent $i\in\{1,2\}$ is $(1-x_a)v_i(a)-x_bv_i(b)\geq (1-x_a)v_i(a)-x_bv_i(a)=0$, which is an increment, and the value is increased by at most a $(1-x_a)$ fraction of the value of item $a$.
If $b$ is included and $a$ is discarded, the value change becomes $-x_av_i(a)+(1-x_b)v_i(b)\leq-x_av_i(a)+(1-x_b)v_i(a)=0$, which is a decrement, and the value is decreased by at most an $x_a$ fraction of the value of the item $a$.
The same observation can be made for items $c$ and $d$.

If both (1) and (2) are executed, the value change is bounded by a $1/3$ fraction of the value of item $a$ or $b$, plus or minus a $1/3$ fraction of the value of item $c$ or $d$. It is easy to check that the value of $X$ is about $1/3$ with the addition or removal of at most a $2/3$ fraction of the value of an item.
In the remaining case where at least one of (1) and (2) is not executed, the signs of the value changes for rounding $\{a,b\}$ and $\{c,d\}$ are opposite.
By our rounding rule, the adjustments in the fractions are all bounded by $2/3$. 

\paragraph{Rounding for Case 2(b).}
The rounding rule is exactly the same as it is in Case 2(a).
If both (1) and (2) are executed, the value change is again bounded by a $2/3$ fraction of the value of an item.
Otherwise, the same rule for Case 2(a) also guarantees that the signs of the value change for $\{a,b\}$ and $\{c,d\}$ are opposite for Case 2(b).

\paragraph{Rounding for Case 2(c).}
This case is difficult to handle.
Consider the following example:
\begin{itemize}
    \item $v_1(a)=v_1(c)=1$, $v_1(b)=v_1(d)=0$;
    \item $v_2(a)=v_2(d)=1$, $v_2(b)=v_2(c)=0$;
    \item $x_a=x_b=x_c=x_d=0.5$.
\end{itemize}
It can be easily checked that, in each possible rounding of $\{a,b\}$ and $\{c,d\}$, the value for one of the two agents is either increased or decreased by $1$.
Fortunately, we will show that, by moving from one vertex $X'$ of the polytope to one of its adjacent vertices $X''$, we must fall into one of the previous cases.
The property for Type II groups ensures the possibility of this.
Let $e$ be the third item in the group containing $c$ and $d$.
Since $v_1(c)>v_1(d)$ and $v_2(c)<v_2(d)$, by the fact that both agents have the same favorite item for Type II groups, we have $v_1(e)>v_1(c)>v_1(d)$ and $v_2(e)>v_2(d)>v_2(c)$.
We can move to the adjacent vertex with $x_e=0$ relaxed and with $x_a,x_b,x_c$, and $x_d$ adjusted.

We consider a slight adjustment $\delta\in (-\varepsilon,\varepsilon)$ in $x_a$ and $x_b$: $x_a+\delta$ and $x_b-\delta$.
Let $t_1$ be the change of utility for agent $1$ and $t_2$ be the change of utility for agent $2$ due to this adjustment.
Since $v_1(a)>v_1(b)$ and $v_2(a)>v_2(b)$, the changes $t_1$ and $t_2$ have the same sign.

Next, we adjust the values of $x_c,x_d$, and $x_e$ such that the value change for both agents is $-t_1$ and $-t_2$ respectively.
Let $\delta_c,\delta_d,\delta_e$ be the adjustments for $x_c,x_d$, and $x_e$ respectively.
To make sure a total of $1$ unit of item is allocated for the group $\{c,d,e\}$ and both agents believe the fractional bundle is still worth $1/3$, we must have
$$\left[\begin{array}{ccc}
    1 & 1 & 1\\
    v_1(c) & v_1(d) & v_1(e)\\
    v_2(c) & v_2(d) & v_2(e)
\end{array}\right]\cdot
\left[
\begin{array}{c}
     \delta_c  \\
     \delta_d \\
     \delta_e
\end{array}\right]=\left[
\begin{array}{c}
     0  \\
     -t_1 \\
     -t_2
\end{array}\right].$$
To ensure that we are moving to another vertex of the polytope, we must ensure
\begin{enumerate}
    \item The system of linear equations has a valid solution;
    \item The solution must satisfy $\delta_e>0$.
\end{enumerate}
For (1), we can check that the determinant of the $3\times 3$ matrix is non-zero.
This is guaranteed by $v_1(e)>v_1(c)>v_1(d)$ and $v_2(e)>v_2(d)>v_2(c)$, and the detailed calculations are left to the readers.
For (2), we apply Cramer's Rule,
$$\delta_e=\frac{\det\left(\left[\begin{array}{ccc}
    1 & 1 & 0\\
    v_1(c) & v_1(d) & -t_1\\
    v_2(c) & v_2(d) & -t_2
\end{array}\right]\right)}{\det\left(\left[\begin{array}{ccc}
    1 & 1 & 1\\
    v_1(c) & v_1(d) & v_1(e)\\
    v_2(c) & v_2(d) & v_2(e)
\end{array}\right]\right)}.$$
By the fact that $t_1$ and $t_2$ have the same sign, it can be checked that we can properly choose the sign of $\delta$ (the adjustments for $x_a$ and $x_b$), which decides the sign of $t_1$ and $t_2$, such that $\delta_e>0$.

As $|\delta|$ increases, one of the constraints $x_a\geq 0, x_b\geq 0,x_c\geq 0$, or $x_d\geq0$ becomes tight.
If the constraint for $x_a$ or $x_b$ first becomes tight, we move to a vertex corresponding to $X''$ that belongs to Case 1.
If the constraint for $x_c$ or $x_d$ first becomes tight, we move to a vertex corresponding to $X''$ that belongs to Case 2(a) (notice that both agents value $e$ strictly higher than any of $c$ or $d$).
We then apply the corresponding rounding rule.

This concludes the construction of $X$.
Notice that all the above operations can be done in polynomial time by standard linear programming methods.

\subsubsection{Constructing $Y$ and $Z$}\label{sect:three-2YZ}

We apply a similar method as in Sect.~\ref{sect:three-1} where we consider the allocations with two bundles.
Notice that $X$ contains exactly one item from each group.
Each group then contains exactly two remaining items.
For each group $G=\{a, b\}$, assume without loss of generality that $v_1(a)\ge v_1(b)$.

We first consider all the groups such that $v_2(a)\ge v_2(b)$, i.e., groups where both agents have the same preference order over the two items.
Given a group $G=\{a, b\}$, let $\Delta_1(G)=v_1(a)-v_1(b)$ and $\Delta_2(G)=v_2(a)-v_2(b)$.
Let both agents $i=1,2$ sort the groups by the descending order of $\Delta_i(\cdot)$, resulting in $G_1^{(i)}, \ldots, G_{k'}^{(i)}$ where $k'$ is the number of the groups.
Define $H_j^{(i)}=\{G_{2j-1}^{(i)}, G_{2j}^{(i)}\}$ for each $j=1,\ldots, k'/2$.
Consider a partition $(A_1, A_2)$, where for each $H_j^{(i)}$ containing two groups, item $a$ from exactly one group and item $b$ from the other group are allocated to $A_1$, and the other items are allocated to $A_2$.
Then both allocations $(A_1, A_2)$ and $(A_2, A_1)$ are EF1 to both agents.
The analysis is similar to Proposition~\ref{prop:typeIEF1}.
Such a partition exists similar to the analysis of Proposition~\ref{prop:typeI3coloring}.

For the groups such that $v_2(a)<v_2(b)$ (those groups where both agents have the opposite preference orders over the two items), define $\Delta'_1(G)=v_1(a)-v_1(b)$ and $\Delta'_2(G)=v_2(b)-v_2(a)$.
Similar to the above, we may obtain a partition $(A'_1, A'_2)$ such that both allocations $(A'_1, A'_2)$ and $(A'_2, A'_1)$ are EF1 to both agents.
The details are left to the readers.

As both partitions are EF1, each of the four allocations $(A_1\cup A'_1, A_2\cup A'_2), (A_1\cup A'_2, A_2\cup A'_1), (A_2\cup A'_1, A_1\cup A'_2),$ and $(A_2\cup A'_2, A_1\cup A'_1)$ is EF$_{-2}^{+0}$ to both agents.
Finally, for the fixed agent $i\in\{1,2\}$ where the stronger EF1 requirement must be met, assume without loss of generality that $v_i(A_1)\ge v_i(A_2)$ and $v_i(A'_1)\ge v_i(A'_2)$.
Let $Y=A_1\cup A'_2$ and $Z=A_2\cup A'_1$.
For agent $i$, $v_i(Y)\ge v_i(Z\setminus \{g'\})$ where $g'\in A'_1$ is the item such that $v_i(A'_2)\ge v_i(A'_1\setminus\{g'\})$, and $v_i(Z)\ge v_i(Y\setminus \{g\})$ where $g\in A_1$ is the item such that $v_i(A_2)\ge v_i(A_1\setminus\{g\})$.
Therefore, we obtain a partition $(Y, Z)$ such that both allocations $(Y, Z)$ and $(Z, Y)$ are EF1 to agent $i$ and EF$_{-2}^{+0}$ to agent $3-i$.

In Sect.~\ref{sect:three-comb}, we will specify the agent $i$ we choose to ensure the EF1 property between $Y$ and $Z$.
As a remark, if our eventual goal is EF$^{+1}_{-2}$, the choice of $i$ can be arbitrary; moreover, the readers can verify that ensuring both $(X,Y)$ and $(Y,X)$ are EF$^{+0}_{-2}$ for both agent $1$ and $2$ would suffice to guarantee that any permutation of $(X,Y,Z)$ is EF$^{+1}_{-1}$.

\subsection{Achieving EF$^{+1}_{-1}$: A More Careful Combination of Type I and Type II Allocations}\label{sect:three-comb}
By now, we have obtained a truthful and EF$_{-2}^{+1}$ mechanism.
We can first arbitrarily choose three allocations of Type I that satisfy our requirement in Sect.~\ref{sect:three-1} from the three-partition of groups $(P,Q,R)$.
Each allocation satisfies regularity and EF$1$.
We then arbitrarily choose three regular allocations of Type II from the three-partition of items $(X,Y,Z)$ in Sect.~\ref{sect:three-2}, where we can guarantee the EF$1$ conditions from $Y$ to $Z$ and from $Z$ to $Y$ hold either for agent $1$ or for agent $2$.
Each allocation above is assigned probability $1/3$.
This has already guaranteed a truthful and EF$_{-2}^{+1}$ mechanism, as any combination of these allocations that preserves the marginal probabilities satisfies regular and EF$_{-2}^{+1}$.

In this section, we will make further improvements to achieve EF$^{+1}_{-1}$. To achieve this, we will combine the three allocations for Type I groups and the three allocations for Type II groups in a more careful way.
We will violate the regularity constraint (defined in Sect.~\ref{sect:three-agent3}) by a little bit, yet the marginal probability that each agent receives each item still follows Table~\ref{tab:ruleF}.
Such violations will weaken the fairness guarantee for agent $3$: instead of EF1, now we have EF$_{-1}^{+1}$.
However, allowing such violations enables more flexibility so that better fairness guarantees can be obtained for agents $1$ and $2$.

For high-level intuitions, if the allocation of Type II is EF$1$ or EF$_{-0}^{+1}$ to an agent, then the allocation of Type I can also be EF$1$ to her, and the combination of the two allocations will satisfy EF$_{-1}^{+1}$.
On the other hand, if the allocation of Type II is neither EF$1$ nor EF$_{-0}^{+1}$, but EF$_{-1}^{+1}$ to an agent, then the allocation of Type I should be envy-free to her so that the combination maintains EF$_{-1}^{+1}$.
This can be seen from Proposition~\ref{prop:combine} whose proof is trivial.

\begin{proposition}\label{prop:combine}
    Let $(A_1,A_2,A_3)$ and $(A_1',A_2',A_3')$ be two allocations of two disjoint sets of items.
    Then the allocation $(A_1\cup A_1', A_2\cup A_2', A_3\cup A_3')$ is EF$_{-1}^{+1}$ when either condition holds:
    \begin{itemize}
        \item $(A_1,A_2,A_3)$ is EF$1$ and $(A_1',A_2',A_3')$ is EF$_{-0}^{+1}$;
        \item $(A_1,A_2,A_3)$ is envy-free and $(A_1',A_2',A_3')$ is EF$_{-1}^{+1}$.
    \end{itemize}
    In addition, the allocation $(A_1\cup A_1', A_2\cup A_2', A_3\cup A_3')$ is EF$_{-1}^{+1}$ for a particular agent $i$ if one of the above two conditions holds for agent $i$.
\end{proposition}

The analysis begins with the Type I allocations.
In Sect.~\ref{sect:three-1}, we partition Type I groups into three sets $(P,Q,R)$ such that each set of groups will be allocated using a different allocation rule among $\chi_s, \chi_1$, and $\chi_2$.
In other words, the three sets correspond to the groups with three colors in Proposition~\ref{prop:typeI3coloring}.
For each set $S\in\{P,Q,R\}$, let $\gamma_1(S)=\sum_{G\in S} (v_1(a_1^G)-v_1(b^G))$ and $\gamma_2(S)=\sum_{G\in S} (v_2(a_2^G)-v_2(b^G))$, where $a_1^G$ is agent $1$'s favorite item in group $G$, $a_2^G$ is agent $2$'s favorite item, and $b^G$ is the remaining item.
Proposition~\ref{prop:typeIEF} gives a sufficient condition for an agent not to envy another agent.
\begin{proposition}\label{prop:typeIEF}
    Let agent $i\in\{1,2\}$ sort the three sets $P,Q$, and $R$ by the descending values of $\gamma_i(\cdot)$.
    If the set allocated using $\chi_s$ is before the set allocated using $\chi_{3-i}$ along the order, then agent $i$ will not envy agent $3$.
    If the set allocated using $\chi_i$ is before the set allocated using $\chi_{3-i}$, then agent $i$ will not envy agent $3-i$.
\end{proposition}
\begin{proof}
    We will prove the proposition for $i=1$ only, as the proof for $i=2$ is symmetric.
    If the set of groups $S\in\{P,Q,R\}$ is allocated using the allocation rule $\chi_2$, agent $1$ will envy agent $3$ by an amount of $\gamma_1(S)$, and agent $1$ will envy agent $2$ by an amount of $\sum_{G\in S}(v_1(a_2^G)-v_1(b^G))$, which is at most $\gamma_1(S)$ as $v_1(a_1^G)\geq v_1(a_2^G)$.
    If the set $S'$ is allocated using $\chi_s$, agent $1$ will receive a bundle with a value higher than that of agent $3$ by an amount of $\gamma_1(S')$, and agent $1$ will also receive a bundle with a value higher than that of agent $2$ (as agent $1$ receives her favorite item from each group).
    If the set $S''$ is allocated using $\chi_1$, agent $1$ receives a bundle with a value higher than the bundle of agent $2$ by an amount of $\gamma_1(S'')$, and agent $1$ will not envy agent $3$ (as agent $1$ will receive her favorite item from each group).
    If the set $S'$ comes before $S$ in the descending order of $\gamma_1(\cdot)$, we have $\gamma_1(S')\ge \gamma_1(S)$, which implies the positive contribution of $\chi_s$ will offset the negative contribution of $\chi_2$.
    Therefore, the envy-freeness from agent $1$ to agent $3$ is guaranteed.
    Similarly, if the set $S''$ comes before $S$, agent $1$ will not envy agent $2$.
\end{proof}

We next consider the partition $(X,Y,Z)$ of Type II in Sect.~\ref{sect:three-2}.
Recall that in a Type II allocation (which is a permutation of the partition), EF$_{-1}^{+1}$ relations between two of the three bundles, instead of EF$1$ or EF$_{-0}^{+1}$, may occur for agent $i\in\{1,2\}$. 
For an agent $i\in\{1,2,3\}$ and a bundle $S\in\{X,Y,Z\}$, if there exists another bundle $S'\in\{X,Y,Z\}$ such that the EF$_{-1}^{+1}$ property from $S$ to $S'$ holds for agent $i$ while both the EF1 and EF$_{-0}^{+1}$ properties fail, we say that $S$ is an \emph{unwanted} bundle for agent $i$.
For a partition specified in Sect.~\ref{sect:three-2}, there are at most two unwanted bundles and two EF$_{-1}^{+1}$ (except for EF$1$ and EF$_{-0}^{+1}$) relations for an agent $i\in\{1,2\}$, at most one unwanted bundle and one EF$_{-1}^{+1}$ relation for the other agent $3-i$, and any permutation is EF$1$ to agent $3$ (see the EF$_{-1}^{+1}$ analysis of the partition $(X,Y,Z)$ for more details).

We are now ready to combine the Type I and Type II allocations.
Recall in Sect.~\ref{sect:three-2} that we construct $(X,Y,Z)$ for Type II allocation by first constructing $X$, then choosing an agent $i$ such that $(Y,Z)$ and $(Z,Y)$ are EF1, and finally constructing $Y$ and $Z$.
Suppose now we have obtained the Type I partition $(P,Q,R)$ from Sect.~\ref{sect:three-1} and Type II bundle $X$ from Sect.~\ref{sect:three-2X}.
We will demonstrate how to choose agent $i$ (which affects the construction of $Y$ and $Z$) and how to choose a proper Type I allocation combined with the specific Type II allocation while maintaining the marginal probability of each item.

We consider the following two cases.
\begin{itemize}
    \item Case i: There is at least one agent $i\in\{1,2\}$ whose value to $X$ is no less than $\frac13$.
    \item Case ii: Both agent $1$'s and agent $2$'s values to $X$ are less than $\frac13$.
\end{itemize}

\subsubsection{Combination of Case i} 
By symmetry, assume that $v_1(X)\ge\frac13$.
When constructing $Y$ and $Z$ in Sect.~\ref{sect:three-2YZ}, we guarantee that the EF$1$ conditions from $Y$ to $Z$ and from $Z$ to $Y$ are satisfied for agent $2$.
In such a partition $(X,Y,Z)$, there is at most one unwanted bundle for agent $1$, which is the bundle with a smaller value between $Y$ and $Z$.
For agent $2$, if $v_2(X)\ge\frac13$, the only possible unwanted bundle is the one with a smaller value between $Y$ and $Z$.
If $v_2(X)<\frac13$, the only possible unwanted bundle is $X$.
Thus, each of the agents $1$ and $2$ has at most one unwanted bundle. 

Assume that both agent $1$ and agent $2$ have one unwanted bundle.
We choose three allocations out of the six permutations of $(X,Y,Z)$, where we ensure that only agent $1$ receives her unwanted bundle in the first allocation, only agent $2$ receives her unwanted bundle in the second allocation, and each bundle is allocated to each agent once among the three allocations.
This can be achieved as each agent only has one unwanted bundle.
For example, if the unwanted bundle is $Y$ for agent $1$ and $Z$ for agent $2$, the three allocations can be $(Y,X,Z),(X,Z,Y)$, and $(Z,Y,X)$; if bundle $Y$ is unwanted for both agents, the three allocations can be $(Y,X,Z),(Z,Y,X)$, and $(X,Z,Y)$.

Let the two agents $i\in\{1,2\}$ sort the three sets of groups $(P,Q,R)$ of Type $I$ according to descending order of $\gamma_i(\cdot)$.
For the first allocation of Type II where only agent $1$ receives her unwanted bundle, let the set of Type I with the smallest $\gamma_1(\cdot)$ be allocated using $\chi_2$.
By Proposition~\ref{prop:typeIEF}, the allocation of Type I is envy-free to agent $1$ no matter which allocation rule each of the remaining two sets is allocated using.
By Proposition~\ref{prop:combine}, the combination of the two allocations is EF$_{-1}^{+1}$ to the three agents.
The second allocation of Type II is handled similarly, where we let the set of Type I with the smallest $\gamma_2(\cdot)$ be allocated using $\chi_1$.
The sets for which we do not specify the allocation rules can be allocated arbitrarily, as long as the three allocations of Type I satisfy: each allocation rule is used exactly once in each Type I allocation, and each set is allocated using each allocation rule exactly once among the three allocations.
It can be easily verified that such requirements of assigning the allocation rules can always be met, and each combination we obtain is EF$_{-1}^{+1}$.

In addition, if we assign probability $1/3$ to each combination, the marginal probabilities for Type I and Type II allocations are satisfied respectively.

Note that EF$_{-1}^{+1}$ combination is easier to achieve when the fairness guarantee for the Type II partition is stronger (for example, only one agent has up to one unwanted bundle).
Therefore, the above analysis is enough to guarantee EF$_{-1}^{+1}$ combination for Case i.

\subsubsection{Combination of Case ii} 
In this case, we guarantee that the EF$1$ conditions from $Y$ to $Z$ and from $Z$ to $Y$ are satisfied for agent $2$, so agent $2$ will have at most one unwanted bundle---which can only be $X$---with the EF$_{-1}^{+1}$ relation from $X$ to the bundle with a higher value between $Y$ and $Z$.
If agent $1$ also has at most one unwanted bundle, it can be handled the same as in Case i.
Thus, we will only focus on the scenario where agent $1$ has two unwanted bundles.

Assume without loss of generality that $v_1(Y)\le v_1(Z)$.
Then, both $X$ and $Y$ are agent $1$'s unwanted bundles.
It is guaranteed by the Type II partition that the EF$_{-1}^{+1}$ relations for agent $1$ are only from $X$ to $Z$ and from $Y$ to $Z$.
In the following, for Type I partition $(P,Q,R)$, we will assume that $\gamma_1(P)\ge \gamma_1(Q)\ge \gamma_1(R)$ without loss of generality.
We will discuss the following three cases.

\paragraph{Agent $2$ has no unwanted bundle.}
In this case, when arbitrarily choosing three Type I allocations and three Type II allocations such that the requirements for the two types of allocations are satisfied, any combination is EF$_{-1}^{+1}$ for agent $2$ according to Proposition~\ref{prop:combine}, and EF$1$ for agent $3$ by regularity.

To guarantee EF$_{-1}^{+1}$ for agent $1$, we let the three Type II allocations be $(X,Y,Z),(Y,Z,X)$, and $(Z,X,Y)$, which are combined with $(P_{\chi_s},Q_{\chi_1},R_{\chi_2}), (P_{\chi_1},Q_{\chi_2},R_{\chi_s})$, and $(P_{\chi_2},Q_{\chi_s},R_{\chi_1})$ respectively,
where the allocation $(P_{\chi_s},Q_{\chi_1},R_{\chi_2})$ denotes that the three sets of groups $P,Q$, and $R$ of Type I are allocated using $\chi_s,\chi_1$, and $\chi_2$ respectively.
Each of the combinations is assigned probability $1/3$, so that the marginal probability of each item is satisfied.
Among the three allocations, whenever there is EF$_{-1}^{+1}$ relation instead of EF$1$ or EF$_{-0}^{+1}$ from agent $1$'s bundle to another agent's bundle in the Type II allocation, agent $1$ will not envy that agent in the Type I allocation. The following paragraph shows this.

In the allocation $(X,Y,Z)$ for Type II, agent $1$ receives the unwanted bundle $X$, and the EF$_{-1}^{+1}$ condition holds from agent $1$'s bundle $X$ to agent $3$'s bundle $Z$; in the matched allocation $(P_{\chi_s},Q_{\chi_1},R_{\chi_2})$ for Type I, Proposition~\ref{prop:typeIEF} implies that agent $1$ does not envy agent $3$, so Proposition~\ref{prop:combine} implies the allocation combined by $(X,Y,Z)$ and $(P_{\chi_s},Q_{\chi_1},R_{\chi_2})$ is EF$_{-1}^{+1}$ for agent $1$.
Similarly, in the allocation $(Y,Z,X)$ for Type II, agent $1$ receives the unwanted bundle $Y$, and the EF$_{-1}^{+1}$ condition holds from agent $1$'s bundle $Y$ to agent $2$'s bundle $Z$; the matched allocation $(P_{\chi_1},Q_{\chi_2},R_{\chi_s})$ for Type I ensures agent $1$ does not envy agent $2$, and Proposition~\ref{prop:combine} implies the allocation combined by $(Y,Z,X)$ and $(P_{\chi_1},Q_{\chi_2},R_{\chi_s})$ is EF$_{-1}^{+1}$ for agent $1$.
Finally, agent $1$ does not receive an unwanted bundle in the Type II allocation $(Z,X,Y)$, so the allocation $(Z,X,Y)$ is EF$_{-0}^{+1}$ for agent $1$, since $(P_{\chi_2},Q_{\chi_s},R_{\chi_1})$ is EF1, the allocation combined by $(Z,X,Y)$ and $(P_{\chi_2},Q_{\chi_s},R_{\chi_1})$ is EF$_{-1}^{+1}$ for agent $1$ by Proposition~\ref{prop:combine}.

\paragraph{Agent $2$ has unwanted bundle $X$ and $v_2(Y)\ge v_2(Z)$.}
The EF$_{-1}^{+1}$ relation for agent $2$ is from $X$ to $Y$.
The combination is shown in Table~\ref{tab:three-comb1}, which is interpreted as: we obtain three combinations where each column of the table denotes a combination, which is assigned probability $\frac13$.
Taking the first column as an example, when Type II allocation is $(Y,X,Z)$ and when $\gamma_2(Q)\ge\gamma_2(R)$, the three sets of groups $P,Q$, and $R$ of Type I will be allocated using $\chi_s,\chi_2$, and $\chi_1$ respectively.
The marginal probability of each item is satisfied.
In the allocation $(Y,X,Z)$, the relation for agent $1$ from $Y$ to $Z$ is not EF$1$ or EF$_{-0}^{+1}$, but EF$_{-1}^{+1}$, and the relation for agent $2$ from $X$ to $Y$ is EF$_{-1}^{+1}$.
In the corresponding Type I allocation, agent $1$ will not envy agent $3$, and agent $2$ will not envy agent $1$ according to Proposition~\ref{prop:typeIEF}.
Similarly, for the second column, in $(X,Z,Y)$, the relation for agent $1$ from $X$ to $Z$ is EF$_{-1}^{+1}$, and in the corresponding Type I allocation, agent $1$ will not envy agent $2$.
For the third column, agent $1$ and $2$ do not receive any of their unwanted bundles, thus the allocation $(Z,Y,X)$ is EF$_{-0}^{+1}$ (or EF1, which is stronger than EF$_{-0}^{+1}$ by Proposition~\ref{prop:EFvw}); for the Type I part, the two allocations are EF1 (Proposition~\ref{prop:typeIEF1}).
Therefore, the combination in each column satisfies EF$_{-1}^{+1}$ according to Proposition~\ref{prop:combine}.

\begin{table}[h]
    \centering
    \begin{tabular}{cccc}
    \hline
        Type II allocation & $(Y,X,Z)$ & $(X,Z,Y)$ & $(Z,Y,X)$ \\
    \hline 
        Type I allocation when $\gamma_2(Q)\ge\gamma_2(R)$ & $(P_{\chi_s}, Q_{\chi_2}, R_{\chi_1})$ & $(P_{\chi_1}, Q_{\chi_s}, R_{\chi_2})$ & $(P_{\chi_2}, Q_{\chi_1}, R_{\chi_s})$ \\[0.5ex]
        Type I allocation when $\gamma_2(Q)<\gamma_2(R)$ & $(P_{\chi_s}, Q_{\chi_1}, R_{\chi_2})$ & $(P_{\chi_1}, Q_{\chi_2}, R_{\chi_s})$ & $(P_{\chi_2}, Q_{\chi_s}, R_{\chi_1})$ \\[0.3ex]
    \hline 
    \end{tabular}
    \caption{Combinations when $X$ is agent $2$'s unwanted bundle and $v_2(Y)\ge v_2(Z)$.}
    \label{tab:three-comb1}
\end{table}

\paragraph{Agent $2$ has unwanted bundle $X$ and $v_2(Y)< v_2(Z)$.}
The EF$_{-1}^{+1}$ relation for agent $2$ is from $X$ to $Z$.
The combination is shown in Table~\ref{tab:three-comb2}, where $g\in Z$ is the item in bundle $Z$ with the highest value to agent $1$.
Each combination is assigned probability $\frac13$.
Different from the above allocations which are regular, regularity for Type II allocations is slightly violated.
However, it can be easily checked that each item of Type II is still allocated to each agent with probability $1/3$, which still matches the marginal probability of the fractional allocation in Sect.~\ref{sect:three-fractional}.

For agent $3$, the fairness guarantee is weakened from EF$1$ to EF$_{-1}^{+1}$ due to the violation of regularity.
Take the allocation combined by $(X\cup\{g\}, Z\setminus\{g\}, Y)$ and $(P_{\chi_s}, Q_{\chi_2}, R_{\chi_1})$ for an example.
As each of the two allocations $(X,Z,Y)$ and $(P_{\chi_s}, Q_{\chi_2}, R_{\chi_1})$ are regular, the combination $(A_1,A_2,A_3)$ by the two allocations $(X,Z,Y)$ and $(P_{\chi_s}, Q_{\chi_2}, R_{\chi_1})$ is regular and thus satisfies EF$1$ for agent $3$ according to Proposition~\ref{prop:regular}.
Therefore, there exists two items $g_1\in A_1$ and $g_2\in A_2$ such that $v_3(A_3)\ge v_3(A_1\setminus\{g_1\})$ and $v_3(A_3)\ge v_3(A_2\setminus\{g_2\})$.
For the allocation $(A_1',A_2',A_3')=(A_1\cup\{g\},A_2\setminus\{g\},A_3)$ combined by $(X\cup\{g\}, Z\setminus\{g\}, Y)$ and $(P_{\chi_s}, Q_{\chi_2}, R_{\chi_1})$, we have $v_3(A'_3)\ge v_3(A'_1\setminus\{g_1,g\})$ and $v_3(A'_3)\ge v_3(A'_2\setminus\{g_2\})$, which guarantees EF$_{-1}^{+1}$.
In the following, we will only consider the fairness guarantee for agent $1$ and agent $2$.

In the Type II allocation $(Y,X,Z)$, both the relations are EF$_{-1}^{+1}$ for agent $1$ from $Y$ to $Z$ and for agent $2$ from $X$ to $Z$.
In $(X\cup\{g\},Z\setminus\{g\},Y)$, EF$_{-0}^{+1}$ relation for agent $1$ is satisfied from $X\cup\{g\}$ to $Y$ guaranteed by the partition, and EF$_{-0}^{+1}$ to $Z\setminus\{g\}$ as the item with the highest value in $Z$ has already been removed.
EF$_{-0}^{+1}$ relation is satisfied for agent $2$ from $Z\setminus\{g\}$ to $Y$ as adding item $g$ back will result in $v_2(Z)> v_2(Y)$.
As $v_2(Z\setminus\{g'\})>v_2(X)$ for any item $g'\in Z$ (otherwise $X$ is no longer agent $2$'s unwanted bundle), we have $v_2(X\cup\{g\})<v_2(Z)$, thus the EF$_{-0}^{+1}$ relation for agent $2$ is satisfied from $Z\setminus\{g\}$ to $X\cup\{g\}$.
In $(Z\setminus\{g\},Y\cup\{g\},X)$, agent $1$ does not envy agent $3$ as the relation from $X$ to $Z$ is not EF$1$, and will not envy agent $2$ by removing $g$ from $Y\cup\{g\}$.
The allocation is EF$_{-0}^{+1}$ to agent $2$.

We can apply similar analyses in the previous case by using Proposition~\ref{prop:combine} and Proposition~\ref{prop:typeIEF} to verify that the combination in each column satisfies EF$_{-1}^{+1}$ for both agent $1$ and agent $2$.
We have also shown earlier that breaking regularity by moving one item $g$ still makes the allocation EF$_{-1}^{+1}$ for agent $3$.

\begin{table}[h]
    \centering
    \begin{tabular}{cccc}
    \hline
        Type II allocation & $(Y,X,Z)$ & $(X\cup\{g\},Z\setminus\{g\},Y)$ & $(Z\setminus\{g\},Y\cup\{g\},X)$ \\
    \hline 
        Type I allocation when $\min\gamma_2(\cdot)=\gamma_2(P)$ & $(P_{\chi_1}, Q_{\chi_s}, R_{\chi_2})$ & $(P_{\chi_s}, Q_{\chi_2}, R_{\chi_1})$ & $(P_{\chi_2}, Q_{\chi_1}, R_{\chi_s})$ \\[0.5ex]
        Type I allocation when $\min\gamma_2(\cdot)=\gamma_2(Q)$ & $(P_{\chi_s}, Q_{\chi_1}, R_{\chi_2})$ & $(P_{\chi_2}, Q_{\chi_s}, R_{\chi_1})$ & $(P_{\chi_1}, Q_{\chi_2}, R_{\chi_s})$ \\[0.5ex]
        Type I allocation when $\min\gamma_2(\cdot)=\gamma_2(R)$ & $(P_{\chi_s}, Q_{\chi_2}, R_{\chi_1})$ & $(P_{\chi_2}, Q_{\chi_1}, R_{\chi_s})$ & $(P_{\chi_1}, Q_{\chi_s}, R_{\chi_2})$ \\[0.3ex]
    \hline 
    \end{tabular}
    \caption{Combinations when $X$ is agent $2$'s unwanted bundle and $v_2(Y)<v_2(Z)$. $g=\argmax_{g'\in Z}v_1(g')$.}
    \label{tab:three-comb2}
\end{table}

\section{Discussions on Mechanisms for $n$ Agents}
\label{sect:nagents}
Designing truthful and almost envy-free mechanisms for a general numbers of agents seems to be significantly more challenging.
We first state a preliminary result for $n$ agents based on~\citet{manurangsi2022almost}: for $n$ agents, there exists a polynomial time randomized truthful mechanism that achieves the EF$^{+0}_{-O(\sqrt{n})}$.
This already demonstrates that randomized mechanisms can do significantly better than deterministic mechanisms (comparing the following theorem with Theorem~\ref{thm:impossibility_deterministic}).
The equal division rule is sufficient for the following result.

\begin{theorem}
    There exists a truthful and EF$^{+0}_{-O(\sqrt{n})}$ randomized mechanism $(\calF_=,\calD)$ for $n$ agents. 
    In addition, the mechanism samples an allocation in polynomial time.
\end{theorem}
\begin{proof}
    The result follows from~\citet{manurangsi2022almost}, which demonstrates, through discrepancy theory, that there exists a consensus $1/n$ division up to $O(\sqrt{n})$ goods and can be computed in polynomial time.
    Specifically, this refers to a partition $(X_1,\ldots,X_n)$ of $M$ such that for each agent $i$ and each pair of bundles $X_j, X_{j'}$, the condition $v_i(X_j)\ge v_i(X_{j'}\setminus S)$ holds, where $S\subseteq X_{j'}$ and $|S|=O(\sqrt{n})$.
    By assigning each bundle to each agent with probability $1/n$, the fractional allocation rule becomes the equal division rule, thus ensuring truthfulness.
\end{proof}

\subsection*{Truthful Mechanisms with Share-Based Fairness Criteria}
\label{sec:sharebased}
Other than envy-freeness and its relaxations, another line of fairness criteria is \emph{share-based}, which defines a threshold for each agent and requires each agent's utility to be weakly larger than the threshold.
Examples include \emph{proportionality} and \emph{maximin share}.
With indivisible items, the exact versions of these two criteria are not always satisfiable.
Common relaxations include \emph{proportionality up to one item} (denoted by PROP1) and an approximation version of maximin share (denoted by $\alpha$-MMS).
\begin{definition}
    Given a valuation profile $(v_1,\ldots,v_n)$, an allocation $\calA=(A_1,\ldots,A_n)$ is \emph{proportional up to one item} (PROP1) if, for each $i\in N$, there exists $g\notin A_i$ such that $v_i(A_i\cup\{g\})\geq\frac1n\cdot v_i(M)$.
\end{definition}

\begin{definition}
    Given a valuation profile $(v_1,\ldots,v_n)$, the \emph{maximin share} of an agent $i$, denoted by $\text{MMS}_i$ is the value of the least-preferred bundle in the optimal $n$-partition of the item set:
    $$\text{MMS}_i=\max_{\mathcal{X}=(X_1,\ldots,X_n)}\min_{t=1,\ldots,n}v_i(X_t).$$
    Given $\alpha\in(0,1)$, an allocation $\calA=(A_1,\ldots,A_n)$ satisfies \emph{$\alpha$-MMS} if $v_i(A_i)\geq\alpha\cdot\text{MMS}_i$ for each agent $i$.
\end{definition}

PROP1 allocations always exist.
For maximin share, the current state-of-art is that $\alpha$-MMS allocations always exist for some $\alpha$ that is slightly larger than $3/4$~\cite{akrami2024breaking}.
However, when enforcing truthfulness, PROP1 cannot be achieved by deterministic mechanisms, and the best achievable approximation ratio for MMS is $\frac1{\lfloor m/2\rfloor}$ even for two agents~\cite{amanatidis2016truthful,amanatidis2017truthful}.

For randomized mechanisms, \citet{AzizFrSh23} show that the equal division rule is PROP1-realizable.
Since the equal division rule is truthful, we can conclude that PROP1 can be achieved by randomized truthful mechanisms, as remarked by~\citet{babaioff2022best}.
We show that we can simultaneously achieve PROP1 and $\frac1n$-MMS with truthful randomized mechanisms.
The proof applies some similar techniques used in~\citet{AzizFrSh23}, and the equal division rule is sufficient.

\begin{theorem}
    There exists a randomized truthful mechanism $(\calF_=,\calD)$ that simultaneously achieves PROP1 and $\frac1n$-MMS.
\end{theorem}
\begin{proof}
    We begin by defining the decomposition rule $\calD$.

Assume without loss of generality that $m$ is a multiple of $n$ (by adding dummy items with value $0$ to all agents).
For each agent $i$, sort the items in the descending order of $v_i(\cdot)$: $g_1^{(i)},\ldots,g_m^{(i)}$.
Let $G^{(i)}_j=\{g_{nj-n+1}^{(i)},g_{nj-n+2}^{(i)},\ldots,g_{nj}^{(i)}\}$ for each $j=1,\ldots,m/n$.
Construct an $m\times m$ bipartite graph $G=(V_1,V_2,E)$ where $V_1$ corresponds to the $m$ groups $\{G^{(i)}_j\}_{i=1,\ldots,n;j=1,\ldots,m/n}$, $V_2$ corresponds to the $m$ items, and there is an edge between a vertex in $V_1$ representing group $G^{(i)}_j$ and a vertex in $V_2$ representing item $g$ if and only if $g\in G^{(i)}_j$.
This is an $n$-regular bipartite graph.
By Lemma~\ref{lem:matching}, $G$ can be decomposed to $n$ disjoint matchings $\mathcal{M}_1,\ldots,\mathcal{M}_n$.
Each matching $\mathcal{M}_k$ defines an allocation $\calA_k$: if an item $g$ is matched with a vertex $G^{(i)}_j$ in the graph, then item $g$ is allocated to agent $i$.
The rule $\calD$ outputs each of the $n$ allocations with probability $1/n$.
It is straightforward to check that each item is allocated to each agent with probability $1/n$, so the marginal probabilities match the equal division rule $\calF_=$.

It remains to show that each allocation $\calA_k$ is PROP1 and $\frac1n$-MMS.
For an arbitrary $k$, let $\calA_k=(A_1,\ldots,A_n)$.
By our rule $\calD$, for each $i$ and $j$, we have $|A_i\cap G^{(i)}_j|=1$.

To show that $(A_1,\ldots,A_n)$ is PROP1, we consider adding item $g^{(i)}_1$ to $A_i$ if $g^{(i)}_1\notin A_i$ or adding item $g^{(i)}_2$ to $A_i$ otherwise.
Let $A_i^+$ be the bundle after this addition, and we need to prove $v_i(A_i^+)\geq\frac1n\cdot v_i(M)$.
Notice that $g^{(i)}_1\in A_i^+$, and $A_i^+\setminus\{g_1^{(i)}\}$ contains exactly one item from each group $G^{(i)}_j$.
By our definition of groups, for each $j=1,\ldots,m/n-1$, any single item in the $j$-th group $G^{(i)}_j$ has a weakly larger value than any item in the $(j+1)$-th group $G^{(i)}_{j+1}$.
Thus, any single item in the $j$-th group has value at least $\frac1n\cdot v_i(G^{(i)}_{j+1})$.
Summing up the $n$ items in $A_i^+\setminus\{g_1^{(i)}\}$, we have
$$v_i(A_i^+\setminus\{g_1^{(i)}\})\geq\frac1n\cdot\sum_{j=2}^{m/n}v_i(G^{(i)}_j).$$
In addition, $g^{(i)}_1$ has a value weakly larger than any item in the first group $G^{(i)}_1$, so
$$v_i(A_i^+)\geq \frac1n\cdot\sum_{j=1}^{m/n}v_i(G^{(i)}_j)=\frac1n\cdot v_i(M).$$
This proves that $(A_1,\ldots,A_n)$ is PROP1.

To show that $(A_1,\ldots,A_n)$ is $\frac1n$-MMS, we first find an upper bound to $\text{MMS}_i$.
Let $T_i=\{g_n^{(i)},g_{n+1}^{(i)},\ldots,g_m^{(i)}\}$ be the bundle consists of all but the first $n-1$ items with the largest values to agent $i$.
We will show that $\text{MMS}_i\leq v_i(T_i)$.
To see this, suppose $(X_1,\ldots,X_n)$ be the partition that defines MMS$_i$, i.e., assuming $v_i(X_1)\geq v_i(X_2)\geq\cdots\geq v_i(X_n)$, we have $v_i(X_n)=\text{MMS}_i$.
By the pigeonhole principle, there exists $X_k$ such that $X_k$ does not contain the first $n-1$ items that are not in $T_i$.
This implies $X_k\subseteq T_i$, so $\text{MMS}_i\leq v_i(X_k)\leq v_i(T_i)$.
It now suffices to show $v_i(A_i)\geq\frac1n\cdot v_i(T_i)$.

Let $A_i=\{g_1^\ast,\ldots,g_{m/n}^\ast\}$ where $g_j^\ast$ is the item in the singleton set $A_i\cap G^{(i)}_j$.
By our definition of groups, $g_j^\ast$ has a value weakly higher than the last element in $G^{(i)}_j$, namely, $g^{(i)}_{nj}$, and it has a value weakly higher than any item in the next group $G^{(i)}_{j+1}$.
Therefore, we have
$$v_i(g_j^\ast)\geq\frac1n\cdot\left(g_{nj}^{(i)}+g_{nj+1}^{(i)}+\cdots+g_{nj+n-1}^{(i)}\right),$$
and, by summing up items in $A_i$,
$$v_i(A_i\setminus\{g_{m/n}^\ast\})\geq\frac1n\sum_{j=n}^{m-1}v_i(g_j^{(i)}).$$
Finally, we have $v_i(g_{m/n}^\ast)\geq v_i(g_m^{(i)})$, so
$$v_i(A_i)\geq v_i(g_m^{(i)})+\frac1n\sum_{j=n}^{m-1}v_i(g_j^{(i)})\geq\frac1n v_i(T_i)\geq\text{MMS}_i.\qedhere$$
\end{proof}

\section{Incompatibility between Truthfulness and EF$_{-v}^{+u}$X}\label{sect:incom_efkx}
Another widely considered envy-based fairness notion stronger than EF$1$ is EFX.
Although the existence of EFX allocations remains a major open problem, EFX is known to be achievable in the case of two agents through the I-cut-you-choose protocol. 
One may expect that for two agents, similar to the compatibility between truthfulness in expectation and EF1, there is also a truthful randomized mechanism that outputs EFX allocations.
However, in this section, we will show that even the relaxation of EFX, denoted by EF$_{-v}^{+u}$X, is incompatible with truthfulness even when randomness is allowed.
\begin{definition}\label{def:efx}
    For nonnegative integers $u$ and $v$, an allocation $\calA=(A_1,\ldots,A_n)$ is \emph{envy-free up to adding any $u$ items and removing any $v$ items,} denoted by EF$_{-v}^{+u}$X, if for every pair of agents $i$ and $j$, for \emph{any} sets of items $S_i$ and $S_j$ satisfying $S_i\cap A_i=\emptyset$, $S_j\subseteq A_j$, $|S_i|=\min\{u, |M\setminus A_i|\}$, and $|S_j|=\min\{v, |A_j|\}$, we have $v_i(A_i\cup S_i)\ge v_i(A_j\setminus S_j)$.
\end{definition}

The relationship between EF$_{-v}^{+u}$X and EF$_{-(u+v)}^{+0}$X is opposite to that between EF$_{-v}^{+u}$ and EF$_{-(u+v)}^{+0}$.
\begin{proposition}\label{prop:efuvx}
    If an allocation is EF$_{-v}^{+u}$X, then it is EF$_{-(u+v)}^{+0}$X.
\end{proposition}
\begin{proof}
    Assume $\{A_1,\ldots,A_n\}$ satisfies EF$_{-v}^{+u}$X.
    Consider two agents $i$ and $j$.
    If $|A_j|\le u+v$, then EF$_{-(u+v)}^{+0}$X is trivially satisfied since agent $i$ will not envy agent $j$ after removing all items from $A_j$.
    Otherwise, consider any set of items $S\subset A_j$ such that $|S|=u+v$. 
    We partition $S$ into two sets $S_1, S_2\subset A_j$ such that $|S_1|=u$, $|S_2|=v$, and $S_1\cap S_2=\emptyset$.
    As the allocation is EF$_{-v}^{+u}$X, it holds that $v_i(A_i)+v_i(S_1)=v_i(A_i\cup S_1)\ge v_i(A_j\setminus S_2)$.
    Further, since $S_1\subset A_j$, we have $v_i(A_i)\ge v_i(A_j\setminus S_2)-v_i(S_1)=v_i(A_j\setminus S_2\setminus S_1)=v_i(A_j\setminus S)$.
    Therefore, the allocation is also EF$_{-(u+v)}^{+0}$X.
\end{proof}

In the following, we show the result for two agents.
Combining Proposition~\ref{prop:efuvx}, we obtain that for any constant $u$ and $v$, EF$_{-v}^{+u}$X is incompatible with truthfulness in expectation.

\begin{theorem}\label{thm:incom_efkx}
    For any nonnegative integer $k$, there is no randomized mechanism that is truthful and EF$_{-k}^{+0}$X even for two agents.
\end{theorem}
Before proving the above theorem, we first derive some restrictions on the EF$_{-k}^{+0}$X allocations under a series of instances, which will be adopted in the main proof.
All of the following instances contain two agents and $m+1$ items $M=\{g_0,g_1,\ldots,g_m\}$, and we denote the set $\{g_1,\ldots,g_m\}$ by $M'$.

We first show that in the following instances, each agent needs to receive at least $k$ items from $M'$.
Therefore, when we consider the $k$ removed items in the EF$_{-k}^{+0}$X definition, these items must come from $M'$.
\begin{proposition}\label{prop:incom_efkx_1}
    Consider the following instance where $x>1$ and agent $2$'s valuation is arbitrary.
    When $m>\max\left\{x, \frac{3kx}{x-1}\right\}$, an allocation $(A_1, A_2)$ is not EF$_{-k}^{+0}$X if $A_1$ contains less than $k$ items from $M'$.
\end{proposition}
\begin{center}
    \begin{tabular}{cccccc}
        \hline
        & $g_0$ & $g_1$ & $g_2$ & $\cdots$ & $g_m$ \\
        \hline
        $v_1$ & $1$ & $x/m$ & $x/m$ & $\cdots$ & $x/m$ \\
        \hline
    \end{tabular}
    \label{tab:incom_efkx_p1}
\end{center}
\begin{proof}
    Note that $3kx/(x-1)>3k$, therefore if $|A_1\cap M'|<k$, the number of items received by agent $2$ will be $|A_2\cap M'|>2k$.
    To prove Proposition~\ref{prop:incom_efkx_1}, we only need to show that even if $g_0$ is allocated to agent $1$, agent $1$ will still envy agent $2$ after removing the set $S_2$ of $k$ items from $A_2$ where $S_2\subset A_2\cap M'$.
    This is given by the following inequality,
    $$v_1(A_1)<1+kx/m<(m-k)x/m-kx/m<v_1(A_2\setminus S_2),$$
    where the second inequality holds due to $m>3kx/(x-1)$.
\end{proof}

We next show that item $g_0$ must be allocated to a certain person in any EF$_{-k}^{+0}$X allocation.
\begin{proposition}\label{prop:incom_efkx_2}
    Consider the following instance where $1<x_1<x_2$.
    When $m>\max\left\{x_2, \frac{3kx_1}{x_1-1}, \frac{2x_1x_2k}{x_2-x_1}\right\}$, an allocation $(A_1, A_2)$ is EF$_{-k}^{+0}$X only if $g_0\in A_1$.
\end{proposition}
\begin{center}
    \begin{tabular}{cccccc}
        \hline
        & $g_0$ & $g_1$ & $g_2$ & $\cdots$ & $g_m$ \\
        \hline
        $v_1$ & $1$ & $x_1/m$ & $x_1/m$ & $\cdots$ & $x_1/m$ \\
        $v_2$ & $1$ & $x_2/m$ & $x_2/m$ & $\cdots$ & $x_2/m$ \\
        \hline
    \end{tabular}
    \label{tab:incom_efkx_p2}
\end{center}
\begin{proof}
    For the sake of contradiction, assume that $g_0\in A_2$.
    As $m> 3kx_1/(x_1-1)>3kx_2/(x_2-1)$, according to Proposition~\ref{prop:incom_efkx_1}, agent $2$ will receive at least $k$ items from $M'$.
    Therefore, to guarantee EF$_{-k}^{+0}$X for agent $1$, we consider the case where agent $1$ removes $k$ items in $A_2\cap M'$ from $A_2$.
    To prevent agent $1$ from envying agent $2$ after removing the $k$ items, agent $1$ needs to be allocated at least
    $$\left\lceil \frac{m}{x_1}-k+\left(m-\left(\frac{m}{x_1}-k\right)\right)/2 \right\rceil = \left\lceil \frac{(x_1+1)m}{2x_1}-\frac{k}{2}\right\rceil\ge k$$
    items from $M'$.
    However, agent $2$ will envy agent $1$ even if removing a set $S_1$ of $k$ items from $A_1$ according to the following inequality,
    $$v_2(A_2)\le 1+\frac{x_2}{m}\left(m-\left(\frac{(x_1+1)m}{2x_1}-\frac{k}{2}\right)\right)<\frac{x_2}{m}\left(\frac{(x_1+1)m}{2x_1}-\frac{k}{2}-k\right)\le v_2(A_1\setminus S_1),$$
    where the second inequality holds due to $m>2x_1x_2k/(x_2-x_1)$.
\end{proof}

We are now ready to prove Theorem~\ref{thm:incom_efkx} based on the above two propositions.
\paragraph{Proof of Theorem~\ref{thm:incom_efkx}.}
We begin with the following instance with two agents and $m+1$ items where $m>32k$.
We show that for any fractional allocation rule $\calF$ that guarantees $\expo$ EF$_{-k}^{+0}$X under this instance, there always exist beneficial misreports for either agent $1$ or agent $2$.
\begin{center}
    \begin{tabular}{cccccc}
        \hline
        & $g_0$ & $g_1$ & $g_2$ & $\cdots$ & $g_m$ \\
        \hline
        $v_1$ & $1$ & $2/m$ & $2/m$ & $\cdots$ & $2/m$ \\
        $v_2$ & $1$ & $4/m$ & $4/m$ & $\cdots$ & $4/m$ \\
        \hline
    \end{tabular}
    \label{tab:incom_efkx_i}
\end{center}

Denote $\{g_1,\ldots,g_m\}$ by $M'$.
For any allocation $(A_1,A_2)$ that satisfies EF$_{-k}^{+0}$X, according to Proposition~\ref{prop:incom_efkx_1}, each agent needs to be allocated at least $k$ items from $M'$.
According to Proposition~\ref{prop:incom_efkx_2}, we have $g_0\in A_1$.
Further, to guarantee EF$_{-k}^{+0}$X for agent $2$, agent $2$ needs to be allocated at least
$\left\lceil \frac{m}{4}-k+\frac{m-(m/4-k)}{2}\right\rceil=\left\lceil\frac{5m}{8}-\frac{k}{2}\right\rceil$
items from $M'$.
To guarantee EF$_{-k}^{+0}$X for agent $1$, agent $2$ can be allocated at most
$\left\lfloor \frac{m}{2}+k+\frac{m-(m/2+k)}{2}\right\rfloor=\left\lfloor\frac{3m}{4}+\frac{k}{2}\right\rfloor$
items from $M'$.

Consider the fractional allocation induced by the fractional division rule $\calF$.
Denote by $\alpha_2$ the (possibly fractional) number of items that agent $2$ receives from $M'$ in the fractional allocation, and $\alpha_1=m-\alpha_2$ the (possibly fractional) number of items agent $1$ receives from $M'$.
As the fractional allocation is a probability distribution over the set of all possible \expo EF$_{-k}^{+0}$X allocations, we have $\alpha_2\in\left[\left\lceil\frac{5m}{8}-\frac{k}{2}\right\rceil, \left\lfloor\frac{3m}{4}+\frac{k}{2}\right\rfloor\right]$.
We now consider two cases based on $\alpha_2$ to show that whichever $\alpha_2$ we choose for the fractional division rule, there always exists a deviation for some agent that leads to a higher expected utility.

If $\alpha_2\le \frac{11m}{16}$, agent $2$ may misreport $v_2$ to $v'_2$ to increase her expected utility, as shown in the following instance.
\begin{center}
    \begin{tabular}{cccccc}
        \hline
        & $g_0$ & $g_1$ & $g_2$ & $\cdots$ & $g_m$ \\
        \hline
        $v_1$ & $1$ & $2/m$ & $2/m$ & $\cdots$ & $2/m$ \\
        $v'_2$ & $1$ & $16/(7m)$ & $16/(7m)$ & $\cdots$ & $16/(7m)$ \\
        \hline
    \end{tabular}
    \label{tab:incom_efkx_m1}
\end{center}
From Proposition~\ref{prop:incom_efkx_2}, we know that $g_0$ will be allocated to agent $1$.
Still, to guarantee EF$_{-k}^{+0}$X for agent $1$, agent $2$ can be allocated at most $\left\lfloor\frac{3m}{4}+\frac{k}{2}\right\rfloor$ items from $M'$.
However, to guarantee EF$_{-k}^{+0}$X for agent $2$ under $v'_2$, agent $2$ needs to be allocated at least
$\left\lceil \frac{7m}{16}-k+\frac{m-(7m/16-k)}{2}\right\rceil=\left\lceil\frac{11.5m}{16}-\frac{k}{2}\right\rceil$ items from $M'$.
Denote by $\alpha'_2$ the (possibly fractional) number of items agent $2$ receives under $\calF$ under $v'_2$, then we have $\alpha'_2\in\left[\left\lceil\frac{11.5m}{16}-\frac{k}{2}\right\rceil, \left\lfloor\frac{3m}{4}+\frac{k}{2}\right\rfloor\right]$.
Since $m>16k$, agent $2$ will receive at least $\frac{11.5m}{16}-\frac{k}{2}-\frac{11m}{16}>0$ more items from $M'$ in expectation, leading to a strict increase in her expected utility.

Conversely, if $\alpha_2>\frac{11m}{16}$, then $\alpha_1\le \frac{5m}{16}$, and agent $1$ may misreport $v_1$ to $v_1'$ to increase her utility.
\begin{center}
    \begin{tabular}{cccccc}
        \hline
        & $g_0$ & $g_1$ & $g_2$ & $\cdots$ & $g_m$ \\
        \hline
        $v'_1$ & $1$ & $16/(5m)$ & $16/(5m)$ & $\cdots$ & $16/(5m)$ \\
        $v_2$ & $1$ & $4/m$ & $4/m$ & $\cdots$ & $4/m$ \\
        \hline
    \end{tabular}
    \label{tab:incom_efkx_m2}
\end{center}
The analysis is very similar to the above.
In the new instance, agent $1$ still receives $g_0$ and agent $2$ needs to receive at least $\left\lceil\frac{5m}{8}-\frac{k}{2}\right\rceil$ items from $M'$.
However, to guarantee EF$_{-k}^{+0}$X for agent $1$ under $v'_1$, agent $2$ can be allocated at most
$\left\lfloor \frac{5m}{16}+k+\frac{m-(5m/16+k)}{2}\right\rfloor=\left\lfloor\frac{10.5m}{16}+\frac{k}{2}\right\rfloor$ items from $M'$.
Denote by $\alpha'_1$ the (possibly fractional) number of items agent $1$ receives under $\calF$ under $v'_1$, then we have $\alpha'_1\ge\frac{5.5m}{16}-\frac{k}{2}$, thus she will receive at least $\frac{5.5m}{16}-\frac{k}{2}-\frac{5m}{16}>0$ more items from $M'$ from the misreport.\hfill $\square$

For more than two agents, the existence of EFX allocations remains open.
Nevertheless, even when we restrict our attention to instances in which EFX allocation is guaranteed to exist, the incompatibility result extends to hold.
In particular, we construct $n-2$ additional agents whose values are the same as agent $2$, and modify the number of items accordingly (specifically, ensure that $m$ is sufficiently large).
We can similarly prove that in any EF$_{-k}^{+0}$X allocations, item $g_0$ needs to be allocated to agent $1$, and agent $2$ will receive a fraction of items from $M'$ within a certain range in the allocation output by the fractional allocation rule $\calF$.
However, for each (possibly fractional) number of items received by agent $2$ from $M'$, there is always a beneficial misreport for either agent $1$ or agent $2$.
This gives us the following theorem, whose proof is similar to the case of two agents and thus deferred to Appendix~\ref{append:incom_efkx_n}.

\begin{theorem}\label{thm:incom_efkx_n}
    For any fixed number of agents and any nonnegative integer $k$, there is no randomized mechanism that is truthful and EF$_{-k}^{+0}$X.
\end{theorem}

\section{Truthful, Almost Envy-Free, and Pareto-Optimal Mechanisms}
\label{sect:PO}
Even if we do not care about the decomposition rule $\calD$ of indivisible item allocations and only focus on fractional allocations, the fractional allocation rules $\calF$ that satisfy truthfulness and Pareto-optimality are already quite restrictive.
It is known that any truthful and Pareto-optimal mechanism for divisible item allocations gives an agent all the items where she has positive values~\cite{momi2017efficient,garg2022efficient}.
This type of dictatorship-styled mechanism clearly has a poor performance on fairness.

In this section, we consider restrictive valuation functions.
For binary valuation functions (where $v_i(g)\in\{0,1\}$ for any agent $i$ and any item $g$), the maximum Nash welfare rule provides a deterministic mechanism that satisfies truthfulness and the EF1 property~\cite{halpern2020fair,babaioff2021fair,barman2022truthful}.
We aim to investigate how far this can be generalized.

We will show that the existence of Pareto-optimal, truthful, and almost envy-free mechanisms can be at most generalized to bi-valued valuation functions:
when agents' valuation functions are \emph{bi-valued} (i.e., $v_i(g)\in\{p,q\}$ for some $p>q\geq0$), there exists a randomized truthful EF1 mechanism that satisfies $\exan$ Pareto-optimality; when agents' valuation functions are \emph{tri-valued} (i.e., $v_i(g)\in\{p,q,r\}$ for some $p>q>r\geq 0$), for any $u$ and $v$, even with two agents, there does not exist a mechanism that is truthful, $\expo$ Pareto-optimal, and EF$^{+u}_{-v}$.

\subsection{Positive Result for Bi-Valued Valuation Functions}
\label{sect:PObi}
\citet{halpern2020fair} and \citet{babaioff2021fair} independently show that the maximum Nash welfare rule, defined by finding an allocation $(A_1,\ldots,A_n)$ with maximum Nash welfare $\prod_{i=1}^nv_i(A_i)$, is truthful (with some consistent tie-breaking rule) for binary valuations $v_i(g)\in\{0,1\}$.
Moreover, even for general valuation functions, an allocation with maximum Nash welfare is always EF1~\cite{CKMP+19}.
This gives us a deterministic truthful EF1 mechanism.

However, the maximum Nash welfare rule fails to guarantee truthfulness for bi-valued valuation functions.
Consider the following example with two agents and six items.
\begin{center}
    \begin{tabular}{ccccccc}
    \hline
         & $g_1$ & $g_2$ & $g_3$ & $g_4$ & $g_5$ & $g_6$ \\
    \hline
        $v_1$ & 2 & 2 & 1 & 1 & 1 & 1\\
        $v_2$ & 1 & 1 & 1 & 1 & 1 & 1\\
    \hline
    \end{tabular}
\end{center}
The allocation maximizing the Nash welfare $\prod_{i=1}^nv_i(A_i)$, among integral allocations or fractional allocations, is unique: $A_1=\{g_1,g_2\}$ and $A_2=\{g_3,g_4,g_5,g_6\}$.
However, if agent $1$ misreports her valuation function by changing $v_1(g_3)$ from $1$ to $2$, the allocation maximizing the Nash welfare, among integral allocations or fractional allocations, becomes $A_1=\{g_1,g_2,g_3\}$ and $A_2=\{g_4,g_5,g_6\}$.
This is beneficial for agent $1$.
Therefore, the maximum Nash welfare rule with both the $\exan$ version (a randomized mechanism $(\calF,\calD)$ with $\calF$ being the rule that finds a possibly fractional allocation with the highest Nash welfare among all possibly fractional allocations) and the $\expo$ version (find integral allocations with the highest Nash welfare among all integral allocations) fail to be truthful for bi-valued valuations.

Nevertheless, we will show that the maximum Nash welfare rule can be carefully twisted to achieve truthfulness.
\begin{theorem}\label{thm:PObi}
    There exists a truthful, EF1, and $\exan$ Pareto-optimal randomized mechanism $(\calF,\calD)$ if agents' valuation functions satisfy $v_i(g)\in\{p,q\}$ for every $i\in N$ and $g\in M$.
\end{theorem}

as the case with $q=0$ reduces to the binary setting (by rescaling the valuations such that $p=1$) and the maximum Nash welfare rule satisfies EF1 and truthfulness by~\citet{halpern2020fair} and~\citet{babaioff2021fair}.
We use the alternative notation $(X_1,\ldots,X_n)$ for a fractional allocation $\bX=\{x_{ig}\}_{i\in N, g\in M}$, where $(X_1,\ldots,X_n)$ is a partition of $[0,m]$.
In particular, each item $g=1,\ldots,m$ is viewed as an interval $[g-1,g]$, the item set is then the union of $m$ intervals (which is $[0,m]$).
Given a fractional bundle $X$, let $|X|$ be the size of $X$, for which we say \emph{the number of items} in $X$ (although this number may be fractional).
Naturally, $|X|=\sum_{g=1}^m|X\cap [g-1,g]|$, and the notation $\bX=\{x_{ig}\}_{i\in N,g\in M}$ is translated to $x_{ig}=|X_i\cap[g-1,g]|$ in our new notation.
We use $v_i(X)$ to denote agent $i$'s value on $X\subseteq [0,m]=M$.

We first describe the fractional division rule $\calF$.
Let $L=m/n$, and we ensure each agent receives a total of exactly $L$ items.
Notice that $L$ may not be an integer.
We then let $(v_1',\ldots,v_n')$ be the valuation profile where, for each $i\in N$ and $g\in M$, 
$v_i'(g)=\left\{\begin{array}{ll}
    1 & \mbox{if }v_i(g)=p \\
    0 & \mbox{if }v_i(g)=q
\end{array}\right..$

In the first phase, we compute a (possibly fractional and partial) allocation $(X_1',\ldots,X_n')$ that maximizes the Nash welfare with respect to $(v_1',\ldots,v_n')$.
In particular, it first maximizes the number of agents receiving positive utilities, then maximizes the Nash welfare among the set $S$ of agents receiving positive utilities $\prod_{i\in S}v_i'(X_i')$.
Note that in the first phase, we only focus on items with value $1$ under some $v'_i$; if an item has value $0$ under each $v'_i$, it will not be allocated in this phase.
Next, in the second phase, if the total number of items in $X_i'$ is more than $L=m/n$, we truncate $X_i'$ such that its size is exactly $L$.
That is, find an arbitrary subset $X_i''$ of $X_i'$ such that $|X_i''|=L$.
Let $(X_i'',\ldots,X_n'')$ be the allocation after this operation, which may be a partial allocation with unallocated items.
Finally, in the third phase, we allocate the remaining unallocated items (including the unallocated items from the first phase and the truncated items from the second phase) to the agents $i$ with $|X_i''|<L$ in a way such that each agent receives exactly $L$ units of items at the end.
This is done in a way that each unallocated item is allocated ``uniformly''.
Specifically, let $T=m-\sum_{i=1}^n|X_i''|$ be the total amount of unallocated items and $\alpha_i=\frac{L-|X_i''|}{T}$ (notice that $\sum_{i=1}^n\alpha_i=1$).
A fraction $\alpha_i$ of each unallocated item is added to $X_i''$.
We obtain an allocation $(X_1,\ldots,X_n)$ such that $|X_i|=L$ for each agent $i$.
Notice that an item $g$ may have been allocated partially in $(X_1'',\ldots,X_n'')$ so that only a fraction $\beta_g$ of $g$ is unallocated before the third phase.
In this case, we include an $\alpha_i\cdot\beta_g$ fraction of item $g$ to each agent's bundle $X_i''$.
This completes the description of the division rule $\calF$.

The following three propositions prove Theorem~\ref{thm:PObi}.
\begin{proposition}\label{prop:POtruthful}
    The division rule $\calF$ is truthful.
\end{proposition}
\begin{proof}[Proof (sketch)]
    We only give a very high-level idea here. The formal proof involves many technical details and is deferred to Appendix~\ref{append:PObi}.

    The proof is intuitively based on the fact that the (fractional version of) maximum Nash welfare rule is truthful for binary valuations~\cite{chen2013truth,aziz2014cake}\footnote{A truthful mechanism for binary valuations was first given by~\citet{chen2013truth}, and \citet{aziz2014cake} realized that the mechanism by Chen et al. is the maximum Nash welfare rule.}.
    If an agent $i$ misreports her valuation function, by the truthfulness of the maximum Nash welfare rule, the number of items with value $p$ to agent $i$ allocated in the first and second phases in $X_i''$ cannot be increased.
    Agent $i$ can only hope that some of the items where she has value $p$ will be fractionally allocated to her in the third phase (if agent $i$ misreports her valuation functions such that the value of some of the items where she has value $p$ is reported as $q$, then these items may be fractionally allocated to her in the third phase).
    However, by our uniform way of allocating remaining items in the third phase, we can guarantee that the misreporting is not beneficial to agent $i$.
    Proving this requires careful analysis including reducing the problem to the truthfulness of a hypothetical division rule and breaking down into the analysis in~\citet{chen2013truth}.
    It is discussed in Appendix~\ref{append:PObi}.
\end{proof}

\begin{proposition}\label{prop:POmarket}
    The division rule $\calF$ is Pareto-optimal.
\end{proposition}
\begin{proof}
    We begin by introducing the Fisher market.
    A Fisher market with $n$ agents and $m$ \emph{divisible} items takes a set of valuation functions $(v_1,\ldots,v_n)$ and a set of \emph{budgets} $(b_1,\ldots,b_n)\in\mathbb{R}_{\geq0}^n$ as inputs, and outputs $(\bX,\mathbf{p})$ where $\bX$ is an allocation and $\mathbf{p}=(p_1,\ldots,p_m)\in\mathbb{R}_{\geq0}^n$ is the price vector.
    We say that $(\bX,\mathbf{p})$ is a \emph{market equilibrium} if
    \begin{itemize}
        \item all items with positive prices (i.e., $p_g> 0$) are fully allocated $\sum_{i=1}^nx_{ig}=1$,
        \item each agent spends all her budget: for each $i=1,\ldots,n$, we have $b_i=\sum_{g=1}^mp_g\cdot x_{ig}$, and
        \item each agent only buys items with the best value-to-price ratio; formally, for each agent $i$, let $\gamma_i=\max_{g=1,\ldots,m}v_{ig}/p_g$ be the \emph{maximum bang-per-buck} ratio, and we require $v_{ig}/p_g=\gamma_i$ whenever $x_{ig}>0$.
    \end{itemize}
    The \emph{first welfare theorem} states that $\bX$ is Pareto-optimal if $(\bX,\mathbf{p})$ is a market equilibrium.

    To prove Proposition~\ref{prop:POmarket}, we will define budgets $b_1,\ldots,b_n$ for the $n$ agents and the prices $p_1,\ldots,p_m$ for the $m$ items such that the allocation $(X_1,\ldots,X_n)$ output by $\calF$, together with the price vector, is a market equilibrium.
    Recall that $(X_1',\ldots,X_n')$ is an allocation that maximizes the Nash welfare for valuation functions $v_1',\ldots,v_n'$ defined by modifying $v_1,\ldots,v_n$ with $p$ changed to $1$ and $q$ changed to $0$.
    If $|X_i'|>L=m/n$, it is truncated such that a subset $X_i''$ of $X_i'$ with length $L$ is finally allocated to agent $i$.
    Let $Z$ be the set of agents whose bundles have been truncated.
    Notice that $X_i'=X_i''$ for each $i\in N\setminus Z$.
    
    We first show a property that, in the intermediate allocation $(X_1',\ldots,X_n')$, an item $g$ cannot be shared between an agent in $N\setminus Z$ and an agent in $Z$: if $v_i(g)=p$ for some $i\in N\setminus Z$, then $g$ will be allocated only among agents in $N\setminus Z$, and no fraction of $g$ will be included in $X'_{i'}$ for $i' \in Z$.
    Suppose this is not the case, and $g$ is included in some $X_{i'}'$ for some $i'\in Z$.
    By the fact that $X_i'$ is not truncated (since $i\in N\setminus Z$) and $X_{i'}'$ is truncated, we have $|X_i'|<|X_{i'}'|$.
    It is easy to see that moving some fraction of $g$ from $X_{i'}'$ to $X_i'$ increase the Nash welfare for the valuation profile $(v_1',\ldots,v_n')$, which contradicts to that $(X_1',\ldots,X_n')$ is a maximum Nash welfare solution.

    We are now ready to define the prices for all items and budgets for agents.
    If an item $g$ is not allocated in the first phase (i.e., $v_i(g)=q$ for all $i\in N$), its price is set to $p_g=q$.
    If some fraction of an item $g$ is included in some $X_i'$ for some $i\in N\setminus Z$, the price of $g$ is defined by $p_g=p$; otherwise, by our observation in the previous paragraph, $g$ is allocated among the agents in $Z$ in the intermediate allocation $(X_1',\ldots,X_n')$, and its price is set to $p_g=q$.
    The budget $b_i$ for each agent $i$ is set to the value such that agent $i$ spends exactly all her budget $b_i$ to buy $X_i$.
    That is,
    $$b_i=\left\{\begin{array}{ll}
        p|X_i'|+q(L-|X_i'|) & \mbox{if }i\in N\setminus Z \\
        qL & \mbox{if }i\in Z
    \end{array}\right..$$

    We next show that this is a market equilibrium.
    For each agent $i\in N\setminus Z$, her maximum bang-per-buck ratio is $\gamma_i=1$, as we have shown that there does not exist an item $g$ such that $v_i(g)=p$ while $p_g=q$.
    It is straightforward to see that agent $i$ spends all her budget on the items with the maximum bang-per-buck ratio $\gamma_i=1$.
    For each agent $i'\in Z$, her maximum bang-per-buck ratio is $\gamma_{i'}=p/q$.
    The $L$ items she receives have value $p$ (as these $L$ items form $X_{i'}''$, which is truncated from $X_{i'}'$ that includes only items of value $p$).
    The prices of them are set to $q$, as we have shown that each item with value $p$ will not be included in $X'_{i'}$.
    Therefore, agent $i'$ spends all her budget on items with the maximum bang-per-buck ratio.
\end{proof}

\begin{proposition}\label{prop:POrealizability}
    The division rule $\calF$ always outputs fractional allocations that are EF1-realizable. Moreover, the decomposition of EF1 allocations can be done in polynomial time.
\end{proposition}
\begin{proof}
    We will show that the fractional allocation $(X_1,\ldots,X_n)$ output by $\calF$ is the outcome of the probabilistic serial rule under certain tie-breakings.
    Then, the proposition follows by the result in~\citet{AzizFrSh23} that shows the outcome of the probabilistic serial rule is EF1-realizable and the decomposition can be computed in polynomial time.
    
    Suppose each agent $i$ eats $X_i'$ first and then $X_i\setminus X_i'$.
    By the time $L$, all the items are eaten.
    Notice that some agent $i$ does not have enough time to finish $X_i'$; in particular, this happens when $X_i'$ has been truncated in the second phase of $\calF$.
    
    We need to show that, when $X_i'$ is fully eaten by agent $i$, no fraction of an item $g$ with $v_{ig}=p$ remains.
    By the time $X_i'$ is fully eaten, all the fractional bundles $X_{i'}'$ with $|X_{i'}'|\leq|X_i'|$ are fully eaten.
    By this time, if some item $g$ is not fully eaten, some fraction of $g$ is in the bundle $X_{i''}'$ with $|X_{i''}'|>|X_i'|$.
    We must have $v_i'(g)=0$ for the allocation $(X_1',\ldots,X_n')$ to be the maximum Nash welfare solution for $(v_1',\ldots,v_n')$ (otherwise, move some fraction of $g$ from $X_{i''}'$ to $X_i'$ increases the Nash welfare).
    As a result, $v_i(g)=q$.
    Therefore, by the time $X_i'$ is fully eaten by agent $i$, no fraction of an item $g$ with $v_{ig}=p$ remains.
\end{proof}

\subsection{Negative Result for Tri-Valued Valuation Functions}
For tri-valued valuations, truthfulness and almost envy-freeness are not even compatible with the weaker notion of $\expo$ Pareto-optimality, even if there are only two agents.

\begin{theorem}\label{thm:POtri}
    For $n=2$, there exist $p,q,r$ with $p>q>r\geq0$ such that for all $u,v\in\mathbb{Z}^+$ there does not exist an \expo Pareto-optimal, EF$^{+u}_{-v}$, and truthful randomized mechanism even when agents' valuation functions satisfy $v_i(g)\in\{p,q,r\}$ for every $i\in N$ and $g\in M$.
\end{theorem}
\begin{proof}
    We consider $p=1$, $q=0.02$, and $r=0$.
    Suppose such a mechanism $(\calF,\calD)$ exists.
    Consider the number of items $m$ to be sufficiently large compared with $u$ and $v$, say, $m=200(u+v)$.
    Let $M_1$ be the set of the first $0.5m$ items and $M_2$ be the set of the remaining $0.5m$ items.
    Consider the first instance where both agents have value $1$ on items in $M_1$ and value $0$ on items in $M_2$.
    \begin{center}
        \begin{tabular}{ccc}
        \hline
             & value on items in $M_1$ & value on items in $M_2$ \\
        \hline
          agent 1   & $1$ & $0$\\
          agent 2   & $1$ & $0$\\
          \hline
        \end{tabular}
    \end{center}
    
    Let $\alpha_1$ be the (possibly fractional) number of items in $M_1$ that agent $1$ receives under the rule $\calF$ and $\alpha_2=0.5m-\alpha_1$ be the (possibly fractional) number of items in $M_1$ that agent $2$ receives.
    We must have $\alpha_1,\alpha_2\in[0.25m-(u+v),0.25m+(u+v)]$ to guarantee EF$^{+u}_{-v}$.
    If not, say $\alpha_1<0.25m-(u+v)$, then there exists integral allocation output by $\calD$ where agent $1$ receives less than $0.25m-(u+v)$ items in $M_1$, which violates EF$^{+u}_{-v}$.
    
    Next, consider the second instance where agent $1$ has value $0.02$ on the items in $M_2$ instead.
    \begin{center}
        \begin{tabular}{ccc}
        \hline
             & value on items in $M_1$ & value on items in $M_2$ \\
        \hline
          agent 1   & $1$ & $0.02$\\
          agent 2   & $1$ & $0$\\
          \hline
        \end{tabular}
    \end{center}
    By Pareto-optimality, all items in $M_2$ should be given to agent $1$.
    Moreover, agent $1$ should receive at least $\alpha_1-0.01m\geq 0.24m-(u+v)$ items from $M_1$ to guarantee truthfulness.
    If less than this, agent $1$ would misreport her valuation function to the one in the first instance, which is beneficial.
    
    Consider the third instance where agent $1$ has value $0.02$ on items in $M_1$ and value $1$ on items in $M_2$.
    \begin{center}
        \begin{tabular}{ccc}
        \hline
             & value on items in $M_1$ & value on items in $M_2$ \\
        \hline
          agent 1   & $0.02$ & $1$\\
          agent 2   & $1$ & $0$\\
          \hline
        \end{tabular}
    \end{center}
    By Pareto-optimality, all items in $M_2$ should be allocated to agent $1$.
    Moreover, to guarantee truthfulness, the number of items agent $1$ receives from $M_1$ should not be less than what she has received in the previous case.
    Thus, agent $1$ receives at least $0.24m-(u+v)$ items from $M_1$, and agent $2$ receives at most $0.26m+(u+v)$ items from $M_1$.
    
    Consider the fourth instance with valuation functions defined as follows.
    \begin{center}
        \begin{tabular}{ccc}
        \hline
             & value on items in $M_1$ & value on items in $M_2$ \\
        \hline
          agent 1   & $0$ & $1$\\
          agent 2   & $1$ & $0.02$\\
          \hline
        \end{tabular}
    \end{center}
    By symmetry of $M_1$ and $M_2$ and symmetry of both agents, the same analysis in the first three instances indicates that agent $1$ can receive at most $0.26m+(u+v)$ items from $M_2$.
    
    Finally, consider the fifth instance defined below.
    \begin{center}
        \begin{tabular}{ccc}
        \hline
             & value on items in $M_1$ & value on items in $M_2$ \\
        \hline
          agent 1   & $0.02$ & $1$\\
          agent 2   & $1$ & $0.02$\\
          \hline
        \end{tabular}
    \end{center}
    Agent $2$ receives at most $0.26m+(u+v)$ items from $M_1$, for otherwise, in the case agent $2$'s true valuation function is the one in the third instance, she would report the valuation function in this instance instead.
    Similarly, agent $1$ receives at most $0.26m+(u+v)$ items from $M_2$, for otherwise, in the fourth instance, agent $2$ would report her valuation function as it is in this instance.
    This already violates $\exan$ Pareto-optimality.
    Next, we show that it is impossible to decompose such a fractional allocation to integral allocations that are $\expo$ Pareto-optimal and EF$^{+u}_{-v}$.
    
    Consider agent $2$.
    To guarantee $\expo$ Pareto-optimality, if agent $2$ receives at least one item from $M_2$ in some (integral) allocation generated by $\calD$, she must receive all items in $M_1$ in this allocation.
    The fractional allocation indicates that the expected number of items agent $2$ receives from $M_2$ is at least $0.24m-(u+v)$.
    Thus, the probability that agent $2$ receives at least one item from $M_2$ is at least $0.48-\frac{2(u+v)}m$ (otherwise, the expected number of items is less than $(0.48-\frac{2(u+v)}m)\cdot0.5m<0.24m-(u+v)$).
    As a result, with probability at least $0.48-\frac{2(u+v)}m$, agent $2$ receives all the items in $M_1$.
    Since the expected number of items agent $2$ received in $M_1$ is at most $0.26m+(u+v)$, there exists an allocation output by $\calD$ where agent $2$ receives at most $0.1m$ items from $M_1$.
    Otherwise, if agent $2$ receives strictly more than $0.1m$ items from $M_1$ in all allocations, we have a contradiction: the expected number of items agent $2$ receives from $M_1$ is more than
    $$0.1m\times\left(0.52+\frac{2(u+v)}m\right)+0.5m\times\left(0.48-\frac{2(u+v)}m\right)=0.292m-0.8\cdot(u+v)>0.26m+(u+v),$$
    where the last inequality is due to $m=200(u+v)$.
    It is clear that an allocation where agent $2$ receives at most $0.1m$ items from $M_1$ and no item from $M_2$ is far from being EF$^{+u}_{-v}$.
\end{proof}

\section{Discussion on Ex-Ante Envy-Freeness}
From the best-of-both-worlds aspect, our mechanisms also provide $\exan$ fairness guarantees.
\begin{definition}\label{def:ex-ante-ef}
    A randomized mechanism is \emph{$\exan$ envy-free} if the fractional allocation $\bX$ it implements is envy-free.
    That is, for every pair of agents $i$ and $j$, it holds that $\sum_{g\in M} v_{ig}x_{ig}\ge\sum_{g\in M} v_{ig}x_{jg}$.
\end{definition}

It is straightforward to see that the equal division rule satisfies $\exan$ envy-freeness: for any $i,j\in N$, we have $\sum_{g\in M}v_{ig}x_{jg} =\frac{1}{n}v_i(M)$.
Moreover, it is known that the probabilistic serial rule also satisfies $\exan$ envy-freeness~\cite{AzizFrSh23}.
This directly implies that our mechanism for two agents in Sect.~\ref{sect:twoagents}, mechanisms for $n$ agents in Sect.~\ref{sect:nagents}, and the mechanism for bi-valued valuation functions in Sect.~\ref{sect:PObi} satisfy $\exan$ envy-freeness.

However, $\exan$ envy-freeness fails for our mechanism for three agents in Sect.~\ref{sect:threeagents}, as in the fractional allocation of Type I items, agent $3$ may envy agents $1$ and $2$.
Nevertheless, it provides an approximate envy-freeness guarantee, $\exan$ $\frac12$-envy-freeness.
\begin{definition}\label{def:ex-ante-aef}
    A randomized mechanism is \emph{$\exan$ $\alpha$-envy-free} if the fractional allocation $\bX$ it implements is $\alpha$-approximate envy-free. 
    That is, for every pair of agents $i$ and $j$, it holds that $\sum_{g\in M} v_{ig}x_{ig}\ge\alpha\cdot\sum_{g\in M} v_{ig}x_{jg}$.
\end{definition}

Under our mechanism for three agents, for Type I items, it is easy to verify that agent $1$ will not envy agents $2$ and $3$ within each group, thus the fractional allocation of Type I items is envy-free to agent $1$; the same holds for agent $2$.
For agent $3$ and each group of items $G=\{a,b,c\}$ of Type I, we have $$\sum_{g\in G}v_{3g}x_{3g}=\frac13v_{3a}+\frac13v_{3b}+\frac13v_{3c}\ge\frac12\left(\frac23v_{3a}+\frac13v_{3c}\right)=\frac12\sum_{g\in G}v_{3g}x_{1g},$$ and $\sum_{g\in G}v_{3g}x_{3g}\ge \frac12\sum_{g\in G}v_{3g}x_{2g}$.
Therefore, the fractional allocation of Type I items is $\frac12$-approximate envy-free to agent $3$.
For Type II items, as the fractional allocation is given by the equal division rule, it satisfies $\exan$ envy-freeness.
Combining the fractional allocations of the two types, we may conclude that the mechanism is $\exan$ envy-free to agents $1$ and $2$, and $\exan$ $\frac12$-envy-free to agent $3$.

\section{Future Work}
A natural future direction is to strengthen the result in this paper by designing a truthful and EF1 randomized mechanism (or proving such mechanisms do not exist) for three agents, or even $n$ agents.
We showed that the equal division rule does not work (Theorem~\ref{thm:notEF1realizable}).
Moreover, we note that the fractional allocation rule $\calF$ designed in Sect.~\ref{sect:three-fractional} fails to output allocations that are EF1-realizable: the same counterexample in the proof of Theorem~\ref{thm:notEF1realizable} shows this.
For $n$ agents, known fractional allocation rules such as the ones given by~\citet{freeman2023equivalence} and~\citet{shende2023strategy} provide fractional allocations that are close to the equal division.
Specifically, the fraction of each item allocated to each agent is restricted to the range $[0,\frac2n]$ for both rules.
It is unclear if being this close to the equal division makes these fractional allocations not EF1-realizable.

Another direction is to establish lower bounds on the best possible envy-based fairness that a randomized truthful mechanism can achieve.
To obtain a lower bound of EF$c$ (corresponding to EF$_{-c}^{+0}$ in our paper), a natural idea is to characterize all possible truthful fractional rules and then show that each rule satisfying the characterization is not EF$c$-realizable.
As we discussed in the introduction, both steps are technically involved, especially when considering more than two agents.
Even for deterministic mechanisms, the impossibility result for two agents~\cite{amanatidis2017truthful} does not imply the case for three agents, and to our best knowledge, there is no existing impossibility result for more than two agents without additional assumptions.
Another approach is to construct a series of instances that finally lead to an instance where either fairness or truthfulness is violated, which is widely used in the impossibility results for deterministic mechanisms~\cite{amanatidis2017truthful,tao2022existence,tao2023fair,garg2022efficient}.
This is also how we show the impossibility result of EF$_{-v}^{+u}$X.
However, contrary to the power of randomness, the approach is more difficult to apply to randomized mechanisms due to the flexibility in fractional allocation when considering EF$_{-v}^{+u}$, making it more challenging to restrict a randomized mechanism's behavior.
The reason it works for EF$_{-v}^{+u}$X is exactly that the fairness constraint, much stronger than EF$_{-v}^{+u}$, restricts the feasible integral allocations, thus largely reduces the flexibility in fractional allocations.

We have presented a randomized truthful mechanism for bi-valued valuations that is EF1 and Pareto-optimal.
An interesting open question is whether a deterministic truthful mechanism can achieve these guarantees.
We conjecture that the answer is no.

Finally, all of our positive results rely on additive valuations and fail to generalize to richer valuation domains.
Under more general valuations, an agent's expected utility cannot be calculated directly from a fractional allocation as it is no longer well-defined.
The expected utility can only be obtained from a probability distribution of integral allocations, that is, after the decomposition rule is applied to the fractional allocation.
Therefore, our framework to separately consider $\calF$ and $\calD$ no longer applies, and we need other techniques for the case beyond additive.

\bibliographystyle{plainnat}
\bibliography{reference}

\begin{thebibliography}{85}
\providecommand{\natexlab}[1]{#1}
\providecommand{\url}[1]{\texttt{#1}}
\expandafter\ifx\csname urlstyle\endcsname\relax
  \providecommand{\doi}[1]{doi: #1}\else
  \providecommand{\doi}{doi: \begingroup \urlstyle{rm}\Url}\fi

\bibitem[Akrami and Garg(2024)]{akrami2024breaking}
Hannaneh Akrami and Jugal Garg.
\newblock Breaking the 3/4 barrier for approximate maximin share.
\newblock In \emph{Proceedings of the 2024 Annual ACM-SIAM Symposium on Discrete Algorithms (SODA)}, pages 74--91. SIAM, 2024.

\bibitem[Akrami et~al.(2023)Akrami, Alon, Chaudhury, Garg, Mehlhorn, and Mehta]{akrami2023efx}
Hannaneh Akrami, Noga Alon, Bhaskar~Ray Chaudhury, Jugal Garg, Kurt Mehlhorn, and Ruta Mehta.
\newblock {EFX}: a simpler approach and an (almost) optimal guarantee via rainbow cycle number.
\newblock In \emph{Proceedings of the 24th ACM Conference on Economics and Computation}, pages 61--61, 2023.

\bibitem[Amanatidis et~al.(2022)Amanatidis, Birmpas, Filos-Ratsikas, Voudouris, et~al.]{amanatidis2022fair}
G~Amanatidis, G~Birmpas, A~Filos-Ratsikas, AA~Voudouris, et~al.
\newblock Fair division of indivisible goods: A survey.
\newblock In \emph{IJCAI International Joint Conference on Artificial Intelligence}, pages 5385--5393. International Joint Conferences on Artificial Intelligence, 2022.

\bibitem[Amanatidis et~al.(2016)Amanatidis, Birmpas, and Markakis]{amanatidis2016truthful}
Georgios Amanatidis, Georgios Birmpas, and Evangelos Markakis.
\newblock On truthful mechanisms for maximin share allocations.
\newblock In \emph{Proceedings of the Twenty-Fifth International Joint Conference on Artificial Intelligence}, pages 31--37, 2016.

\bibitem[Amanatidis et~al.(2017{\natexlab{a}})Amanatidis, Birmpas, Christodoulou, and Markakis]{amanatidis2017truthful}
Georgios Amanatidis, Georgios Birmpas, George Christodoulou, and Evangelos Markakis.
\newblock Truthful allocation mechanisms without payments: Characterization and implications on fairness.
\newblock In \emph{Proceedings of the 2017 ACM Conference on Economics and Computation}, pages 545--562, 2017{\natexlab{a}}.

\bibitem[Amanatidis et~al.(2017{\natexlab{b}})Amanatidis, Markakis, Nikzad, and Saberi]{amanatidis2017approximation}
Georgios Amanatidis, Evangelos Markakis, Afshin Nikzad, and Amin Saberi.
\newblock Approximation algorithms for computing maximin share allocations.
\newblock \emph{ACM Transactions on Algorithms (TALG)}, 13\penalty0 (4):\penalty0 1--28, 2017{\natexlab{b}}.

\bibitem[Amanatidis et~al.(2021)Amanatidis, Birmpas, Fusco, Lazos, Leonardi, and Reiffenh{\"{a}}user]{DBLP:conf/wine/AmanatidisBFLLR21}
Georgios Amanatidis, Georgios Birmpas, Federico Fusco, Philip Lazos, Stefano Leonardi, and Rebecca Reiffenh{\"{a}}user.
\newblock Allocating indivisible goods to strategic agents: Pure nash equilibria and fairness.
\newblock In \emph{{WINE}}, volume 13112 of \emph{Lecture Notes in Computer Science}. Springer, 2021.

\bibitem[Amanatidis et~al.(2023{\natexlab{a}})Amanatidis, Aziz, Birmpas, Filos-Ratsikas, Li, Moulin, Voudouris, and Wu]{amanatidis2023fair}
Georgios Amanatidis, Haris Aziz, Georgios Birmpas, Aris Filos-Ratsikas, Bo~Li, Herv{\'e} Moulin, Alexandros~A Voudouris, and Xiaowei Wu.
\newblock Fair division of indivisible goods: Recent progress and open questions.
\newblock \emph{Artificial Intelligence}, page 103965, 2023{\natexlab{a}}.

\bibitem[Amanatidis et~al.(2023{\natexlab{b}})Amanatidis, Birmpas, Lazos, Leonardi, and Reiffenh{\"{a}}user]{DBLP:conf/sigecom/AmanatidisBL0R23}
Georgios Amanatidis, Georgios Birmpas, Philip Lazos, Stefano Leonardi, and Rebecca Reiffenh{\"{a}}user.
\newblock Round-robin beyond additive agents: Existence and fairness of approximate equilibria.
\newblock In \emph{{EC}}, pages 67--87. {ACM}, 2023{\natexlab{b}}.

\bibitem[Asano and Umeda(2020)]{asano2020cake}
Takao Asano and Hiroyuki Umeda.
\newblock Cake cutting: An envy-free and truthful mechanism with a small number of cuts.
\newblock In \emph{31st International Symposium on Algorithms and Computation (ISAAC 2020)}. Schloss Dagstuhl-Leibniz-Zentrum f{\"u}r Informatik, 2020.

\bibitem[Aziz and Ye(2014)]{aziz2014cake}
Haris Aziz and Chun Ye.
\newblock Cake cutting algorithms for piecewise constant and piecewise uniform valuations.
\newblock In \emph{International conference on web and internet economics}, pages 1--14. Springer, 2014.

\bibitem[Aziz et~al.(2020)Aziz, Moulin, and Sandomirskiy]{aziz2020polynomial}
Haris Aziz, Herv{\'e} Moulin, and Fedor Sandomirskiy.
\newblock A polynomial-time algorithm for computing a pareto optimal and almost proportional allocation.
\newblock \emph{Operations Research Letters}, 48\penalty0 (5):\penalty0 573--578, 2020.

\bibitem[Aziz et~al.(2023)Aziz, Freeman, Shah, and Vaish]{AzizFrSh23}
Haris Aziz, Rupert Freeman, Nisarg Shah, and Rohit Vaish.
\newblock Best of both worlds: Ex ante and ex post fairness in resource allocation.
\newblock \emph{Operations Research}, 2023.

\bibitem[Babaioff and Manaker~Morag(2025)]{babaioff2024truthful}
Moshe Babaioff and Noam Manaker~Morag.
\newblock On truthful mechanisms without pareto-efficiency: Characterizations and fairness.
\newblock In Itai Ashlagi and Aaron Roth, editors, \emph{Proceedings of the 26th {ACM} Conference on Economics and Computation, {EC} 2025, Stanford University, Stanford, CA, USA, July 7-10, 2025}, page 447. {ACM}, 2025.
\newblock \doi{10.1145/3736252.3742569}.
\newblock URL \url{https://doi.org/10.1145/3736252.3742569}.

\bibitem[Babaioff et~al.(2021)Babaioff, Ezra, and Feige]{babaioff2021fair}
Moshe Babaioff, Tomer Ezra, and Uriel Feige.
\newblock Fair and truthful mechanisms for dichotomous valuations.
\newblock In \emph{Proceedings of the AAAI Conference on Artificial Intelligence}, volume~35, pages 5119--5126, 2021.

\bibitem[Babaioff et~al.(2022)Babaioff, Ezra, and Feige]{babaioff2022best}
Moshe Babaioff, Tomer Ezra, and Uriel Feige.
\newblock On best-of-both-worlds fair-share allocations.
\newblock In \emph{International Conference on Web and Internet Economics}, pages 237--255. Springer, 2022.

\bibitem[Barman and Krishnamurthy(2019)]{barman2019proximity}
Siddharth Barman and Sanath~Kumar Krishnamurthy.
\newblock On the proximity of markets with integral equilibria.
\newblock In \emph{Proceedings of the AAAI Conference on Artificial Intelligence}, volume~33, pages 1748--1755, 2019.

\bibitem[Barman and Verma(2022)]{barman2022truthful}
Siddharth Barman and Paritosh Verma.
\newblock Truthful and fair mechanisms for matroid-rank valuations.
\newblock In \emph{Proceedings of the AAAI Conference on Artificial Intelligence}, volume~36, pages 4801--4808, 2022.

\bibitem[Bei et~al.(2017)Bei, Chen, Huzhang, Tao, and Wu]{bei2017cake}
Xiaohui Bei, Ning Chen, Guangda Huzhang, Biaoshuai Tao, and Jiajun Wu.
\newblock Cake cutting: Envy and truth.
\newblock In \emph{IJCAI}, pages 3625--3631, 2017.

\bibitem[Bei et~al.(2020)Bei, Huzhang, and Suksompong]{bei2020truthful}
Xiaohui Bei, Guangda Huzhang, and Warut Suksompong.
\newblock Truthful fair division without free disposal.
\newblock \emph{Social Choice and Welfare}, 55\penalty0 (3):\penalty0 523--545, 2020.

\bibitem[Bei et~al.(2023)Bei, Tao, Wu, and Yang]{DBLP:journals/corr/abs-2308-08903}
Xiaohui Bei, Biaoshuai Tao, Jiajun Wu, and Mingwei Yang.
\newblock The incentive guarantees behind nash welfare in divisible resources allocation.
\newblock In \emph{{WINE}}, volume to apear of \emph{Lecture Notes in Computer Science}, page to appear. Springer, 2023.

\bibitem[Berger et~al.(2022)Berger, Cohen, Feldman, and Fiat]{berger2022almost}
Ben Berger, Avi Cohen, Michal Feldman, and Amos Fiat.
\newblock Almost full {EFX} exists for four agents.
\newblock In \emph{Proceedings of the AAAI Conference on Artificial Intelligence}, volume~36, pages 4826--4833, 2022.

\bibitem[Bogomolnaia and Moulin(2001)]{bogomolnaia2001new}
Anna Bogomolnaia and Herv{\'e} Moulin.
\newblock A new solution to the random assignment problem.
\newblock \emph{Journal of Economic theory}, 100\penalty0 (2):\penalty0 295--328, 2001.

\bibitem[Bouveret et~al.(2023)Bouveret, Gilbert, Lang, and M{\'e}rou{\'e}]{bouveret2023thou}
Sylvain Bouveret, Hugo Gilbert, J{\'e}r{\^o}me Lang, and Guillaume M{\'e}rou{\'e}.
\newblock Thou shalt not pick all items if thou are first: of strategyproof and fair picking sequences.
\newblock \emph{arXiv preprint arXiv:2301.06086}, 2023.

\bibitem[Brams and Taylor(1996)]{BramsTa96}
Steven~J. Brams and Alan~D. Taylor.
\newblock \emph{Fair Division: From Cake-Cutting to Dispute Resolution}.
\newblock Cambridge University Press, 1996.

\bibitem[Brams et~al.(2006)Brams, Jones, Klamler, et~al.]{brams2006better}
Steven~J Brams, Michael~A Jones, Christian Klamler, et~al.
\newblock Better ways to cut a cake.
\newblock \emph{Notices of the AMS}, 53\penalty0 (11):\penalty0 1314--1321, 2006.

\bibitem[Brandt et~al.(2016)Brandt, Conitzer, Endriss, Lang, and Procaccia]{BrandtCoEn16}
Felix Brandt, Vincent Conitzer, Ulle Endriss, J\'{e}r\^{o}me Lang, and Ariel~D. Procaccia, editors.
\newblock \emph{Handbook of Computational Social Choice}.
\newblock Cambridge University Press, 2016.

\bibitem[Br{\^a}nzei and Miltersen(2015)]{branzei2015dictatorship}
Simina Br{\^a}nzei and Peter~Bro Miltersen.
\newblock A dictatorship theorem for cake cutting.
\newblock \emph{IJCAI}, pages 482--488, 2015.

\bibitem[Bu et~al.(2023)Bu, Song, and Tao]{bu2023existence}
Xiaolin Bu, Jiaxin Song, and Biaoshuai Tao.
\newblock On existence of truthful fair cake cutting mechanisms.
\newblock \emph{Artificial Intelligence}, 319:\penalty0 103904, 2023.

\bibitem[Budish(2011)]{budish2011combinatorial}
Eric Budish.
\newblock The combinatorial assignment problem: Approximate competitive equilibrium from equal incomes.
\newblock \emph{Journal of Political Economy}, 119\penalty0 (6):\penalty0 1061--1103, 2011.

\bibitem[Budish et~al.(2013)Budish, Che, Kojima, and Milgrom]{budish2013designing}
Eric Budish, Yeon-Koo Che, Fuhito Kojima, and Paul Milgrom.
\newblock Designing random allocation mechanisms: Theory and applications.
\newblock \emph{American economic review}, 103\penalty0 (2):\penalty0 585--623, 2013.

\bibitem[Caragiannis et~al.(2009)Caragiannis, Kaklamanis, Kanellopoulos, and Kyropoulou]{caragiannis2009low}
Ioannis Caragiannis, Christos Kaklamanis, Panagiotis Kanellopoulos, and Maria Kyropoulou.
\newblock On low-envy truthful allocations.
\newblock In \emph{Algorithmic Decision Theory: First International Conference, ADT 2009, Venice, Italy, October 20-23, 2009. Proceedings 1}, pages 111--119. Springer, 2009.

\bibitem[Caragiannis et~al.(2019{\natexlab{a}})Caragiannis, Gravin, and Huang]{caragiannis2019envy}
Ioannis Caragiannis, Nick Gravin, and Xin Huang.
\newblock Envy-freeness up to any item with high {N}ash welfare: {T}he virtue of donating items.
\newblock In \emph{Proceedings of the ACM Conference on Economics and Computation (EC)}, pages 527--545, 2019{\natexlab{a}}.

\bibitem[Caragiannis et~al.(2019{\natexlab{b}})Caragiannis, Kurokawa, Moulin, Procaccia, Shah, and Wang]{CKMP+19}
Ioannis Caragiannis, David Kurokawa, Herv{\'e} Moulin, Ariel~D. Procaccia, Nisarg Shah, and Junxing Wang.
\newblock The unreasonable fairness of maximum {N}ash welfare.
\newblock \emph{ACM Transactions on Economics and Computation}, 7\penalty0 (3):\penalty0 1--32, 2019{\natexlab{b}}.

\bibitem[Chaudhury et~al.(2020)Chaudhury, Garg, and Mehlhorn]{chaudhury2020efx}
Bhaskar~Ray Chaudhury, Jugal Garg, and Kurt Mehlhorn.
\newblock {EFX} exists for three agents.
\newblock In \emph{Proceedings of the ACM Conference on Economics and Computation (EC)}, pages 1--19, 2020.

\bibitem[Chaudhury et~al.(2021)Chaudhury, Kavitha, Mehlhorn, and Sgouritsa]{chaudhury2021little}
Bhaskar~Ray Chaudhury, Telikepalli Kavitha, Kurt Mehlhorn, and Alkmini Sgouritsa.
\newblock A little charity guarantees almost envy-freeness.
\newblock \emph{SIAM Journal on Computing}, 50\penalty0 (4):\penalty0 1336--1358, 2021.

\bibitem[Chen et~al.(2013)Chen, Lai, Parkes, and Procaccia]{chen2013truth}
Yiling Chen, John~K Lai, David~C Parkes, and Ariel~D Procaccia.
\newblock Truth, justice, and cake cutting.
\newblock \emph{Games and Economic Behavior}, 77\penalty0 (1):\penalty0 284--297, 2013.

\bibitem[Christodoulou and Christoforidis(2024)]{christodoulou2024fair}
George Christodoulou and Vasilis Christoforidis.
\newblock Fair and truthful allocations under leveled valuations.
\newblock \emph{arXiv preprint arXiv:2407.05891}, 2024.

\bibitem[Cole et~al.(2013)Cole, Gkatzelis, and Goel]{cole2013mechanism}
Richard Cole, Vasilis Gkatzelis, and Gagan Goel.
\newblock Mechanism design for fair division: allocating divisible items without payments.
\newblock In \emph{Proceedings of the fourteenth ACM conference on Electronic commerce}, pages 251--268, 2013.

\bibitem[Conitzer et~al.(2017)Conitzer, Freeman, and Shah]{conitzer2017fair}
Vincent Conitzer, Rupert Freeman, and Nisarg Shah.
\newblock Fair public decision making.
\newblock In \emph{Proceedings of the 2017 ACM Conference on Economics and Computation}, pages 629--646, 2017.

\bibitem[Dobzinski and Dughmi(2013)]{dobzinski2013power}
Shahar Dobzinski and Shaddin Dughmi.
\newblock On the power of randomization in algorithmic mechanism design.
\newblock \emph{SIAM Journal on Computing}, 42\penalty0 (6):\penalty0 2287--2304, 2013.

\bibitem[Dobzinski et~al.(2023)Dobzinski, Oren, and Vondrak]{dobzinski2023fairness}
Shahar Dobzinski, Sigal Oren, and Jan Vondrak.
\newblock Fairness and incentive compatibility via percentage fees.
\newblock In \emph{Proceedings of the 24th ACM Conference on Economics and Computation}, pages 517--535, 2023.

\bibitem[Feldman et~al.(2023)Feldman, Mauras, Narayan, and Ponitka]{feldman2023breaking}
Michal Feldman, Simon Mauras, Vishnu~V Narayan, and Tomasz Ponitka.
\newblock Breaking the envy cycle: Best-of-both-worlds guarantees for subadditive valuations.
\newblock \emph{arXiv preprint arXiv:2304.03706}, 2023.

\bibitem[Foley(1967)]{Foley67}
Duncan~Karl Foley.
\newblock Resource allocation and the public sector.
\newblock \emph{Yale Economics Essays}, 7\penalty0 (1):\penalty0 45--98, 1967.

\bibitem[Freeman et~al.(2023)Freeman, Witkowski, Vaughan, and Pennock]{freeman2023equivalence}
Rupert Freeman, Jens Witkowski, Jennifer~Wortman Vaughan, and David~M Pennock.
\newblock An equivalence between fair division and wagering mechanisms.
\newblock \emph{Management Science}, 2023.

\bibitem[Garg and Psomas(2022)]{garg2022efficient}
Rohan Garg and Alexandros Psomas.
\newblock Efficient mechanisms without money: Randomization won't let you escape from dictatorships, 2022.

\bibitem[Halpern et~al.(2020)Halpern, Procaccia, Psomas, and Shah]{halpern2020fair}
Daniel Halpern, Ariel~D Procaccia, Alexandros Psomas, and Nisarg Shah.
\newblock Fair division with binary valuations: One rule to rule them all.
\newblock In \emph{Web and Internet Economics: 16th International Conference, WINE 2020, Beijing, China, December 7--11, 2020, Proceedings 16}, pages 370--383. Springer, 2020.

\bibitem[Hartman et~al.(2025)Hartman, Segal-Halevi, and Tao]{hartman2025s}
Eden Hartman, Erel Segal-Halevi, and Biaoshuai Tao.
\newblock It's not all black and white: Degree of truthfulness for risk-avoiding agents.
\newblock \emph{arXiv preprint arXiv:2502.18805}, 2025.

\bibitem[Hosseini and Larson(2019)]{hosseini2019multiple}
Hadi Hosseini and Kate Larson.
\newblock Multiple assignment problems under lexicographic preferences.
\newblock In \emph{Proceedings of the 18th International Conference on Autonomous Agents and MultiAgent Systems}, pages 837--845, 2019.

\bibitem[Huang et~al.(2024)Huang, Wang, Wei, and Zhang]{HUANG2024103491}
Haoqiang Huang, Zihe Wang, Zhide Wei, and Jie Zhang.
\newblock Bounded incentives in manipulating the probabilistic serial rule.
\newblock \emph{Journal of Computer and System Sciences}, 140:\penalty0 103491, 2024.
\newblock ISSN 0022-0000.
\newblock \doi{https://doi.org/10.1016/j.jcss.2023.103491}.
\newblock URL \url{https://www.sciencedirect.com/science/article/pii/S002200002300096X}.

\bibitem[Huang et~al.(2025)Huang, Tao, Yang, and Zhou]{huang2025incentive}
Haoqiang Huang, Biaoshuai Tao, Mingwei Yang, and Shengwei Zhou.
\newblock Incentive analysis of collusion in fair division.
\newblock \emph{arXiv preprint arXiv:2510.01689}, 2025.

\bibitem[Hv et~al.(2025)Hv, Ghosal, Nimbhorkar, and Varma]{hv2025efx}
Vishwa~Prakash Hv, Pratik Ghosal, Prajakta Nimbhorkar, and Nithin Varma.
\newblock Efx exists for three types of agents.
\newblock In \emph{Proceedings of the 26th ACM Conference on Economics and Computation}, pages 101--128, 2025.

\bibitem[Kurokawa et~al.(2016)Kurokawa, Procaccia, and Wang]{kurokawa2016can}
David Kurokawa, Ariel Procaccia, and Junxing Wang.
\newblock When can the maximin share guarantee be guaranteed?
\newblock In \emph{Proceedings of the AAAI Conference on Artificial Intelligence}, volume~30, 2016.

\bibitem[Kyropoulou et~al.(2020)Kyropoulou, Suksompong, and Voudouris]{kyropoulou2020almost}
Maria Kyropoulou, Warut Suksompong, and Alexandros~A Voudouris.
\newblock Almost envy-freeness in group resource allocation.
\newblock \emph{Theoretical Computer Science}, 841:\penalty0 110--123, 2020.

\bibitem[Li et~al.(2015)Li, Zhang, and Zhang]{li2015truthful}
Minming Li, Jialin Zhang, and Qiang Zhang.
\newblock Truthful cake cutting mechanisms with externalities: Do not make them care for others too much!
\newblock In \emph{Twenty-Fourth International Joint Conference on Artificial Intelligence}, 2015.

\bibitem[Li et~al.(2023)Li, Liu, Lu, and Tao]{li2023truthful}
Zihao Li, Shengxin Liu, Xinhang Lu, and Biaoshuai Tao.
\newblock Truthful fair mechanisms for allocating mixed divisible and indivisible goods.
\newblock In \emph{Proceedings of the Thirty-Second International Joint Conference on Artificial Intelligence}, pages 2808--2816, 2023.

\bibitem[Lipton et~al.(2004)Lipton, Markakis, Mossel, and Saberi]{lipton2004approximately}
Richard~J Lipton, Evangelos Markakis, Elchanan Mossel, and Amin Saberi.
\newblock On approximately fair allocations of indivisible goods.
\newblock In \emph{Proceedings of the 5th ACM Conference on Electronic Commerce}, pages 125--131, 2004.

\bibitem[Liu et~al.(2024)Liu, Lu, Suzuki, and Walsh]{liu2023mixed}
Shengxin Liu, Xinhang Lu, Mashbat Suzuki, and Toby Walsh.
\newblock Mixed fair division: A survey.
\newblock In \emph{Proceedings of the AAAI Conference on Artificial Intelligence}, volume~38, pages 22641--22649, 2024.

\bibitem[Mahara(2023)]{mahara2023extension}
Ryoga Mahara.
\newblock Extension of additive valuations to general valuations on the existence of {EFX}.
\newblock \emph{Mathematics of Operations Research}, 2023.

\bibitem[Manurangsi and Suksompong(2022)]{manurangsi2022almost}
Pasin Manurangsi and Warut Suksompong.
\newblock Almost envy-freeness for groups: Improved bounds via discrepancy theory.
\newblock \emph{Theoretical Computer Science}, 930:\penalty0 179--195, 2022.

\bibitem[Maya and Nisan(2012)]{maya2012incentive}
Avishay Maya and Noam Nisan.
\newblock Incentive compatible two player cake cutting.
\newblock In \emph{International Workshop on Internet and Network Economics}, pages 170--183. Springer, 2012.

\bibitem[Menon and Larson(2017)]{menon2017deterministic}
Vijay Menon and Kate Larson.
\newblock Deterministic, strategyproof, and fair cake cutting.
\newblock In \emph{Proceedings of the 26th International Joint Conference on Artificial Intelligence}, pages 352--358, 2017.

\bibitem[Momi(2017)]{momi2017efficient}
Takeshi Momi.
\newblock Efficient and strategy-proof allocation mechanisms in economies with many goods.
\newblock \emph{Theoretical Economics}, 12\penalty0 (3):\penalty0 1267--1306, 2017.

\bibitem[Mossel and Tamuz(2010)]{mossel2010truthful}
Elchanan Mossel and Omer Tamuz.
\newblock Truthful fair division.
\newblock In \emph{International Symposium on Algorithmic Game Theory}, pages 288--299. Springer, 2010.

\bibitem[Nash~Jr(1950)]{nash1950bargaining}
John~F Nash~Jr.
\newblock The bargaining problem.
\newblock \emph{Econometrica: Journal of the econometric society}, pages 155--162, 1950.

\bibitem[Ortega and Segal-Halevi(2022)]{ortega2022obvious}
Josu{\'e} Ortega and Erel Segal-Halevi.
\newblock Obvious manipulations in cake-cutting.
\newblock \emph{Social Choice and Welfare}, pages 1--20, 2022.

\bibitem[P{\'a}pai(2000)]{papai2000strategyproof}
Szilvia P{\'a}pai.
\newblock Strategyproof multiple assignment using quotas.
\newblock \emph{Review of Economic Design}, 5:\penalty0 91--105, 2000.

\bibitem[P{\'a}pai(2001)]{papai2001strategyproof}
Szilvia P{\'a}pai.
\newblock Strategyproof and nonbossy multiple assignments.
\newblock \emph{Journal of Public Economic Theory}, 3\penalty0 (3):\penalty0 257--271, 2001.

\bibitem[Plaut and Roughgarden(2020)]{plaut2020almost}
Benjamin Plaut and Tim Roughgarden.
\newblock Almost envy-freeness with general valuations.
\newblock \emph{SIAM Journal on Discrete Mathematics}, 34\penalty0 (2):\penalty0 1039--1068, 2020.

\bibitem[Procaccia(2013)]{procaccia2013cake}
Ariel~D Procaccia.
\newblock Cake cutting: Not just child's play.
\newblock \emph{Communications of the ACM}, 56\penalty0 (7):\penalty0 78--87, 2013.

\bibitem[Procaccia and Wang(2014)]{procaccia2014fair}
Ariel~D Procaccia and Junxing Wang.
\newblock Fair enough: Guaranteeing approximate maximin shares.
\newblock In \emph{Proceedings of the fifteenth ACM conference on Economics and computation}, pages 675--692, 2014.

\bibitem[Psomas and Verma(2022)]{DBLP:conf/nips/0001V22}
Alexandros Psomas and Paritosh Verma.
\newblock Fair and efficient allocations without obvious manipulations.
\newblock In \emph{NeurIPS}, 2022.

\bibitem[Robertson and Webb(1998)]{robertson1998cake}
Jack Robertson and William Webb.
\newblock \emph{Cake-cutting algorithms: Be fair if you can}.
\newblock CRC Press, 1998.

\bibitem[Satterthwaite and Sonnenschein(1981)]{satterthwaite1981strategy}
Mark~A Satterthwaite and Hugo Sonnenschein.
\newblock Strategy-proof allocation mechanisms at differentiable points.
\newblock \emph{The Review of Economic Studies}, 48\penalty0 (4):\penalty0 587--597, 1981.

\bibitem[Shende and Purohit(2023)]{shende2023strategy}
Priyanka Shende and Manish Purohit.
\newblock Strategy-proof and envy-free mechanisms for house allocation.
\newblock \emph{Journal of Economic Theory}, 213:\penalty0 105712, 2023.

\bibitem[Steinhaus(1948)]{Steinhaus48}
Hugo Steinhaus.
\newblock The problem of fair division.
\newblock \emph{Econometrica}, 16\penalty0 (1):\penalty0 101--104, 1948.

\bibitem[Steinhaus(1949)]{Steinhaus49}
Hugo Steinhaus.
\newblock Sur la division pragmatique.
\newblock \emph{Econometrica}, 17:\penalty0 315--319, 1949.

\bibitem[Sun and Chen(2024)]{sun2024randomized}
Ankang Sun and Bo~Chen.
\newblock Randomized strategyproof mechanisms with best of both worlds fairness and efficiency.
\newblock \emph{arXiv preprint arXiv:2408.01027}, 2024.

\bibitem[Suzuki et~al.(2023)Suzuki, Tamura, Yahiro, Yokoo, and Zhang]{suzuki2023strategyproof}
Takamasa Suzuki, Akihisa Tamura, Kentaro Yahiro, Makoto Yokoo, and Yuzhe Zhang.
\newblock Strategyproof allocation mechanisms with endowments and m-convex distributional constraints.
\newblock \emph{Artificial Intelligence}, 315:\penalty0 103825, 2023.

\bibitem[Svensson(1999)]{svensson1999strategy}
Lars-Gunnar Svensson.
\newblock Strategy-proof allocation of indivisible goods.
\newblock \emph{Social Choice and Welfare}, 16:\penalty0 557--567, 1999.

\bibitem[Tao(2022)]{tao2022existence}
Biaoshuai Tao.
\newblock On existence of truthful fair cake cutting mechanisms.
\newblock In \emph{Proceedings of the 23rd ACM Conference on Economics and Computation}, pages 404--434, 2022.

\bibitem[Tao and Yang(2023)]{tao2023fair}
Biaoshuai Tao and Mingwei Yang.
\newblock Fair and almost truthful mechanisms for additive valuations and beyond.
\newblock \emph{arXiv preprint arXiv:2306.15920}, 2023.

\bibitem[Troyan and Morrill(2020)]{troyan2020obvious}
Peter Troyan and Thayer Morrill.
\newblock Obvious manipulations.
\newblock \emph{Journal of Economic Theory}, 185:\penalty0 104970, 2020.

\bibitem[Varian(1973)]{varian1973equity}
Hal~R Varian.
\newblock Equity, envy, and efficiency.
\newblock 1973.

\bibitem[Xiao and Ling(2020)]{xiao2020algorithms}
Mingyu Xiao and Jiaxing Ling.
\newblock Algorithms for manipulating sequential allocation.
\newblock In \emph{Proceedings of the AAAI Conference on Artificial Intelligence}, volume~34, pages 2302--2309, 2020.

\end{thebibliography}

\clearpage
\appendix

\section{Proof of Theorem~\ref{thm:impossibility_deterministic}}
\label{append:impossibility_deterministic}
We adopt Theorem~3.6 in \citet{amanatidis2017truthful} that every truthful mechanism for two agents can be implemented as a \emph{picking-exchange} mechanism.
A picking mechanism contains two picking components $(N_1,N_2)$ that form a partition of $M$, and two sets of offers $\mathcal{O}_1,\mathcal{O}_2$ where each contains several subsets of $N_i$ for each agent, where $\bigcup_{T\in\mathcal{O}_i}=N_i$ and $\bigcap_{T\in\mathcal{O}_i}=\emptyset$.
In a picking mechanism, each agent $i\in\{1,2\}$ is allocated one offer with the largest utility from $\mathcal{O}_i$, and the remaining items in $N_i$ are allocated to agent $3-i$.
An exchange mechanism contains two exchange components $(E_1,E_2)$ that forms a partition of $M$, and a set of exchange deals $D=\{(T^1_1,T^2_1),\cdots,(T^1_k,T^2_k)\}$, where $T^i_j$ is a non-empty subset of $E_i$ and $T^i_{j_1}\cap T^i_{j_2}=\emptyset$ for any $1\le j_1,j_2\le k$.
In an exchange mechanism, each of the exchange deals $(T^1_j,T^2_j)$ that satisfies $v_i(T^i_j)<v_i(T^{3-i}_j)$ is exchanged between the two agents and $T^i_j$ is allocated to agent $i$ if it is not exchanged.
A picking-exchange mechanism combines the above parts while ensuring $(N_1,N_2,E_1,E_2)$ forms a partition of $M$.
It is not hard to see that a picking-exchange mechanism is truthful.

Given a truthful mechanism, we say agent $i$ \emph{controls} set $T$ as whenever $v_i(g)\ge v_i(M\setminus T)$ for all $g\in T$, $T$ will be allocated to $i$ by the mechanism.
Denote the set of maximal controlled sets of agent $i$ by $\mathcal{C}_i$, we have $\bigcup_{T\in \mathcal{C}_1}T \cup \bigcup_{T\in \mathcal{C}_2}T=M$ and $\bigcup_{T\in \mathcal{C}_1}T\cap \bigcup_{T\in \mathcal{C}_2}T=\emptyset$.
\citet{amanatidis2017truthful} show that a truthful mechanism can be implemented by a picking-exchange mechanism where $E_i=\bigcap_{T\in \mathcal{C}_i}T, N_i=\bigcup_{T\in \mathcal{C}_i}T\setminus E_i$, and $\mathcal{O}_i=\{T\setminus E_i| T\in\mathcal{C}_i\}$ for $i\in\{1,2\}$.

We now prove Theorem~\ref{thm:impossibility_deterministic} using this characterization.

We first notice that to achieve EF$^{+u}_{-v}$, each agent cannot control a set of items with a size larger than $u+v$.
Otherwise, assume that the set of the first $u+v+1$ items $\{g_1,\ldots,g_{u+v+1}\}$ is controlled by agent 1 where $u+v+1\le m$.
Consider the valuation profile where $v_1(g_j)=v_2(g_j)=m$ for $1\le j\le u+v+1$ and $v_1(g_j)=v_2(g_j)=1$ for $u+v+1<j\le m$, then agent 1 will receive the first $u+v+1$ items, violating EF$^{+u}_{-v}$ for agent 2.

Then, for a fixed picking-exchange mechanism, we may assume $0\le |E_i|=k_i\le u+v$.
Consider the following valuation profile where $1\gg\epsilon\gg\delta\gg\mu>0$ and each of the $N_i$ and $E_i$ could be empty.

\begin{table}[h]
    \centering
    \begin{tabular}{c|ccccc|ccccc|ccc|ccc}
    \hline
         & \multicolumn{5}{c|}{items in $N_1$} & \multicolumn{5}{c|}{items in $N_2$} & \multicolumn{3}{c|}{items in $E_1$} & \multicolumn{3}{c}{items in $E_2$} \\
         \hline 
        $v_1$ & $1+\epsilon$ & 1 & 1 & $\cdots$ & 1 & 1 & $\delta$ & $\delta$ & $\cdots$ & $\delta$ & $\mu$ & $\cdots$ & $\mu$ & $\mu$ & $\cdots$ & $\mu$ \\
        \hline
        $v_2$ & 1 & $\delta$ & $\delta$ & $\cdots$ & $\delta$ & $1+\epsilon$ & 1 & 1 & $\cdots$ & 1 & $\mu$ & $\cdots$ & $\mu$ & $\mu$ & $\cdots$ & $\mu$ \\
    \hline
    \end{tabular}
    \label{tab:impossibility_deterministic}
\end{table}

Given $|E_1|=k_1$, agent 1 will receive at most $u+v-k_1$ items from $N_1$ (as the size of each offer in $\mathcal{O}_1$ is at most $u+v-k_1$); we further assume agent 1 receives $x$ items from $N_2$.
When chosen $\epsilon,\delta$ and $\mu$ to be sufficiently small, we have $v_1(A_1)\le u+v+\epsilon-k_1+x\delta+k_1\mu\le u+v-k_1+1$ and $v_1(A_2)\ge |N_1|-(u+v-k_1)$.
To ensure EF$^{+u}_{-v}$ for agent 1, we have $|N_1|\le 3u+3v-2k_1+1$.
By symmetry, $|N_2|\le 3u+3v-2k_2+1$ to ensure EF$^{+u}_{-v}$ for agent 2.
Hence, no picking-exchange mechanism is EF$^{+u}_{-v}$ when $m\ge 6u+6v+2$.

\section{Proof of Proposition~\ref{prop:perfecttwo}}
\label{append:prop:perfecttwo}
    We assume without loss of generality that the number of items $m$ is an even number, for otherwise we can add a dummy item where both agents have value $0$.
    Let agent $1$ sort the items by descending values $g^{(1)}_1,\ldots,g^{(1)}_m$ where $v_1(g^{(1)}_1)\geq v_1(g^{(1)}_2)\geq\cdots\geq v_1(g^{(1)}_m)$, and let agent $2$ do the same with $g^{(2)}_1,\ldots,g^{(2)}_m$.
    Ties are broken arbitrarily.
    Notice that $(g^{(1)}_1,\ldots,g^{(1)}_m)$ is a permutation of $(g^{(2)}_1,\ldots,g^{(2)}_m)$.
    Based on agent $1$'s sorting, define the partition $(G^{(1)}_1,\ldots,G^{(1)}_{m/2})$ of $M$ where $G^{(1)}_j=\{g^{(1)}_{2j-1},g^{(1)}_{2j}\}$ for $j=1,\ldots,m/2$, and define the partition $(G^{(2)}_1,\ldots,G^{(2)}_{m/2})$ similarly for agent $2$.

    Next, we show that it is possible to find a partition $(X,Y)$ such that $|X\cap G^{(i)}_j|=|Y\cap G^{(i)}_j|=1$ for each $i\in \{1,2\}$ and $j=1,\ldots,m$. That is, for each two-item set $G^{(i)}_j$, exactly one item is in $X$ and exactly one item is in $Y$.
    We will show that both $(X,Y)$ and $(Y,X)$ are EF1 and such a partition can be found in polynomial time.

    To show that the allocation $(X,Y)$ and $(Y,X)$ are EF1, we will only show that $v_1(X)\geq v_1(Y\setminus\{g\})$ for some $g\in Y$. The other direction $v_1(Y)\geq v_1(X\setminus\{g\})$ as well as the analysis for agent $2$ are similar and are thus omitted.
    Let $x_j$ be the unique item in the set $X\cap G_j^{(1)}$ and $y_j$ be the unique item in the set $Y\cap G_j^{(1)}$.
    Then $X=\{x_1,\ldots,x_{m/2}\}$ and $Y=\{y_1,\ldots,y_{m/2}\}$.
    By agent $1$'s value-descending ordering of the items $g_1^{(1)},\ldots,g_m^{(1)}$ and the definition of $G^{(1)}_1,\ldots,G^{(1)}_{m/2}$, agent $1$ values $x_j$ weakly higher than any items in $G_{j+1}^{(1)}$.
    In particular, we have $v_1(x_j)\geq v_1(y_{j+1})$ for each $j=1,\ldots,m/2-1$.
    Therefore, by summing up all the items, we have $v_1(X)\geq v_1(Y\setminus\{y_1\})$.

    It now remains to show the existence of such a partition $(X,Y)$ and how it is computed.
    We construct a bipartite graph $G=(V_1,V_2,E)$ where $V_1$ contains $m/2$ vertices corresponding to $G^{(1)}_1,\ldots,G^{(1)}_{m/2}$, $V_2$ contains $m/2$ vertices corresponding to $G^{(2)}_1,\ldots,G^{(2)}_{m/2}$, and $E$ contains $m$ edges corresponding to the $m$ items such that an edge $g$ is incident to the vertex $G^{(i)}_j$ if $g\in G^{(i)}_j$.
    The bipartite graph $G$ constructed is $2$-regular, and a valid $2$-coloring of the edges corresponds to a valid partition $(X,Y)$.
    By Lemma~\ref{lem:matching}, such a $2$-coloring exists and can be found in polynomial time.

\section{Subtlety in Tie-Breaking for Fractional Division Rule in Sect.~\ref{sect:three-fractional}}
\label{append:subtlety}
We will show that, when selecting the two highest-value items respectively for agents $1$ and $2$ (where agent $1$ and agent $2$ receive fractions of $\frac23$ respectively), tie-breaking by a consistent item index order cannot guarantee truthfulness.
Therefore, the delicate tie-breaking rule in Algorithm~\ref{alg:tiebreaking} is necessary. 

Assume that the tie-breaking rule is defined where agent $1$ first chooses the item with the highest value and the smallest index, and then agent $2$ chooses one from the remaining two items with the highest value and the smallest index.
Consider the instance with three items $M=\{g_1,g_2,g_3\}$.
Agent $1$ believes both $g_1$ and $g_2$ have the highest value (i.e., $v_1(g_1)=v_1(g_2)>v_1(g_3)$), and agent $2$ believes both $g_2$ and $g_3$ have the highest value (i.e., $v_2(g_2)=v_2(g_3)>v_2(g_1)$).
Thus, the items should be allocated according to Type I.
Under the tie-breaking rule we defined, agent $1$ will receive $\frac23$ fraction of item $g_1$ and $\frac13$ fraction of item $g_3$.
Agent $2$ will receive $\frac23$ fraction of item $g_2$ and $\frac13$ fraction of item $g_3$.

However, if agent $1$ misreports her valuation such that only item $g_2$ has the highest value, agent $1$ will receive $\frac23$ fraction of item $g_2$ and $\frac13$ fraction of item $g_1$.
Agent $2$ will receive $\frac23$ fraction of item $g_3$ and $\frac13$ fraction of item $g_1$.
As $v_1(g_1)>v_1(g_3)$, the misreport is beneficial.

\section{Proof of Theorem~\ref{thm:incom_efkx_n}}\label{append:incom_efkx_n}
The proof is similar to that of Theorem~\ref{thm:incom_efkx}.
We begin by introducing the constraints on EF$_{-k}^{+0}$X allocations in a series of instances with $n$ agents and $m+1$ items $M=\{g_0,\ldots,g_m\}$, where $m$ is sufficiently large compared to $n$ and $k$.
Still, we write $M'=\{g_1,\ldots,g_m\}$.

We similarly show that each agent needs to receive at least $k$ items from $M'$, which allows us to focus on items within $M'$ when we consider removing items.
\begin{proposition}\label{prop:incom_efkx_n_1}
    Consider the following instance where $x>n-1$ and the valuation functions of other agents are arbitrary.
    When $m>\max\left\{x, \frac{xk(2n-1)}{x-(n-1)}\right\}$, any EF$_{-k}^{+0}$X allocation must allocate at least $k$ items from $M'$ to agent $1$.
\end{proposition}
\begin{center}
    \begin{tabular}{cccccc}
        \hline
        & $g_0$ & $g_1$ & $g_2$ & $\cdots$ & $g_m$ \\
        \hline
        $v_1$ & $1$ & $x/m$ & $x/m$ & $\cdots$ & $x/m$ \\
        \hline
    \end{tabular}
\end{center}
\begin{proof}
    Assume that agent $1$ receives $k'<k$ items from $M'$.
    We will show that the allocation is not EF$_{-k}^{+0}$X for agent $1$ even if agent $1$ also receives $g_0$. 
    Let $A_1=\{g_0,g_1,\ldots,g_{k'}\}$.
    By the pigeonhole principle, there exists some agent $i\in \{2,\ldots,n\}$ that receives at least $\frac{m-k'}{n-1}$ items from $M'$.
    Then, we have
    $$v_1(A_1)=1+\frac{k'x}{m}<1+\frac{kx}{m}<\frac{x}{m}\left(\frac{m-k}{n-1}-k\right)<\frac{x}{m}\left(\frac{m-k'}{n-1}-k\right)\le v_1(A_i\setminus S_i),$$
    where the second inequality holds due to $m>\frac{xk(2n-1)}{x-(n-1)}$, and $S_i$ is any subset of $A_i$ with at most $k$ items.
    Therefore, the allocation is EF$_{-k}^{+0}$X for agent $1$ only if at least $k$ items from $M'$ are allocated to her.
\end{proof}

We next show that item $g_0$ must be allocated to agent $1$ in any EF$_{-k}^{+0}$X allocation.
\begin{proposition}\label{prop:incom_efkx_n_2}
    Consider the following instance where $n-1<x_1<x_2\le x_3$.
    An allocation is EF$_{-k}^{+0}$X only if $g_0\in A_1$ when $m>\max\left\{x_3, \frac{x_1k(2n-1)}{x_1-(n-1)}, \frac{2x_1x_2k}{x_2-x_1}\right\}$.
\end{proposition}
\begin{center}
    \begin{tabular}{cccccc}
        \hline
        & $g_0$ & $g_1$ & $g_2$ & $\cdots$ & $g_m$ \\
        \hline
        $v_1$ & $1$ & $x_1/m$ & $x_1/m$ & $\cdots$ & $x_1/m$ \\
        $v_2$ & $1$ & $x_2/m$ & $x_2/m$ & $\cdots$ & $x_2/m$ \\
        $v_3$ & $1$ & $x_3/m$ & $x_3/m$ & $\cdots$ & $x_3/m$ \\
        $\vdots$ & $\vdots$ & $\vdots$ & $\vdots$ & $\vdots$ & $\vdots$ \\
        $v_n$ & $1$ & $x_3/m$ & $x_3/m$ & $\cdots$ & $x_3/m$ \\
        \hline
    \end{tabular}
    \label{tab:incom_efkx_p2}
\end{center}
\begin{proof}
    We first consider the case that $g_0$ is allocated to agent $2$.
    Assume agent $1$ receives $y_1$ items from $M'$, and agent $2$ receives additionally $y_2$ items from $M'$.
    According to Proposition~\ref{prop:incom_efkx_n_1}, we have $y_1\ge k$ and $y_2\ge k$.
    To guarantee EF$_{-k}^{+0}$X for agent $1$, we consider the case where agent $1$ removes $k$ items in $A_2\cap M'$ from $A_2$, which is given by $$\frac{x_1}{m}y_1\ge 1+\frac{x_1}{m}(y_2-k).$$
    This implies $y_2\le y_1+k-m/x_1$.
    However, agent $2$ will envy agent $1$ even if removing a set $S_1$ of $k$ items from $A_1$ according to the following inequality,
        $$v_2(A_2)=1+\frac{x_2}{m}y_2\le 1+\frac{x_2}{m}(y_1+k-\frac{m}{x_1})< \frac{x_2}{m}(y_1-k)=v_2(A_1\setminus S_1),$$
    where the second inequality holds due to $m>2x_1x_2k/(x_2-x_1)$.

    The case where $g_0$ is allocated to an agent $i\in\{3,\ldots,n\}$ follows by the same analysis.
    The inequality holds as $m>2x_1x_2k/(x_2-x_1)\ge 2x_1x_3k/(x_3-x_1)$ where $x_3\ge x_2$.
    Details are omitted here.
\end{proof}

We are now ready to prove Theorem~\ref{thm:incom_efkx_n} based on the above two propositions.
\paragraph{Proof of Theorem~\ref{thm:incom_efkx_n}.} 
Assume that there exists a truthful and EF$_{-k}^{+0}$X randomized mechanism $(\calF,\calD)$.
We begin with the following instance with $n$ agents and $m+1$ items where $m>24n^2k$.
\begin{center}
    \begin{tabular}{cccccc}
        \hline
        & $g_0$ & $g_1$ & $g_2$ & $\cdots$ & $g_m$ \\
        \hline
        $v_1$ & $1$ & $n/m$ & $n/m$ & $\cdots$ & $n/m$ \\
        $v_2$ & $1$ & $4n/m$ & $4n/m$ & $\cdots$ & $4n/m$ \\
        $\vdots$ & $\vdots$ & $\vdots$ & $\vdots$ & $\vdots$ & $\vdots$ \\
        $v_n$ & $1$ & $4n/m$ & $4n/m$ & $\cdots$ & $4n/m$ \\
        \hline
    \end{tabular}
    \label{tab:incom_efkx_p2}
\end{center}

Denote $\{g_1,\ldots,g_m\}$ by $M'$.
For any allocation that satisfies EF$_{-k}^{+0}$X, according to Propositions~\ref{prop:incom_efkx_n_1}, each agent needs to be allocated at least $k$ items from $M'$.
According to Proposition~\ref{prop:incom_efkx_n_2}, we have $g_0\in A_1$.
Denote by $\alpha_i$ the (possibly fractional) number of items that agent $i$ receives from $M'$ in the fractional allocation generated by $\calF$.

Let $\beta_2=\frac{2n+1}{2n^2}m-\frac{n-1}{n}k$.
We claim that if $\calF$ is truthful, then $\alpha_2\ge \beta_2$.
Otherwise, agent $2$ can misreport by lowering her value of each item $g\in M'$ from $4n/m$ to $2n/m$, which leads to the following instance. 
\begin{center}
    \begin{tabular}{cccccc}
        \hline
        & $g_0$ & $g_1$ & $g_2$ & $\cdots$ & $g_m$ \\
        \hline
        $v_1$ & $1$ & $n/m$ & $n/m$ & $\cdots$ & $n/m$ \\
        $v_2$ & $1$ & $2n/m$ & $2n/m$ & $\cdots$ & $2n/m$ \\
        $v_3$ & $1$ & $4n/m$ & $4n/m$ & $\cdots$ & $4n/m$ \\
        $\vdots$ & $\vdots$ & $\vdots$ & $\vdots$ & $\vdots$ & $\vdots$ \\
        $v_n$ & $1$ & $4n/m$ & $4n/m$ & $\cdots$ & $4n/m$ \\
        \hline
    \end{tabular}
    \label{tab:incom_efkx_p2}
\end{center}
In any EF$_{-k}^{+0}$X allocation of the above instance, agent $1$ will receive $g_0$.
In addition, agent $2$ will receive no less than $\beta_2$ items from $M'$.
To see this, assume that agent $2$ receives $y_2$ items from $M'$, and we will derive a lower bound on $y_2$ for any EF$_{-k}^{+0}$X allocation by considering the worst-case scenario to agent $2$.
In particular, we assume that envy from agent $2$ to every other agent is eliminated only after exactly $k$ items in $M'$ are removed from the other's bundle, which represents the most unfavorable case.
Under such a scenario, agent $1$ will receive no more than $y_2-m/2n+k$ items, and each agent $i\in\{3,\ldots,n\}$ will receive $y_2+k$ items.
As every item in $M'$ is allocated, we have $(y_2-m/2n+k)+y_2+(n-2)(y_2+k)\ge m$, which gives us $y_2\ge \frac{2n+1}{2n^2}m-\frac{n-1}{n}k=\beta_2$.
In the new instance, the expected (possibly fractional) number of items that agent $2$ receives from $\calF$ is $\alpha'_2\ge y_2$.
Hence, when $\alpha_2< \beta_2$, agent $2$'s utility will strictly improve by $\alpha'_2-\alpha_2>0$ from such a misreport.

Given that $\alpha_2\ge\beta_2$, we now derive an upper bound on $\alpha_1$ in the original instance.
As $\calF$ is EF$_{-k}^{+0}$X-realizable, denote the randomized allocations as $\{(p_t,\calA_t)\}_{t=1,\ldots,T}$.
Assume that in each allocation $\calA_t$, agent $2$ receives $y_2^t$ items, where $\sum_{t=1,\ldots,T}p_ty_2^t=\alpha_2$.
Each agent $i\in\{3,\ldots,n\}$ needs to receive at least $y_2^t-k$ items to ensure EF$_{-k}^{+0}$X.
Therefore, the maximum number of items that agent $1$ can receive from $M'$ is $m-y_2^t-(n-2)(y_2^t-k)=m-(n-1)y_2^t+(n-2)k$ in $\calA_t$.
This gives the upper bound of $\alpha_1$ as 
\begin{align*}
    \alpha_1 &\le \sum_{t=1}^T p_t \left(m-(n-1)y_2^t+(n-2)k\right) = m-(n-1)\alpha_2+(n-2)k \\
        &\le m-(n-1)\beta_2+(n-2)k = \frac{n+1}{2n^2}m+\frac{2n^2-4n+1}{n}k.
\end{align*}
Let $\beta_1=\frac{n+1}{2n^2}m+\frac{2n^2-4n+1}{n}k$.

Given that $\alpha_1\le\beta_1$, we provide a beneficial misreport for agent $1$ by increasing her value of each item $g\in M'$ from $n/m$ to $3n/m$, which leads to the following instance.
\begin{center}
    \begin{tabular}{cccccc}
        \hline
        & $g_0$ & $g_1$ & $g_2$ & $\cdots$ & $g_m$ \\
        \hline
        $v_1$ & $1$ & $3n/m$ & $3n/m$ & $\cdots$ & $3n/m$ \\
        $v_2$ & $1$ & $4n/m$ & $4n/m$ & $\cdots$ & $4n/m$ \\
        $\vdots$ & $\vdots$ & $\vdots$ & $\vdots$ & $\vdots$ & $\vdots$ \\
        $v_n$ & $1$ & $4n/m$ & $4n/m$ & $\cdots$ & $4n/m$ \\
        \hline
    \end{tabular}
    \label{tab:incom_efkx_p2}
\end{center}
We will show that in any EF$_{-k}^{+0}$X allocation of the above instance, agent $1$ will receive $g_0$ as well as more than $\beta_1$ items from $M'$.
To see this, assume that agent $1$ receives $y_1$ items from $M'$, and we similarly derive a lower bound on $y_1$ for any EF$_{-k}^{+0}$X allocation by considering her worst-case scenario.
To ensure that agent $1$ will not envy the others after removing exactly $k$ items, each agent $i\in\{2,\ldots,n\}$ will receive no more than $y_1+m/3n+k$ items.
Due to $y_1+(n-1)(y_1+m/3n+k)\ge m$, we have $y_1\ge \frac{2n+1}{3n^2}m-\frac{n-1}{n}k$.
Therefore, in the new instance, the expected (possibly fractional) number of items that agent $1$ receives from an EF$_{-k}^{+0}$X-realizable rule is $\alpha'_1\ge \frac{2n+1}{3n^2}m-\frac{n-1}{n}k$.

As $m>24n^2k$, it can be verified that $\alpha'_1> \alpha_1$, which implies that agent $1$'s expected utility will increase due to such a misreport.
Therefore, we conclude the incompatibility between truthfulness and EF$_{-k}^{+0}$X.

\section{Proof of Proposition~\ref{prop:POtruthful}}
\label{append:PObi}
A key observation of the (possibly fractional) allocation $(X_1',\ldots,X_n')$ that maximizes the Nash welfare is that, for any two agents $i$ and $j$, if $|X_i'|< |X_j'|$, then $v_i(g)=q$ for any item $g$ that has some fraction included in $X_j'$.
Otherwise, if some fraction of $g$ with $v_i(g)=p$ is included in $X_j'$, the allocation $(X_1',\ldots,X_n')$ cannot be a maximum Nash welfare solution to the profile $(v_1',\ldots,v_n')$, as moving some fraction of $g$ from $X_j'$ to $X_i'$ strictly improves the Nash welfare.

Based on the above observation, we now provide an equivalent interpretation of computing the maximum Nash welfare allocation $(X_1',\ldots,X_n')$ in the first phase of the division rule $\calF$ in Sect.~\ref{sect:PObi}.
Given a set of agents $S$ and a set of items $R\subseteq M$, let $C^R(S)$ denote a set of items $g\in R$ where $g\in C^R(S)$ if and only if there is at least one agent $i\in S$ such that $v_i(g)=p$.
For $R=M$, we simply write $C(S)$ for $C^R(S)$.
Let $|S|$ and $|C^R(S)|$ denote the number of the agents in $S$ and the items in $C^R(S)$ respectively. 
Let $N_u$ denote the set of agents that have not received any item and have value $p$ to some unallocated items.
Let $R$ denote the set of items that are still unallocated, and we will ensure $R$ contains only integral items throughout the procedure (but the allocation of $M\setminus R$ to agents in $N\setminus N_u$ may be fractional).
The mechanism starts from $R=M$ and iteratively finds a group of agents $S=\argmin\limits_{S\subseteq N_u}\frac{|C^R(S)|}{|S|}$ and (possibly fractionally) allocates $\frac{|C^R(S)|}{|S|}$ units of items to each agent $i\in S$, where we guarantee that each agent $i$ only receives those items $g$ with $v_i(g)=p$ (this is always possible; otherwise, it is easy to see $S$ cannot minimizes $\frac{|C^R(S)|}{|S|}$).
We then remove $S$ from $N_u$, update $R$, and repeat the procedure.
Let $S_1,\ldots,S_K$ be the sets of agents iteratively chosen by the mechanism.
At the beginning of the $k$-th iteration, the set of items allocated is exactly $C(S_1\cup S_2\cup\cdots\cup S_{k-1})$, let $R_k=M\setminus C(S_1\cup S_2\cup\cdots\cup S_{k-1})$ be the set of unallocated items at this moment (in particular, $R_1=M$). 
We can prove that $\frac{|C^{R_1}(S_1)|}{|S_1|}\le\frac{|C^{R_2}(S_2)|}{|S_2|}\le\cdots\le \frac{|C^{R_K}(S_K)|}{|S_K|}$.
The proof is similar to \citet{chen2013truth}.
Verbally, the agent who receives items later in the first phase will not receive fewer items than the agent who receives items earlier.

By the division rule $\calF$, each agent will receive exactly $L=m/n$ units of items after the three phases.
Therefore, for agent $i$, if $|X_i'|\ge L$ (i.e., agent $i$ will receive no less than $L$ units of items with value $p$ in the first phase under truthful report), there is no incentive to misreport as agent $i$ has already received the highest possible value.
Hence, we will only focus on agents that receive less than $L$ units of items in the first phase.

From now on, we will analyze the incentive of a particular agent $i$, and we will stick to the following assumption in the rest of this section.

\noindent\textbf{Assumption: } $|X_i'|<L$ when agents report truthfully.

We now show the truthfulness of $\calF$.
The following proof consists of two parts.
In Sect.~\ref{append:Fg_mechanism}, we will introduce a new hypothetical division rule $\calF^g_i$ in which agent $i$'s misreporting is, intuitively, more beneficial.
The hypothetical division rule $\calF_i^g$ is assumed to know agent $i$'s true valuation function.
That is, $\calF_i^g$ takes $n+1$ valuation functions as inputs: the true valuation function $v_i$ for agent $i$, the reported valuation function $u_i$ for agent $i$, and the valuation functions of the remaining $n-1$ agents.

We will show in Proposition~\ref{prop:Fg_reduce} that if there is a beneficial misreport for agent $i$ under the original division rule $\calF$, then there is also a beneficial misreport for her under the new division rule $\calF^g_i$.
Equivalently, for agent $i$, the truthfulness of the new division rule $\calF^g_i$ implies the truthfulness of the original division rule $\calF$.
In Sect.~\ref{append:Fg_truthful}, we will show that $\calF^g_i$ is truthful.
Combining the two parts, we conclude that $\calF$ is truthful for agent $i$.
Since $i$ is an arbitrary agent, we conclude Proposition~\ref{prop:POtruthful}.

\subsection{A New Division Rule $\calF_i^g$ and Its Relation to $\calF$}
\label{append:Fg_mechanism}

We begin by describing the hypothetical division rule $\calF_i^g$.
Let $D=\{g:v_i(g)=p,u_i(g)=q\}$.
Upon receiving the reported valuation function $u_i$ of agent $i$ and the valuation functions of the remaining $n-1$ agents, the rule $\calF_i^g$ does the same as $\calF$ by iteratively choosing agent sets $S_1,S_2,\ldots,S_K$ and compute the allocation $(X_1',\ldots,X_n')$ (which maximizes the Nash welfare if value $p$ is treated as $1$ and value $q$ is treated as $0$) in the first phase.
If $|X_i'|\geq L$, we let $\calF_i^g$ do exactly the same as $\calF$.
Otherwise, we make the following changes to $\calF_i^g$.
When deciding the allocation of $C^{R_k}(S_k)$ to agents in $S_k$ in each iteration and when truncating the bundles with size larger than $L$, the rule $\calF_i^g$ does them in a way that maximizes the (possibly fractional) number of truncated items in $D$, i.e., $\calF_i^g$ does the best to reserve the items in $D$ to the third phase.
(Notice that, under $\calF$, the allocation of $C^{R_k}(S_k)$ to $S_k$ and the truncation are done in an arbitrary consistent way.)
In addition, for those $S_k$ such that $|C^{R_k}(S_k)|/|S_k|\leq L$, i.e., those agents whose bundles are not truncated in the second phase, we require that $\calF_i^g$ handles the allocation of $C^{R_k}(S_k)$ to $S_k$ in exactly the same way as it is in $\calF$.
This finishes the description of $\calF_i^g$.

As a remark, the division rule $\calF_i^g$ needs to know $D$, which depends on agent $i$'s true valuation function $v_i$.
This is why $\calF_i^g$ is a ``hypothetical'' division rule.

Intuitively, by lying that a high-valued item is low-valued, agent $i$ would like these items to be allocated to her in the third phase.
Therefore, agent $i$ hopes that more items in $D$ can be allocated in the third phase, and $\calF_i^g$ does exactly this for agent $i$. 

We now prove Proposition~\ref{prop:Fg_reduce}.
This allows us to reduce the truthfulness of $\calF$ to the truthfulness of $\calF_i^g$.
\begin{proposition}\label{prop:Fg_reduce}
    If there is no beneficial misreporting under $\calF_i^g$, there is also no beneficial misreporting under $\calF$.
\end{proposition}
\begin{proof}
    We will prove the contra-positive.
    We will show that, if a beneficial misreport $u_i$ for agent $i$ exists under the division rule $\calF$, the same misreport $u_i$ is beneficial for agent $i$ under the division rule $\calF_i^g$.
    Suppose the valuation profile $(v_1,\ldots,v_{i-1},u_i,v_{i+1},\ldots,v_n)$ is given as input to both $\calF$ and $\calF_i^g$.
    We consider two cases.
    Notice that the length of agent $i$'s allocation in the first phase, $|X_i'|$, is the same under both mechanisms.
    We discuss two cases: $|X_1'|\geq L$ and $|X_1'|<L$.

    Suppose $|X_i'|\geq L$. Agent $i$ ends up receiving the same allocation under both $\calF$ and $\calF_i^g$ by our definition.

    Suppose $|X_i'|<L$, which implies that $X_i'=X_i''$ is the same under both $\calF$ and $\calF_i^g$ (by our definition).
    Since both $\calF$ and $\calF_i^g$ use the same iterative procedure, each agent receives the same \emph{length} under both rules in the first and the second phases.
    Thus, the values $T=m-\sum_{i=1}^n|X_i''|$ and $\alpha_i=\frac{L-|X_i''|}T$ are also the same under both $\calF$ and $\calF_i^g$.
    Agent $i$'s utility is then given by $v_i(X_i')+\alpha_i(|D_\truc|\cdot p+(T-|D_\truc|)\cdot q)$, where $D_\truc$ is the set of (possibly fractional) items in $D$ that is truncated and allocated in the second and the third phases.
    The utility is maximized with maximum $|D_\truc|$, and $\calF_i^g$ maximizes $|D_\truc|$ by definition.
    
    Thus, in both cases, by the same misreporting $u_i$, the utility of agent $i$ in $\calF_i^g$ is weakly higher than in $\calF$.
    If misreporting is beneficial under $\calF$, the same misreporting is beneficial under $\calF_i^g$.
\end{proof}

\subsection{Truthfulness of $\calF_i^g$ for agent $i$}
\label{append:Fg_truthful}
Since agent $i$ will receive exactly $L$ units of items after the three phases, the truthfulness of $\calF_i^g$ follows from the claim that agent $i$ will not receive more (possibly fractional) items with value $p$ in the three phases.

The proof of the claim above consists of two steps:
\begin{itemize}
    \item Step 1: we first show that if agent $i$ with valuation $v_i$ can benefit from a misreport to $u_i$ where there exist items $g$ such that $v_i(g)=q$ and $u_i(g)=p$, we can construct another beneficial misreport for agent $i$ where such deviation from $q$ to $p$ does not exist.
    \item Step 2: We then show that without deviation from $q$ to $p$, any deviation of items from value $p$ to $q$ is also not beneficial.
\end{itemize}
Hence, we conclude that there is no beneficial misreport, which guarantees the truthfulness of $\calF_i^g$.
In the following, we extend the notations defined for $\calF$ to $\calF_i^g$.

\paragraph{Step 1.}
First of all, if agent $i$ receives a length of at least $L$ after the first phase for reporting $u_i$, this misreporting is obviously non-beneficial: we have assumed agent $i$ receives a length of less than $L$ in the first phase when reporting truthfully; the truthfulness of maximum Nash welfare mechanism for binary valuations (proved by \citet{chen2013truth}) implies agent $i$ cannot receive more high-valued items by reporting $u_i$.
From now on, we assume agent $i$ receives a length of less than $L$ for reporting $u_i$.

Let $S_k$ be the group containing agent $i$ when agent $i$ reports $u_i$, and let $C^{R_k}(S_k)$ be those items that are allocated at the $k$-th iteration.
Let
$$w_i'(g)=\left\{\begin{array}{ll}
    q & \mbox{if }g\notin C^{R_k}(S_k) \\
    u_i(g) & \mbox{otherwise} 
\end{array}\right..$$
It is clear that reporting $w_i'$ leads to the same allocation as reporting $u_i$, as the iterative procedure of selecting $S_k$ and allocating $C^{R_k}(S_k)$ is exactly the same as if $u_i$ were reported.
Therefore, we will assume agent $i$ has reported $w_i'$ instead of $u_i$.

Let $D=\{g:w_i'(g)=q,v_i(g)=p\}$ and $E=\{g:w_i'(g)=p,v_i(g)=q\}$.
We aim to show that
$$w_i(g)=\left\{\begin{array}{ll}
    q & \mbox{if }g\in E \\
    w_i'(g) & \mbox{otherwise} 
\end{array}\right.$$
is at least as good as reporting $w_i'$ for agent $i$.
Notice that this will conclude the proof of this part: in $w_i$, no low-valued item is reported as high-valued.

Let $D_\truc\subseteq D$ be the set of items in $D$ that are truncated in the second phase, where $D_\truc$ may contain fractional items.
Let $X_i'$ be the allocation of agent $i$ in the first phase when reporting $w_i'$.
Let $F=X_i'\cap E$, and notice that $F$ may also contain fractional items.
We first show that, by reporting $w_i$ instead of $w_i'$, the following two properties hold:
\begin{enumerate}
    \item $|D_\truc|$ increases;
    \item The overall size of the truncation $T=m-\sum_{j=1}^n|X_j''|$ increases, but by an length of at most $|F|$.
\end{enumerate}

To see the above two properties intuitively, notice that agents in $S_1\cup\cdots\cup S_{k-1}$ have value $q$ on items in $E$, and thus in $F$.
When agent $i$ reports $w_i$ instead, items in $F$ will be reallocated to agents in $S_k\cup S_{k+1}\cup\cdots\cup S_K$.
Notice that we have assumed $|X_i|<L$, so bundles for the agents in $S_k$ are not large enough to be truncated.
Reallocating $F$ to agents in $S_k\cup S_{k+1}\cup\cdots\cup S_K$ can only make more items truncated, and the extra length truncated is at most $|F|$.

Now we prove these two properties formally.
When $F$ is removed from the item set $M$ (items are treated as divisible), consider the maximum Nash welfare allocation for the resource set $M\setminus F$ for $p$ and $q$ treated as $1$ and $0$ respectively.
By resource monotonicity (a well-known property for maximum Nash welfare allocation), each agent receives less value than it is in the case where $M$ is allocated.
Moreover, it is easy to verify by the iterative procedure that the allocations for agents in $S_{k+1}\cup S_{k+2}\cup\cdots\cup S_K$ remain unchanged.
To see this, first notice that agents in $S_1\cup S_2\cup\cdots\cup S_k$ value all items in $C(S_{k+1}\cup S_{k+2}\cup\cdots \cup S_K)$ as $q$ (or $0$ after the treatment).
Therefore, in the maximum Nash welfare allocation, a superset of $C(S_{k+1}\cup S_{k+2}\cup\cdots \cup S_K)$ is allocated to agents in $S_{k+1}\cup S_{k+2}\cup\cdots\cup S_K$.
By resource monotonicity, agents in $S_{k+1}\cup S_{k+2}\cup\cdots\cup S_K$ receive weakly more value than before.
On the other hand, this superset cannot be proper: we have shown that each of the $n$ agents cannot receive more value in the allocation of $M\setminus F$ compared with the allocation of $M$.
Therefore, agents in $S_{k+1}\cup S_{k+2}\cup\cdots\cup S_K$ receives exactly $C(S_{k+1}\cup S_{k+2}\cup\cdots \cup S_K)$, and the allocations for agents in $S_{k+1}\cup S_{k+2}\cup\cdots\cup S_K$ remain unchanged.
Furthermore, in the allocation of $M\setminus F$, the overall size of the truncated items $T$ remains unchanged, as the agents whose bundles are truncated are those in $S_\ell\cup S_{\ell+1}\cup\cdots \cup S_K$ for some $\ell>k$ (recall that we have shown that the sizes of bundles for agents in $S_k$ are not large enough to be truncated).

Now, consider the scenario where $F$ is added back but agent $i$ has value $q$ (or $0$ after the treatment) on items in $F$.
By resource monotonicity again, the value received by each agent is weakly increased.
The overall size of the truncation thus increases, and it cannot be increased by a size of more than $|F|$, for otherwise some agent in $S_1\cup S_2\cup\cdots\cup S_{\ell-1}$ must have received less value.
This proves property 2.

To prove property 1, first recall that the allocations for agents in $S_{k+1}\cup S_{k+2}\cup\cdots\cup S_K$ remain unchanged when $F$ is removed from the resource set, so the allocations for the agents in $S_\ell\cup S_{\ell+1}\cup\cdots \cup S_K$ with truncated bundles are also unchanged since $\ell>k$.
Therefore, when $F$ is removed, $|D_\truc|$ remains unchanged in the new allocation.

Next, we describe an iterative procedure to add $F$ back while maintaining maximum Nash welfare, where the procedure resembles resource monotonicity.
Each iteration of the procedure involves allocating parts of $F$ and moving some part of an agent's allocation to another agent, and we will show that the truncated part of $D$ remains truncated during these.
To describe the procedure, we start with the maximum Nash welfare allocation of $M\setminus F$ and define a directed graph with $n+1$ vertices, where the $n$ vertices represent the $n$ agents, and the last vertex represents the pool $F$ of the unallocated (possibly fractional) items.
We build a directed edge from agent $j_1$ to agent $j_2$ if $j_2$'s bundle contains some fraction of an item $g$ where $v_{j_1}(g)=p$.
We build a directed edge from an agent $j$ to the pool $F$ if $F$ contains a (possibly fractional) item $g$ where $v_j(g)=p$.
Notice that, to guarantee maximum Nash welfare, if $j_1$ receives more value than $j_2$, there should not be an edge from $j_2$ to $j_1$.
In each iteration of the procedure, we identify a set $S_{\min}$ of agents such that 1) there is a path from each agent in $S_{\min}$ to $F$ and 2) agents in $S_{\min}$ currently have the equally minimum value for their bundles among those agents satisfying 1).
We build a spanning tree rooted at $|F|$ where the tree nodes are those agents in $S_{\min}$ (note that there cannot exist ``intermediate node'' in the tree that is not in $S_{\min}$).
Then, each agent in the tree takes a portion of an item with value $p$ from her parent.
We let all agents in $S_{\min}$ simultaneously ``eat'' the item from their parents in a continuous way while keeping the utilities for agents in $S_{\min}$ the same.
This stops when one of the following two critical events happen: 1) the graph structure changes, and 2) more agents are included in $S_{\min}$, i.e., the utility for the agents in $S_{\min}$ is increased to an amount that begins to equal to the utility of some other agents that is not in $S_{\min}$.
When critical events happen, we move on to the next iteration and do the same, until $F$ becomes empty.
It is easy to verify that the maximum Nash welfare property is preserved and agents' utilities never decrease throughout this procedure.

By describing the procedure in this way, it is then easy to see that $|D_\truc|$ never decreases during the procedure.
Since all agents' utilities can only increase throughout this procedure, the only possibility for $|D_\truc|$ to decrease is when a part of $D$ in agent $j_1$'s bundle, which was initially truncated, is reallocated to an agent $j_2$ whose utility is below $L$ (if a truncated part of $D$ is reallocated to an agent $j_2$ whose utility is already above $L$, this part remains truncated in agent $j_2$'s bundle and the overall size of truncated part of $D$ is unchanged).
However, this is impossible: reallocation between two agents only happens when their utilities are the same; on the other hand, if a part of $j_1$'s bundle is truncated, the utility of $j_1$ is larger than $L$, which is larger than the utility of $j_2$.

We have described a procedure to allocate $F$ such that $|D_\truc|$ does not decrease.
Since $\calF_i^g$ optimizes $|D_\truc|$ by our definition, we have proved property 1.

After proving the two properties, we show that the length received by agent $i$ in the first phase is in the interval $[|X_i'\setminus F|,|X_i'|]$ when reporting $w_i$ by the truthfulness of the maximum Nash welfare rule under binary valuations.
Consider the allocation in the first phase.
If agent $i$ receives a length less than $|X_i'\setminus F|$ under $w_i$, assume that $w_i$ is agent $i$'s truthful valuation.
However, when agent $i$ misreports the valuation to $w_i'$, she will receive a length of $|X_i'|$, among which a length of $|X_i'\setminus F|$ has value $p$ to agent $i$, which is beneficial to misreport.
If agent $i$ receives a length more than $|X_i'|$ under $w_i$, assume that agent $i$'s truthful valuation is $w_i'$.
When agent $i$ reports the truthful valuation, she will receive a length of $|X_i'|$ with value $p$, which implies that misreporting to $w_i$ is beneficial.
Both cases contradict the truthfulness of the maximum Nash welfare rule.

Finally, we show that the value received by agent $i$ when reporting $w_i$ is no less than that when reporting $w'_i$.
For simplicity, denote the length $|X_i'|$ received by agent $i$ in the first phase when reporting $w_i'$ by $x$, and the length received by agent $i$ in the first phase when reporting $w_i$ by $\bar{x}$.
Here, we specify that $T=m-\sum_{i=1}^n|X_i''|$ and $D_\truc$ respectively denote the number of unallocated items and items with value $p$ to agent $i$ that are truncated after the second phase when agent $i$ reports $w_i'$.

After the three phases, if the valuation reported by agent $i$ is $w_i'$, the number of items agent $i$ will receive with value $p$ after the three phases is $x-|F|+\frac{|D_\truc|}{T}(L-x)$.
If agent $i$ reports $w_i$ instead, as $|D_\truc|$ will increase and $T$ will increase by at most $|F|$ by the two properties above, the number of items agent $i$ will receive with value $p$ after the three phases is lower bounded by $\bar{x}+\frac{|D_\truc|}{T+|F|}(L-\bar{x})$.
As agent $i$ will receive a length of $L$ after the three phases under both valuations, to show that reporting $w_i$ is at least as good as reporting $w_i'$, we only need to guarantee the inequality $$\left(\bar{x}+\frac{|D_\truc|}{T+|F|}(L-\bar{x})\right)-\left(x-|F|+\frac{|D_\truc|}{T}(L-x)\right)\ge 0$$ holds for $\bar{x}\in[x-|F|,x]$.
Notice that the inequality is linear in $\bar{x}$, hence it holds for $\bar{x}\in[x-|F|,x]$ as long as it holds for $\bar{x}=x-|F|$ and $\bar{x}=x$.
When $\bar{x}=x-|F|$, the inequality is simplified to $T-(L-x)\ge 0$, which holds trivially by the definition of $T$.
When $\bar{x}=x$, it is simplified to show $T^2-|D_\truc|(L-x)+T|F|\ge 0$, which holds as $T\ge |D_\truc|$ and $T-(L-x)\ge 0$.

\paragraph{Step 2.}
From now on, we will assume that agent $i$ will not misreport from value $q$ to $p$, that is, there is no item $g\in M$ such that $v_i(g)=q$ and $u_i(g)=p$. 

We begin by defining some notations.
For a set of agents $N'\subseteq N$, we denote $C(N')$ under $v_i$ by $C_v(N')$ and under $u_i$ by $C_u(N')$.
Similarly, let $((X_1')_{u},\ldots,(X_n')_{u})$ and $((X_1')_{v},\ldots,(X_n')_{v})$ be the allocations right after the first phase for valuations $u_i$ and $v_i$ respectively, and
let $((X_1'')_{u},\ldots,(X_n'')_{u})$ and $((X_1'')_{v},\ldots,(X_n'')_{v})$ be the allocations right after the second phase for valuations $u_i$ and $v_i$ respectively.
Let $T_u=m-\sum_{j=1}^n|(X_j'')_u|$ and $T_v=m-\sum_{j=1}^n|(X_j'')_v|$ be the respective sizes of the truncated items at phase two.

Let $S$ be the first group of agents found by the mechanism under the truthful valuation $v_i$, and let $S'$ be that under $u_i$.
We claim that to be profitable for agent $i$, it holds that $S=S'$ and $C_v(S)=C_u(S')$.
We prove it by contradiction.
Assume that agent $i$ receives $x$ units of items with value $p$ under $v_i$.
Compared to reporting truthfully, there are three types of deviations such that $S\neq S'$ or $C_v(S)\neq C_u(S')$, and we demonstrate that under each of the following cases, agent $i$ cannot receive more than $x$ units of the items with value $p$.
\begin{itemize}
    \item Case 1: $i\in S$, yet $i\notin S'$.
    \item Case 2: $i\in S$ and $i\in S'$, yet $S\neq S'$ or $C_v(S)\neq C_u(S')$.
    \item Case 3: $i\notin S$, yet $i\in S'$.
\end{itemize}

In the following, we use $x$ to denote the units of items that agent $i$ will receive under truthful report $v_i$ in the first phase.

In Case 1, we have $C_v(S') = C_u(S')$ and $|C_v(S')|/|S'|=|C_u(S')|/|S'|$, as $i\notin S'$ and any deviation by agent $i$ will not impact the valuation of agents in $S'$.
Consequently, $|C_u(S')|/|S'|\ge |C_v(S)|/|S|=x$, for otherwise, the mechanism will choose $S'$ instead of $S$ initially under $v_i$.
Given the fact that every agent in the latter group will not receive fewer items than agents in the first group, each agent in $S\setminus \{i\}$ will receive at least $x$ units of items in the first phase under $u_i$.
As agent $i$ only values $p$ to a subset of items within set $C_v(S)$, those items left for agent $i$ will be at most $C_v(S)-x|S\setminus\{i\}|=x$, which is not profitable.

In Case 2, denote the number of items that $v_i(g)=p$, $u_i(g)=q$, and $v_j(g)=q$ for $j\in S\setminus\{i\}$ by $\delta$ where $0\le \delta\le x$, and the units of items that agent $i$ receives in the first phase after deviation by $x'$.
These $\delta$ units of items with value $p$ to agent $i$ will not be allocated within $S$, so we have $|C_u(S)|=x|S|-\delta$.
Further, as each (possibly fractional) remaining item will be uniformly allocated, agent $i$ will receive no more than
$\frac{\delta}{T_u}(L-x')$ units of items with value $p$ in the third phase.
Since 
$$T_u\geq\sum_{j\in S}(L-|(X_j'')_u|)=L\cdot|S|-\sum_{j\in S}|(X_j'')_u|\geq L\cdot|S|-|C_u(S)|=(L-x)|S|+\delta,$$
agent $i$ will receive no more than $\frac{\delta}{(L-x)|S|+\delta}(L-x')$ units of items with value $p$ in the third phase where $L=\frac{m}{n}>x$.
Therefore, we only need to prove 
\begin{equation}\label{eqn:eqn1}
x'+\frac{\delta}{(L-x)|S|+\delta}(L-x')\le x,
\end{equation}
which is equivalent to
\begin{equation}\label{eqn:eqn2}
x'((L-x)|S|+\delta)+\delta(L-x')\leq x((L-x)|S|+\delta).
\end{equation}
The above inequality (\ref{eqn:eqn2}) is linear in $L$ and $L\geq x$. Hence, if it holds when $L=x$ and $L\rightarrow+\infty$ respectively, it holds for all values of $L$.
When $L=x$, the inequality is equivalent to $x'+(x-x')\le x$, which trivially holds.
When $L\rightarrow \infty$, the inequality (\ref{eqn:eqn1}) is equivalent to $x'\le x-\frac{\delta}{|S|}$.
By contradiction, if $x'>x-\frac{\delta}{|S|}$, for set $S\setminus\{i\}$, we have 
$$\frac{|C_u(S\setminus\{i\})|}{|S\setminus\{i\}|}=\frac{|C_u(S)|-x'}{|S|-1}<\frac{x|S|-\delta-(x-\delta /|S|)}{|S|-1}=x-\frac{\delta}{|S|}<x'.$$
It is implied that at least the set $S\setminus \{i\}$ should be chosen by the procedure before agent $i$ under $u_i$, contradicting to the assumption that $i\in S'$.
Hence, inequality (\ref{eqn:eqn1}) holds.

The analysis for Case 3 is similar to Case 2.
As $i\notin S$, assume $i\in S_k$ under truthful report (which denotes the group found by the mechanism at the $k$-th round), and let $S_{<k}=S_1\cup\cdots\cup S_{k-1}$ denote the set of agents who receive items before $S_k$ in the first phase.
Each agent $j\in S_{<k}$ receives $x_j$ numbers of items in the first phase under $v_i$ where $x_j\le x$, and $|C_v(S_{<k})|/|S_{<k}|\le x$.
We still define $\delta$ as the number of items that $v_i(g)=p$, $u_i(g)=q$, and $v_j(g)=q$ for $j\in S_{<k}\cup S_k\setminus\{i\}$, and $x'$ be the numbers of items that agent $i$ receives in the first phase under $u_i$.
Thus, agent $i$ will receive no more than $\frac{\delta}{T_u}(L-x')$ items with value $p$ in the third phase.
Since
\begin{align*}
    T_u&\geq\sum_{j\in S_{<k}\cup S_k}(L-|(X_j'')_u|)\geq\sum_{j\in S_{<k}\cup S_k}(L-|(X_j')_u|)\\
    &\geq L\cdot|S_{<k}\cup S_k|-C_u(S_{<k}\cup S_k)\tag{as $(X_j')_u$ can only contain $g$ with value $p$ to agent $j$}\\
    &= L\cdot|S_{<k}\cup S_k|-\left(C_v(S_{<k}\cup S_k)-\delta\right)\tag{Note that $v_j(g)=q$ for $j\in S_{<k}\cup S_k\setminus\{i\}$ by definition of $\delta$}\\
    &=\sum_{j\in S_{<k}}(L-x_j)+(L-x)|S_k|+\delta,
\end{align*}
agent $i$ will receive no more than $\frac{\delta}{\sum_{j\in S_{<k}}(L-x_j)+(L-x)|S_k|+\delta}(L-x')$ items with value $p$ in the third phase.
Hence, the problem is reduced to the validity of the inequality 
\begin{equation}\label{eqn:eqn3}
x'+\frac{\delta}{\sum_{j\in S_{<k}}(L-x_j)+(L-x)|S_k|+\delta}(L-x') \le x,
\end{equation}
and we similarly show it holds for $L=x$ and $L\rightarrow+\infty$.
When $L\rightarrow +\infty$, the inequality is transformed into $x'\le x-\frac{\delta}{|S_{<k}\cup S_k|}$. 
If it does not hold, we have the contradiction that
$$\frac{|C_u(S_{<k}\cup S_k\setminus\{i\})|}{|S_{<k}\cup S_k\setminus\{i\}|}=\frac{|C_u(S_{<k}\cup S_k)|-x'}{|S_{<k}\cup S_k|-1}<\frac{x|S_{<k}\cup S_k|-\delta-(x-\delta/|S_{<k}\cup S_k|)}{|S_{<k}\cup S_k|-1}=x-\frac{\delta}{|S_{<k}\cup S_k|}<x',$$
indicating that $i\notin S'$, which is a contradiction.
When $L=x$, (\ref{eqn:eqn3}) is equivalent to
$$\frac{\delta}{\sum_{j\in S_{<k}}(L-x_j)+(L-x)|S_k|+\delta}(L-x') \leq L-x',$$
which always holds as $\frac{\delta}{\sum_{j\in S_{<k}}(L-x_j)+(L-x)|S_k|+\delta}\leq 1$.

By far, we have shown that $S=S'$ and $C_u(S)=C_v(S)$, so the agents in $S$ will receive items in the same way under $u_i$ and $v_i$ in the first phase.
If agent $i$ belongs to $S$, it can be seen from $C_u(S)=C_v(S)$ that no item with value $p$ to agent $i$ is left to the second phase, so misreporting is not profitable.
If agent $i$ does not belong to $S$, we remove the set of agents $S$ from $N$ and the set of items $C(S)$ from $M$, and consider the next step of the mechanisms under $v_i$ and $u_i$.
By adopting the same analysis, it can be inductively shown that if agent $i$ aims to gain higher utility, the mechanism should behave entirely identically under $v_i$ and $u_i$, which implies that profitable misreporting of $\calF_i^g$ does not exist.

\end{document}